\documentclass[]{elsarticle}
\usepackage[left=3cm,right=3cm,top=3cm,bottom=3cm]{geometry}
\usepackage{lineno,hyperref}
\modulolinenumbers[5]

\usepackage[colorinlistoftodos]{todonotes}

\usepackage{amssymb,amsmath,amsfonts}
\usepackage{mathtools}
\usepackage{tikz} 
\usepackage[framemethod=tikz]{mdframed} 
\usetikzlibrary{cd, arrows, arrows.meta, decorations.markings, shapes, shapes.geometric, shadows, calc, fit, intersections}
\usepackage{bbm}      
\usepackage{xfrac}    
\usepackage{bm}
\usepackage{faktor}
\usepackage{enumerate}

\newtheorem{lemma}{Lemma}
\newtheorem{theorem}{Theorem}
\newtheorem{corollary}{Corollary}
\newtheorem{definition}{Definition}
\newtheorem{proposition}{Proposition}
\newtheorem{remark}{Remark}
\newtheorem{example}{Example}
\newenvironment{proof}{\noindent{\bf Proof.}}{\qed}

\newcommand{\LL}{\mathbb{L}}

\newcommand{\GG}{\mathbb{G}}
\newcommand{\UU}{\mathbb{U}}
\newcommand{\EE}{\mathbb{E}}
\newcommand{\HH}{\mathbb{H}}
\newcommand{\KK}{\mathbb{K}}
\newcommand{\NN}{\mathbb{Z}_{\geq0}}
\newcommand{\ZZ}{\mathbb{Z}}
\newcommand{\FF}{\mathbb{F}}
\newcommand{\QQ}{\mathbb{Q}}
\newcommand{\DD}{\mathbb{D}}
\newcommand{\VV}{\mathbb{V}}
\renewcommand{\AA}{\mathbb{A}}

\newcommand{\rpE}{$\mathrm{RP}$}

\newcommand{\apiE}{$\mathrm{A\Pi}$}
\newcommand{\pisiSE}{$\mathrm{\Pi\Sigma}$}
\newcommand{\pisiE}{$\mathrm{\Pi\Sigma}$}
\newcommand{\rpisiSE}{$\mathrm{R}\mathrm{\Pi\Sigma}$}
\newcommand{\aE}{$\mathrm{A}$}
\newcommand{\rE}{$\mathrm{R}$}
\newcommand{\apE}{$\mathrm{AP}$}
\newcommand{\rpiE}{$\mathrm{R\Pi}$}

\newcommand{\pE}{$\mathrm{P}$}
\newcommand{\piE}{$\mathrm{\Pi}$}
\newcommand{\sigmaE}{$\mathrm{\Sigma}$}
\newcommand{\rpisiE}{$\mathrm{R}\mathrm{\Pi\Sigma}$}

\newcommand{\geno}[1]{\left\langle {#1} \right\rangle}
\newcommand{\mySum}[4]{\smashoperator{\sum_{\mathclap{\substack{{#1}={#2}}}}^{{#3}}}{#4}} 
\newcommand{\ProdLst}[4]{ {{#1}^{{#4}_{#2}}_{#2}}\cdots{{#1}^{{#4}_{#3}}_{#3}} } 
\newcommand{\prodLst}[5]{ \prod_{{#1}={#2}}^{{#3}} {{#4}_{#1}}^{{#5}_{#1}} } 
\newcommand{\fourargsexpsubscript}[4]{ {#1}^{{#2}_{{#3},{#4}}}_{{#3},{#4}} }

\newcommand{\threeargssubscript}[3]{ {#1}_{{#2},{#3}} }

\newcommand{\ord}{\mathrm{ord}}
\newcommand{\const}{\mathrm{const}} 
\newcommand{\zv}[1]{  {\bf 0}_{#1} }
\newcommand{\zs}{\{ 0 \}}

\newcommand{\myd}{\mathfrak{d}}
\newcommand{\genn}[1]{\langle{#1}\rangle}

\newcommand{\ii}{\mathbbm{i}} 
\newcommand{\ee}{\mathbbm{e}} 
\newcommand{\sm}{\setminus}
\newcommand{\s}{\sigma}
\newcommand{\bs}{\boldsymbol}

\newcommand{\zvs}[1]{\{ {\bf 0}_{#1} \}} 
\newcommand{\dField}[2][\sigma]{( {#2}, {#1})}

\newcommand{\ringOfEquivSeqs}[1][\mathbb{K}]{\mathcal{S}({#1})}

\newcommand{\seqA}[4][0]{\big\langle {#2}_{#3}{(#4)} \big\rangle_{{#4} \geq {#1}} }
 
\newcommand{\funcSeqA}[3][0]{\big\langle {#2} \big\rangle_{{#3} \geq {#1}}}

\newcommand{\algSeq}[2]{\big(\sqrt{#1}\big)^{#2}}

\newcommand{\intSeq}[2]{{#1}^{#2}} 
\newcommand{\intSeqExp}[3]{\big({#1}^{#2}\big)^{#3}} 

\newcommand{\myProduct}[4]{\smashoperator{\prod_{\mathclap{\substack{{#1}={#2}}}}^{{#3}}} {#4}} 

\newcommand{\myProd}[3]{\smashoperator{\prod_{\mathclap{\substack{{#1}={#2}}}}^{{#3}}}} 
\newcommand{\Lst}[3]{{#1}_{#2},\dots,{#1}_{#3}}

\newcommand{\id}{\mathrm{id}}

\DeclareMathOperator{\lcm}{lcm}

\DeclareMathOperator{\ev}{ev} 

\DeclareMathOperator{\ProdExpr}{ProdE_n}
\DeclareMathOperator{\ProdMon}{ProdM_n}
\DeclareMathOperator{\Prod}{Prod_n}

\DeclareMathOperator{\per}{per}

\newcommand{\nn}{\nonumber} 
\newcommand{\cali}[1]{\mathcal{#1}}

\newcommand{\dr}{difference ring} 
\newcommand{\df}{difference field}
\newcommand{\drE}{difference ring extension}

\newcommand{\rc}{\textcolor{red}} 
\definecolor{blaugrau}{rgb}{0.796887, 0.789075, 0.871107}
\newcounter{mmacnt}
\def\restartmma{\setcounter{mmacnt}{0}}
\restartmma \catcode`|=\active
\def|#1|{\mathrm{#1}}
\catcode`|=12
\newenvironment{mma}{
	\par
	\catcode`|=\active
	\parskip=6pt\parindent=0pt 
	\small
	\def\In##1\\{%
		\def\linebreak{\hfill\break\null\qquad}%
		\refstepcounter{mmacnt}
		\hangindent=2.5em\hangafter=0
		\leavevmode
		\llap{\tiny\sffamily In[\arabic{mmacnt}]:=\kern.5em}%
		\mathversion{bold}\scriptsize$\tt\bf\displaystyle##1$\normalsize
		\mathversion{normal}\par
	}%
	\def\Print##1\\{%
		\def\linebreak{\hfill\break}%
		\hangindent=2.5em\hangafter=0
		\leavevmode\scriptsize ##1\par}%
	\def\Out##1\\{%
		\vspace*{-0.5cm}\def\linebreak{$\hfill\break\null\hfill$}%
		\kern\abovedisplayskip\par
		\hangindent=2.5em\hangafter=0
		\leavevmode
		\llap{\tiny\sffamily Out[\arabic{mmacnt}]=\kern.5em}
		\scriptsize$\displaystyle\tt##1$\normalsize\hfill\null\par
		\kern\belowdisplayskip\vspace*{-0.3cm}
	}%
	\def\Warning##1##2\\{%
		\def\linebreak{\hfill\break}%
		\hangindent=2.5em\hangafter=0
		\leavevmode
		{\scriptsize##1 : ##2}\par}%
}{%
	\par\smallskip
}

\newcommand{\LoadP}[1]{\fcolorbox{black}{blaugrau}{
\begin{minipage}[t]{12.5cm}
			\scriptsize\normalfont #1
\end{minipage}}}

\newcommand{\myIn}[1]{{\sffamily In[#1]}}

\def\MLabel#1{{\refstepcounter{mmacnt}\label{#1}}\addtocounter{mmacnt}{-1}}
\newcommand{\MText}[1]{\textbf{\ttfamily#1}}

\allowdisplaybreaks[4]

\makeatletter
\def\ProblemSpecBox{
    \@ifnextchar[\ProblemSpecBox@opt{\ProblemSpecBox@noopt}}
\def\ProblemSpecBox@opt[#1]#2{
    \protected@edef\@currentlabelname{#1}
    \protected@edef\@currentlabel{#1}
    \begin{mdframed}[
        innerlinewidth=0.5pt,
        innerleftmargin=5.5pt,
        innerrightmargin=5.5pt,
        innertopmargin = 10pt,
        innerbottommargin=10pt,
        skipabove=\dimexpr\topsep+\ht\strutbox\relax,
        roundcorner=5pt,
        frametitle={#2},
        frametitlerule=true,
        frametitlerulewidth=1pt]
    }
    \def\ProblemSpecBox@noopt#1{
        \ProblemSpecBox@opt[#1]{#1}
    }
    \def\endProblemSpecBox{
    \end{mdframed}
}
\makeatother

\makeatletter
\newcommand{\manuallabel}[2]{\def\@currentlabel{#2}\label{#1}}
\makeatother

\newcommand{\op}[1]{{#1}roblem {O}} 
\newcommand{\gop}[1]{{#1}roblem {GO}} 
\newcommand{\rpe}[1]{{#1}roblem {RPE}}

\pgfdeclarehorizontalshading{mmawithcounterbackground}{100bp}
{color(0bp)=(green!40); color(100bp)=(black!5)}
\pgfdeclarehorizontalshading{mmawithcountertitle}{100bp}
{color(0bp)=(red!40);color(100bp)=(black!5)}
\newcounter{mmawithcounter}
\renewcommand*\themmawithcounter{\bf{Mathematica Session}~\arabic{mmawithcounter}}
\makeatletter
\def\mdf@@mmawithcounterpoints{}
\define@key{mdf}{mmawithcounterpoints}{%
    \def\mdf@@mmawithcounterpoints{#1}
}

\tikzset{titregris/.style =
    {draw=gray, thick, fill=white, shading = mmawithcountertitle, %
        text=gray, rectangle, rounded corners, right,minimum height=.7cm}}

\mdfdefinestyle{mmawithcounterstyle}{%
    outerlinewidth=1em,outerlinecolor=white,%
    leftmargin=-1em,rightmargin=-1em,%
    middlelinewidth=1.2pt,roundcorner=5pt,linecolor=gray,
    apptotikzsetting={\tikzset{mdfbackground/.append style ={%
                shading = mmawithcounterbackground}}},
    innertopmargin=1.2\baselineskip,
    skipabove={\dimexpr0.5\baselineskip+\topskip\relax},
    skipbelow={-1em},
    needspace=3\baselineskip,
    frametitlefont=\sffamily\bfseries,
    settings={\global\stepcounter{mmawithcounter}},
    singleextra={
        \node[titregris,xshift=6.15cm] at (P-|O) %
        {~\mdf@frametitlefont{\themmawithcounter}\hbox{~}};
        \ifdefempty{\mdf@@mmawithcounterpoints}%
        {}%
        {\node[titregris,left,xshift=-0.5cm] at (P)%
            {~\mdf@frametitlefont{\mdf@@mmawithcounterpoints points}\hbox{~}};}%
    },
    firstextra={
        \node[titregris,xshift=6.15cm] at (P-|O) %
        {~\mdf@frametitlefont{\themmawithcounter}\hbox{~}};
        \ifdefempty{\mdf@@mmawithcounterpoints}%
        {}%
        {\node[titregris,left,xshift=-0.5cm] at (P)%
            {~\mdf@frametitlefont{\mdf@@mmawithcounterpoints points}\hbox{~}};}%
    },
}

\definecolor{blaugrau}{rgb}{0.796887, 0.789075, 0.871107}

\begin{document}

\begin{frontmatter}

\title{{\footnotesize\vspace*{-1.3cm}
	RISC-Linz Report Series No. 20-19}\\[0.2cm]
	Representation of hypergeometric products of higher nesting depths in difference rings}
\tnotetext[t1]{This work was supported by the Austrian Science Fund (FWF) grant P32301 and by the Austrian Science Fund (FWF) grant F5009-N15 in the framework of the SFB ``Algorithmic and Enumerative Combinatorics''.}
	
\author[1,2]{Evans Doe Ocansey}
\ead{evans.ocansey@jku.at}
\author[2]{Carsten Schneider}
\ead{cschneider@risc.jku.at}

\address[1]{Johannes Kepler University Linz, Institute for Algebra, Linz, Austria}
\address[2]{Johannes Kepler University Linz, Research Institute for Symbolic Computation (RISC), Linz, Austria}


\begin{abstract}
A non-trivial symbolic machinery is presented that can rephrase algorithmically a finite set of nested hypergeometric products in appropriately designed difference rings. As a consequence, one obtains an alternative representation in terms of one single product defined over a root of unity and nested hypergeometric products which are algebraically independent among each other. In particular, one can solve the zero-recognition problem: the input expression of nested hypergeometric products evaluates to zero if and only if the output expression is the zero expression.
Combined with available symbolic summation algorithms in the setting of difference rings, one obtains a general machinery that can represent (and simplify) nested sums defined over nested products.
\end{abstract}


\end{frontmatter}

\section{Introduction}\label{Sec:Introduction}

An important problem in symbolic summation is the simplification of sums defined over products to expressions in terms of simpler sums and products; in the best case, one might find an expression without sums. 
A first milestone was Gosper's algorithm~\cite{Gosper:78} and Zeilberger's groundbreaking application for creative telescoping~\cite{Zeilberger:91} where the summand is given by one hypergeometric product. Further extensions have been accomplished for $q$-hypergeometric and  multibasic products~\cite{PauleRiese:97}, their mixed versions~\cite{bauer1999multibasic} and ($q$--)multi-summation~\cite{Wilf:92,Wegschaider,Riese:03}.
In addition, structural properties and further insight in this setting have been elaborated, e.g., in~\cite{paule1995greatest,chen2011structure,CJKS:13}. More generally, the holonomic system approach~\cite{Zeilberger:90a} and their refinements~\cite{CHYZAK:00,Koutschan:13,BRS:18} represent, e.g., multi-sums over \hbox{($q$--)}hypergeometric products by systems of linear recurrence relations. 

In particular, Karr's difference field approach~\cite{karr1981summation,Karr:85} paved the way for a general framework to represent rather complicated product expressions in a formal way. Here the generators of his $\Pi\Sigma$-field construction enables one to model (up to a certain level)
indefinite nested products of the form
\begin{equation}\label{eqn:nestedHypergeometricProduct}
P(n) =\myProduct{k_{1}}{\ell_{1}}{n}{f_{1}(k_{1})}\cdots\myProduct{k_{m}}{\ell_{m}}{k_{m-1}}{f_{m}(k_{m})}
\end{equation} 
where the multiplicands $f_i(n)=\frac{p_i(n)}{q_i(n)}$ for all $i$ with $1\le i \le m$ are built by polynomial expressions $p_i(n)$ and $q_i(n)$ in terms of indefinite nested products that are again of the form~\eqref{eqn:nestedHypergeometricProduct}. In Karr's seminal works~\cite{karr1981summation,Karr:85}, which can be considered as the discrete version of Risch's integration algorithm~\cite{Risch:69}, a sophisticated algorithm is provided that enables one to test if a given product representation (and sums defined over such products) are expressed properly in his $\Pi\Sigma$-field setting. As a bonus, his toolbox and refinements in~\cite{schneider2015fast,schneider2007simplifying,schneider2008refined} enable one to decide constructively if a given summand represented in a $\Pi\Sigma$-field has a solution in the same field or in an appropriate extension of it. In addition, first contributions have been provided in~\cite{schneider2005product,abramov2010polynomial} to simplify products further such that the degrees in the numerators and denominators are minimal. 

For such complicated classes of products the following task is non-trivial: given an arbitrary expression in terms of nested products~\eqref{eqn:nestedHypergeometricProduct}, design algorithmically an appropriate difference field or ring in which the expression can be represented and in which one can solve, e.g., the (creative) telescoping problem. Already for a hypergeometric product $P_1(n)=\prod_{k=\ell}^n f(k)$ with a rational function $f(x)\in\KK(x)^*$ and $\ell\in\NN$ chosen properly (i.e., $P_1(n)$ is well defined and nonzero for all $n\in\NN$), it has been shown in~\cite{schneider2005product} that Karr's $\Pi\Sigma$-fields are not sufficient: namely, such a product $P_1(n)$ can be represented in a $\Pi\Sigma$-field if and only if it cannot be rewritten in the form $P_1(n)=\zeta^n\,r(n)$ where $r(x)\in\KK(x)$ and $\zeta\in\KK$ is a primitive $\lambda$-th root of unity with $\lambda>1$. More precisely, such objects can be only represented in ring extensions taking care of the relation $(\zeta^n)^{\lambda}=1$. As a consequence, zero-divisors are introduced coming from
$(1+\zeta^n+\dots+(\zeta^{n})^{\lambda-1})(1-\zeta^n)=0.$

Motivated by this observation (among others) a refined difference ring theory has been elaborated in~\cite{schneider2016difference,schneider2017summation} that combines big parts of Karr's general framework together with generators of the form $\zeta^n$. In this way, not only the hypergeometric product $P_1(n)$, but more generally any polynomial expression in terms of hypergeometric products can be represented in the class of so-called \rpisiE-extensions. The first algorithms derived in~\cite{schneider2005product,Schneider:14} require that the  input products are defined over $\KK(x)$ where $\KK=\QQ(\kappa_1,\dots,\kappa_u)$ with $u\geq0$ is a rational function field defined over the rational numbers $\QQ$. More generally, a complete algorithm has been elaborated in~\cite{ocansey2018representing} (utilizing ideas from~\cite{ge1993testing,Schneider:14}) that can represent a finite set of hypergeometric products  over $\KK(x)$ where $\KK=K(y_1,\dots,y_r)$ is a rational function field defined over an algebraic number field $K$. In addition, using algorithms from~\cite{bauer1999multibasic} it can deal also with $q$-hypergeometric, multibasic and mixed hypergeometric products. Finally, a general framework has been elaborated in~\cite{schneider2020minimal} that considers single nested products defined over a general class of difference fields; an extra bonus is that the latter approach constructs a difference ring in which the given products are rephrased optimally with following property: the number of generators of the ring and the order $\lambda$ of the used $\zeta^n$ are minimal.\\
A remarkable feature of the above algorithms~\cite{schneider2005product,Schneider:14,ocansey2018representing,schneider2020minimal} is that the given input expression of hypergeometric products (and their generalized versions) are rephrased in terms of a finite set of alternative products $Q_1(n),\dots,Q_s(n)$ together with a distinguished root of unity product $\zeta^n$ such that the sequences produced by $Q_1(n),\dots,Q_s(n)$ are algebraically independent among each other. For this result we rely on ideas of~\cite{schneider2017summation} that are inspired by~\cite{schneider2010parameterized,hardouin2008differential}; compare also~\cite{chen2011structure}. We remark further that these results are also connected to~\cite{kauers2008computing} that can compute all algebraic relations of $C$-finite solutions (i.e., solutions of homogeneous recurrences with constant coefficients). \\
We emphasize that the algorithms from~\cite{schneider2005product,abramov2010polynomial,Schneider:14,ocansey2018representing,schneider2020minimal} can be utilized to simplify hypergeometric solutions~\cite{petkovvsek1992hypergeometric,van1999finite,ABPS:20} of linear difference equations and can be combined with symbolic summation algorithms~\cite{schneider2007simplifying,schneider2015fast,schneider2008refined} to simplify more general solutions, such as d'Alembertian solutions~\cite{abramov1994d,abramov1996d} and Liouvillian solutions~\cite{Singer:99,petkovvsek2013solving}. 

In this article we aim at extending this toolbox significantly for the general class of nested hypergeometric products that can be defined as follows.

\begin{definition}\label{defn:hyperGeometricPrdts}
    \normalfont  Let $\KK(x)$ be a rational function field\footnote{Throughout this article all fields and rings have characteristic $0$.} and let $f_{1}(x),\dots,f_{m}(x)\in\KK(x)^{*}$. Furthermore, let $\Lst{\ell}{1}{m}\in\NN$ such that for all $i$ with $1\le i \le m$, $f_{i}(j)$ is non-zero and has no pole for all $j\in\NN$ with $j\ge \ell_{i}$. Then the indefinite product expression~\eqref{eqn:nestedHypergeometricProduct} 
    is called a \emph{hypergeometric product in $n$ of nesting depth $m$}. The vector $(f_{1}(x),\dots,f_{m}(x))\in(\KK(x)^*)^m$ is also called the \emph{multiplicand representation} of $P(n)$. If $f_{i}(x)\in\KK^{*}$ for $1\le i \le m$, then we call~\eqref{eqn:nestedHypergeometricProduct}
    a \emph{constant} or \emph{geometric product in $n$ of nesting depth $m$}. 
    Further, we define the set of ground expressions with\footnote{Their elements are considered as expressions that can be evaluated for sufficiently large $n\in\NN$.} $\KK(n)=\{f(n)\mid f(x)\in\KK(x)\}$. Moreover, we define $\Prod(\GG)$ with $\GG\subseteq\KK(x)$ as the set of all such products where the multiplicand representations are taken from $\GG$. Furthermore, we introduce the set of
    product monomials 
    $\ProdMon(\GG)$ as the set of all elements
    $$a(n)P_{1}(n)^{\nu_{1}}\cdots P_{e}(n)^{\nu_{e}}$$
    with $a(x)\in\GG$, $e\in\NN$, $\nu_1,\dots,\nu_e\in\ZZ$ and $P_{1}(n),\dots,P_{e}(n)\in\Prod(\GG)$. Finally, we introduce
    the set of product expressions $\ProdExpr(\GG)$ as the set of all elements	
    \begin{equation}\label{Equ:ProdEDef}
    A(n)=\smashoperator{\sum_{\mathclap{\substack{{\bs{v}=(\Lst{\nu}{1}{e})\in S}{}}}}^{{}}}{a_{\bs{v}}(n)\,P_{1}(n)^{\nu_{1}}\cdots P_{e}(n)^{\nu_{e}}}
    \end{equation}
    with $e\in\NN$, $S \subseteq \ZZ^{e}$ finite, $a_{\bs{v}}(x)\in\GG$ for $\bs{v}\in S$ and $P_{1}(n),\dots,P_{e}(n)\in\Prod(\GG)$. Note that $\Prod(\GG)\subseteq\ProdMon(\GG)\subseteq\ProdExpr(\GG)$.
\end{definition}

\noindent Utilizing the available algorithms from~\cite{ocansey2018representing} we will obtain enhanced algorithms that can rephrase expressions from $\ProdExpr(\KK(x))$ in the setting of $R\Pi\Sigma$-extensions. As a consequence we will solve the following problem; for further details see Theorem~\ref{thm:main} and Corollary~\ref{Cor:SolutionToMainProblem} below.

\begin{ProblemSpecBox}[\rpe{P}]{
        {\bf \rpe{P}: Representation of Product Expressions.}
    }\label{prob:ProblemRPE}
    {
        Let $\KK=K(\Lst{\kappa}{1}{u})$ be a rational function field with $e\ge0$ over an algebraic number field $K$. \emph{Given} $A(n)\in\ProdExpr(\KK(x))$. \emph{Find} $B(n)\in\ProdExpr(\KK(x))$ with $\tilde{\KK}=\tilde{K}(\Lst{\kappa}{1}{u})$  where $\tilde{K}$ is an algebraic field extension of $K$, and a non-negative integer $\delta\in\NN$ with the following properties: 
        \begin{enumerate}[(1)]
            \item\manuallabel{item1:ProblemRPE}{(1)} $A(n) = B(n)$ for all $n \in \NN$ with $n \geq \delta$; 
            \item\manuallabel{item2:ProblemRPE}{(2)} All the products $P_1(n),\dots,P_s(n)\in\Prod(\KK(x))$ arising in $B(n)$ (apart from the distinguished product $\zeta^n$ with $\zeta$ a root of unity) are algebraically independent among each other.
            \item\manuallabel{item3:ProblemRPE}{(3)} The zero-recognition property holds, i.e., $A(n)=0$ holds for all $n$ from a certain point on if and only if $B(n)$ is the zero-expression. 
        \end{enumerate}
    }
\end{ProblemSpecBox} 

\noindent The full machinery have been implemented within Ocansey's Mathematica package \texttt{NestedProducts} whose functionality will be illustrated in Section~\ref{subsec:mathematicaDemo} below; for additional aspects we refer also to~\cite{ocansey2019difference}.
We expect that this implementation will open up new applications, e.g., in combinatorics, such as non-trivial evaluations of determinants~\cite{mills1983alternating,zeilberger1996proof,krattenthaler2001advanced}. 
In particular, in interaction with the symbolic summation algorithms available in the package~\texttt{Sigma}~\cite{schneider2007symbolic} one obtains a fully automatic toolbox to tackle nested sums defined over nested hypergeometric products.

\medskip

The outline of the article is as follows. In Section~\ref{sec:preProcessingStep} we will introduce rewrite rules that enable one to transform expressions from $\ProdExpr(\KK(x))$ to a more suitable form (see Proposition~\ref{pro:preprocessingNestedHypergeometricProductsExtended} below) to solve \ref{prob:ProblemRPE}. Given this tailored form, we show in Section~\ref{sec:apExts} how such expressions can be rephrased straightforwardly in terms of multiple-chain \apE-extensions. In order to solve \ref{prob:ProblemRPE}, we have to refine this difference ring construction. Namely, in Section~\ref{sec:refinedApproachRPiExtns} we introduce \rpiE-extensions: these are \apE-extensions where during the construction the set of constants remain unchanged. In particular, we will elaborate that such rings can be straightforwardly embedded into the ring of sequences and will provide structural theorems that will prepare the ground to solve \ref{prob:ProblemRPE}. With these results we will present in Section~\ref{Sec:ProductsInRPiExt} the main steps how nested products can be represented in \rpiE-extensions. In Section~\ref{Sec:CompleteAlg} we will combine all these ideas yielding a complete algorithm for \ref{prob:ProblemRPE} that is summarized in Theorem~\ref{thm:main} and Corollary~\ref{Cor:SolutionToMainProblem}. In addition, we will illustrate with non-trivial examples how one can solve \ref{prob:ProblemRPE} with the new Mathematica package \texttt{NestedProducts}. The conclusions are given in Section~\ref{Sec:Conclusion}.

\section{Preprocessing hypergeometric products of finite nesting depth}\label{sec:preProcessingStep}
In order to support our machinery to solve~\ref{prob:ProblemRPE}, the arising products $P(n)$ in $A(n)\in\ProdExpr(\KK(x))$ (e.g., given in~\eqref{Equ:ProdEDef}) will be transformed to a particularly nice form. We will illustrate each preprocessing step with an example and then summarize the derived result in Proposition~\ref{pro:preprocessingNestedHypergeometricProductsExtended} below. 

Let $\KK(x)$ be a rational function field together with the \emph{zero-function} (in short \emph{$Z$-function}) defined by
\begin{equation}\label{eqn:hyperGeoShiftBoundedFxns}
    Z(p) = \max\big(\{k \in \NN\,|\,p(k)=0\}\big)+1 \ \text{ for any } p \in \KK[x]
\end{equation}
with $\max(\varnothing)=-1$. We call $\KK$ \emph{computable} if all basic field operations are computable. Note that if $\KK$ is a rational function field over an algebraic number field, then $\KK$ and also its $Z$-function are computable. 

We start with the hypergeometric product in $n$ of nesting depth $m\in\NN$ given by 
\begin{equation}\label{eqn:arbitraryNestedHyperGeoPrdt}
    P(n) = \myProduct{k_{1}}{\ell_{1}}{n}{f_{1}(k_{1})\myProduct{k_{2}}{\ell_{2}}{k_{1}}{f_{2}(k_{2})}\cdots\myProduct{k_{m}}{\ell_{m}}{k_{m-1}}{f_{m}(k_{m})}}\in\Prod(\KK(x))
\end{equation}
where $f_{i}(x)\in\KK(x)^{*}$ and $\ell_{i}\in\NN$ for all $1\le i \le m$. Note that by definition $P(n)\neq0$ for all $n\in\NN$. 
In particular, no poles arise for any evaluation at $n\in\NN$. We remark that the $Z$-function can be used to specify the lower bounds $\ell_i\in\NN$ such that this property holds. Then $P(n)$ in $\Prod(\KK(x))$ is preprocessed as follows. 

\subsection{Transformation of indefinite products to product factored form}\label{subsec:transformationToProductFactoredForm}

The first transformation is based on the following simple observation.

\begin{proposition}\label{pro:productFactoredForm}
   For $P(n)$ given in~\eqref{eqn:arbitraryNestedHyperGeoPrdt} with multiplicands $\Lst{f}{1}{m}\in\KK(x)^{*}$ we have 
       \begin{equation}\label{eqn:productFactoredForm}
       P(n) = \Bigg(\,\myProduct{k_{1}}{\ell_{1}}{n}{f_{1}(k_{1})}\Bigg)\Bigg(\,\myProd{k_{1}}{\ell_{1}}{n}\;\;\;\myProduct{k_{2}}{\ell_{2}}{k_{1}}{f_{2}(k_{2})}\Bigg)\cdots\Bigg(\,\myProd{k_{1}}{\ell_{1}}{n}\;\;\,\myProd{k_{2}}{\ell_{2}}{k_{1}}\cdots\myProduct{k_{m}}{\ell_{m}}{k_{m-1}}{f_{m}(k_{m})}\Bigg)\in\ProdMon(\KK(x)).
       \end{equation}
\end{proposition}

\begin{definition}
\normalfont
The right hand side of~\eqref{eqn:productFactoredForm} is also called a \emph{product factored form} of $P(n)$.
Moreover, a product of the form
$$P'(n)=\myProd{k_{1}}{\ell_{1}}{n}\;\;\,\myProd{k_{2}}{\ell_{2}}{k_{1}}\cdots\myProduct{k_{m}}{\ell_{m}}{k_{m-1}}{p_{m}(k_{m})}$$
is also called a \emph{product in factored form}. In particular, we also call $p_m(x)\in\KK(x)^*$ (instead of $(1,\dots,1,p_m(x))$) the \emph{multiplicand representation} of $P'(n)$. 
\end{definition}

Further, for $1 \le i \le m$ write 
\begin{equation}\label{eqn:factorAandB}
    \begin{aligned}
        f_{i} &= u_{i}\fourargsexpsubscript{f}{e}{i}{1}\cdots\fourargsexpsubscript{f}{e}{i}{r_{i}}\in\KK(x)
    \end{aligned}
\end{equation}
in its complete factorization. This means that $f_i$ can be decomposed by $u_{i}\in\KK^{*}$ and irreducible monic polynomials 
$f_{i,j}\in\KK[x]\sm\KK$ with $e_{i,j}\in\ZZ$ for some $1 \le j \le r_{i}$ with $r_i\in\NN$.
Substituting~\eqref{eqn:factorAandB} into the right-hand side of~\eqref{eqn:productFactoredForm} and expanding the product quantifiers over each factor in~\eqref{eqn:factorAandB} we get 
\begin{equation*}
    P(n) = A_{1}(n)\,A_{2}(n)\cdots A_{m}(n)\in\ProdMon(\KK(x))
\end{equation*}
where 
\begin{align}\label{eqn:nestedProductsOverNumerator}
    A_{i}(n)&=\Bigg(\,\myProduct{k_{1}}{l_{1}}{n}{\cdots \myProduct{k_{i}}{\ell_{i}}{k_{i-1}}{u_{i}}}\Bigg)\,\Bigg(\,\myProduct{k_{1}}{\ell_{1}}{n}{\cdots \myProduct{k_{i}}{\ell_{i}}{k_{i-1}}{\threeargssubscript{f}{i}{1}(k_{i})}}\Bigg)^{\threeargssubscript{e}{i}{1}}\cdots\Bigg(\,\myProduct{k_{1}}{\ell_{1}}{n}{\cdots \myProduct{k_{i}}{\ell_{i}}{k_{i-1}}{\threeargssubscript{f}{i}{r_{i}}(k_{i})}}\Bigg)^{\threeargssubscript{e}{i}{r_{i}}}
\end{align}
for all $1\le i\le m$. In particular, the first product on the right hand side in~\eqref{eqn:nestedProductsOverNumerator} with innermost multiplicand $u_{i}\in\KK^{*}$ is a geometric product of nesting depth $i$ in $\Prod(\KK)$, while the rest are nesting depth $i$ hypergeometric products in $\ProdMon(\KK(x))$ which are not geometric products.
\begin{example}\label{exa:nestedHyperGeoPrdtPreprocessingStep1}
    \normalfont Let $\KK=\QQ(\sqrt{3})$ and $\KK(x)$ be the rational function field over $\KK$ with the $Z$-function~\eqref{eqn:hyperGeoShiftBoundedFxns}. Suppose we are given the nesting depth $2$ hypergeometric product
    \begin{equation}\label{eqn:depth2HyperGeoPrdt}
    P(n) = \myProduct{k}{1}{n}{\frac{24\,k+1}{-\sqrt{3}}\,\myProduct{j}{3}{k}{\dfrac{-2\,(j^{3}-3\,j+2)}{5\,(j^{2}-j-2)}}}\in\Prod(\KK(x)).
    \end{equation}
    Then with 
        \begin{align}
        A_{1}(n) &= \left(\myProduct{k}{1}{n}{-1}\right) \left(\myProduct{k}{1}{n}{\sqrt{3}}\right)^{-1}\left(\myProduct{k}{1}{n}{24}\right)\left(\myProduct{k}{1}{n}{\left(k+\tfrac{1}{24}\right)}\right)\label{eqn:depth1ProductFactors}\\[5pt]
        \begin{split}\label{eqn:depth2ProductFactors}
            A_{2}(n) &= \left(\myProduct{k}{1}{n}{\myProduct{j}{3}{k}{-1}}\right)\left(\myProduct{k}{1}{n}{\myProduct{j}{3}{k}{5}}\right)^{-1} \left(\myProduct{k}{1}{n}{\myProduct{j}{3}{k}{2}}\right)\left(\myProduct{k}{1}{n}{\myProduct{j}{3}{k}{\left(j-2\right)}}\right)^{-1}\left(\myProduct{k}{1}{n}{\myProduct{j}{3}{k}{\left(j-1\right)}}\right)^{2}\\[1pt]
            & \quad \left(\myProduct{k}{1}{n}{\myProduct{j}{3}{k}{\left(j+1\right)}}\right)^{-1}\left(\myProduct{k}{1}{n}{\myProduct{j}{3}{k}{\left(j+2\right)}}\right)
        \end{split}
        \end{align}%
    equation~\eqref{eqn:depth2HyperGeoPrdt} can be written in the form 
    \begin{equation*}
    P(n) = A_{1}(n)\,A_{2}(n)\in\ProdMon(\KK(x))
    \end{equation*}
    where the multiplicand representations of the products in $A_{1}(n)$ and $A_{2}(n)$ are either from $\KK$ or are irreducible polynomials from $\KK[x]$.
\end{example}

\subsection{Synchronization of lower bounds}\label{subsec:SyncLowerBounds}

Another transformation will guarantee that all arising products have the same lower bound, i.e., that the expression is $\delta$-refined for some $\delta\in\NN$. 

\begin{definition}\label{defn:deltaRefined}
	\normalfont Let $\KK(x)$ be a rational function field over a field $\KK$ and $\delta\in\NN$. $H(n)\in\ProdExpr(\KK(x))$ is said to be \emph{$\delta$-refined} if the lower bounds in all the arising products of $H(n)$ are $\delta$.
\end{definition}

Such a transformation of a given product expression to a $\delta$-refined version can be accomplished by taking $\delta$ to be the maximum of all arising lower bounds within the given expression.

\begin{example}[Cont. 
Example~\ref{exa:nestedHyperGeoPrdtPreprocessingStep1}]\label{exa:nestedHyperGeoPrdtPreprocessingStep2}
\normalfont
In $P(n)$ (resp.\ $A_1(n)$ and $A_2(n)$ of Example~\ref{exa:nestedHyperGeoPrdtPreprocessingStep1} we choose $\delta=3$. 
\end{example}

\noindent Namely, for all $1\le i \le m$, rewrite each product in~\eqref{eqn:nestedProductsOverNumerator} such that the lower bounds are synchronized to $\delta$. More precisely we apply the formula 
\begin{multline}\label{eqn:prdtRewriteRule}
\myProduct{k_{1}}{\ell_{1}}{n}{\,\;\;\myProduct{k_{2}}{\ell_{2}}{k_{1}}{\cdots\myProduct{k_{i}}{\ell_{i}}{k_{i-1}}{h(k_{i})}}} = \Bigg(\,\myProduct{k_{1}}{\ell_{1}}{\delta-1}{\,\;\;\;\myProduct{k_{2}}{\ell_{2}}{k_{1}}{\cdots\myProduct{k_{i}}{\ell_{i}}{k_{i-1}}{h(k_{i})}}}\Bigg)\Bigg(\,\myProduct{k_{1}}{\delta}{n}{\,\;\;\;\myProduct{k_{2}}{\ell_{2}}{\delta-1}{\cdots\myProduct{k_{i}}{\ell_{i}}{k_{i-1}}{h(k_{i})}}}\Bigg)\\ \Bigg(\,\myProduct{k_{1}}{\delta}{n}{\,\;\;\myProduct{k_{2}}{\delta}{k_{1}}{\;\;\;\myProduct{k_{3}}{\ell_{3}}{\delta-1}{\cdots\myProduct{k_{i}}{\ell_{i}}{k_{i-1}}{h(k_{i})}}}}\Bigg)\cdots\Bigg(\,\myProduct{k_{1}}{\delta}{n}{\,\;\;\myProduct{k_{2}}{\delta}{k_{1}}{\cdots\myProduct{k_{i-1}}{\delta}{k_{i-2}}{\;\;\;\;\myProduct{k_{i}}{\ell_{i}}{\delta-1}{h(k_{i})}}}}\Bigg)\Bigg(\,\myProduct{k_{1}}{\delta}{n}{\,\;\;\myProduct{k_{2}}{\delta}{k_{1}}{\cdots\myProduct{k_{i}}{\delta}{k_{i-1}}{h(k_{i})}}}\Bigg)
\end{multline}
to each of the products in~\eqref{eqn:nestedProductsOverNumerator}. Note that the first product on the right-hand side in~\eqref{eqn:prdtRewriteRule} evaluates to a constant in $\KK^*$, the last product is from $\Prod(\KK(x))$, and all the remaining products (after all finite multiplications are carried out) are from $\Prod(\KK)$. Summarizing we obtain 
\begin{equation*}
\tilde{P}(n) = \tilde{A}_{1}(n)\,\tilde{A}_{2}(n)\cdots \tilde{A}_{m}(n)\in\ProdMon(\KK(x))
\end{equation*}
with 
{\fontsize{10.9pt}{0}\selectfont
    \begin{align}
        \tilde{A}_{i}(n)&=a_{i}\Bigg(\myProduct{k_{1}}{\delta}{n}{\tilde{u}_{i,1}}\Bigg)\cdots\Bigg(\myProduct{k_{1}}{\delta}{n}{\cdots\myProduct{k_{i}}{\delta}{k_{i-1}}{\tilde{u}_{i,i}}}\Bigg)\Bigg(\myProduct{k_{1}}{\delta}{n}{\cdots \myProduct{k_{i}}{\delta}{k_{i-1}}{\threeargssubscript{f}{i}{1}(k_{i})}}\Bigg)^{\hspace*{-0.4em}\threeargssubscript{e}{i}{1}}\hspace*{-0.45em}\cdots\Bigg(\myProduct{k_{1}}{\delta}{n}{\cdots \myProduct{k_{i}}{\delta}{k_{i-1}}{\threeargssubscript{f}{i}{r_{i}}(k_{i})}}\Bigg)^{\hspace*{-0.4em}\threeargssubscript{e}{i}{r_{i}}}\label{eqn:refinedNestedProductsOverNumerator}
    \end{align}}%
where  $a_{i},\,\tilde{u}_{i,j}\in\KK^{*}$ for some $j\in\NN$. Since $\delta$ is chosen as the maximum among all lower bounds of the input expression, no poles or zero-evaluations will be introduced. As a consequence, the obtained result is again an element from $\ProdMon(\KK(x))$. In particular, we have that $A_{i}(n) = \tilde{A}_{i}(n)$ for all $n\ge\max(0,\delta-1)$ and consequently,
$P(n) = \tilde{P}(n)$
holds for all $n\ge\max(\delta-1,0)$. 

\begin{example}[Cont. Example~\ref{exa:nestedHyperGeoPrdtPreprocessingStep2}]\label{exa:nestedHyperGeoPrdtPreprocessingStep3}
    \normalfont Synchronizing the lower bounds of each product factor in~\eqref{eqn:depth1ProductFactors} and~\eqref{eqn:depth2ProductFactors} to $3$ computed in Example~\ref{exa:nestedHyperGeoPrdtPreprocessingStep2} and rewriting each product factor in~\eqref{eqn:depth1ProductFactors} and~\eqref{eqn:depth2ProductFactors} we get
    \begin{align}
        \tilde{A}_{1}(n) &= \dfrac{1225}{3}\,\left(\myProduct{k}{3}{n}{-1}\right)\left(\myProduct{k}{3}{n}{\sqrt{3}}\right)^{-1}\left(\myProduct{k}{3}{n}{24}\right)\left(\myProduct{k}{3}{n}{\left(k+\tfrac{1}{24}\right)}\right),\label{eqn:depth1RefinedProductFactors}\\[5pt]
        \begin{split}
            \tilde{A}_{2}(n) &=  \left(\myProduct{k}{3}{n}{\myProduct{j}{3}{k}{-1}}\right)\left(\myProduct{k}{3}{n}{\myProduct{j}{3}{k}{5}}\right)^{-1} \left(\myProduct{k}{3}{n}{\myProduct{j}{3}{k}{2}}\right)\left(\myProduct{k}{3}{n}{\myProduct{j}{3}{k}{\left(j-2\right)}}\right)^{-1}\left(\myProduct{k}{3}{n}{\myProduct{j}{3}{k}{\left(j-1\right)}}\right)^{2}\\[1pt]
            & \quad \left(\myProduct{k}{3}{n}{\myProduct{j}{3}{k}{\left(j+1\right)}}\right)^{-1}\left(\myProduct{k}{3}{n}{\myProduct{j}{3}{k}{\left(j+2\right)}}\right).\label{eqn:depth2RefinedProductFactors}
        \end{split} 
    \end{align}
    In particular, for $i=1,2$, and for all $n\ge\delta-1$ where $\delta=3$, $A_{i}(n)=\tilde{A}_{i}(n)$ holds. Consequently, with  
     $\tilde{P}(n) = \tilde{A}_{1}(n)\,\tilde{A}_{2}(n)$
    we have that 
    $P(n) = \tilde{P}(n)$
    holds for all $n\ge2$.
\end{example}

Since geometric products never introduce poles or zeroes, we can bring each geometric product in~\eqref{eqn:refinedNestedProductsOverNumerator} to a $1$-refined form by using a similar formula as given in~\eqref{eqn:prdtRewriteRule}. This yields  
\begin{equation*}
    P'(n) = A'_{1}(n)\,A'_{2}(n)\cdots A'_{m}(n)\in\ProdMon(\KK(x))
\end{equation*}
where 
{\fontsize{10.75pt}{0}\selectfont
    \begin{align}
        A'_{i}(n)&=\tilde{a}_{i}\Bigg(\myProduct{k_{1}}{1}{n}{\tilde{u}_{i,1}}\Bigg)\cdots\Bigg(\myProduct{k_{1}}{1}{n}{\cdots\myProduct{k_{i}}{1}{k_{i-1}}{\tilde{u}_{i,i}}}\Bigg)\Bigg(\myProduct{k_{1}}{\delta}{n}{\cdots \myProduct{k_{i}}{\delta}{k_{i-1}}{\threeargssubscript{f}{i}{1}(k_{i})}}\Bigg)^{\hspace*{-0.4em}\threeargssubscript{e}{i}{1}}\hspace*{-0.45em}\cdots\Bigg(\myProduct{k_{1}}{\delta}{n}{\cdots \myProduct{k_{i}}{\delta}{k_{i-1}}{\threeargssubscript{f}{i}{r_{i}}(k_{i})}}\Bigg)^{\hspace*{-0.4em}\threeargssubscript{e}{i}{r_{i}}}\label{eqn:refinedNestedProductsOverNumerator2}
    \end{align}}%
with $\tilde{a}_{i},\,\tilde{u}_{i,j}\in\KK^{*}$ for some $j\in\NN$. In particular we have that, $A_{i}(n) = A'_{i}(n)$ holds for all $n\ge\max(\delta-1,0)$ and consequently,
$P(n) = P'(n)$
holds for all $n\ge\max(0,\delta-1)$. By rearranging the arising products in $P'(n)$ we obtain the decomposition
\begin{equation}\label{Equ:P'(n)}
    P'(n) = c\,G(n)\,H(n)
\end{equation}
with $c\in\KK^{*}$, $G(n)\in\ProdMon(\KK)$ is composed multiplicatively by geometric products in factored form of nesting depth at most $m$ which are $1$-refined and $H(n)\in\ProdMon(\KK(x))$ is composed multiplicatively by hypergeometric products (which are not geometric) in factored form of nesting depth at most $m$ which are $\delta$-refined. In particular, the multiplicand representations are given by monic irreducible polynomials.
 
\begin{example}[Cont. Example~\ref{exa:nestedHyperGeoPrdtPreprocessingStep3}]\label{exa:nestedHyperGeoPrdtPreprocessingStep4}
    \normalfont Synchronizing the lower bounds of each geometric product in~\eqref{eqn:depth1RefinedProductFactors} and~\eqref{eqn:depth2RefinedProductFactors} to $1$ and rewriting these geometric products we get 
    \begin{align*}
        A'_{1}(n) &= \frac{1225}{576}\left(\myProduct{k}{1}{n}{-1}\right)\left(\myProduct{k}{1}{n}{\sqrt{3}}\right)^{-1}\left(\myProduct{k}{1}{n}{24}\right)\left(\myProduct{k}{3}{n}{\left(k+\tfrac{1}{24}\right)}\right),\\[5pt]
    \begin{split}
        A_{2}'(n) &= -\dfrac{2}{5}\left(\myProduct{k}{1}{n}{4}\right)^{-1}\left(\myProduct{k}{1}{n}{25}\right)\left(\myProduct{k}{1}{n}{\myProduct{j}{1}{k}{-1}}\right)\left(\myProduct{k}{1}{n}{\myProduct{j}{1}{k}{5}}\right)^{-1}\left(\myProduct{k}{1}{n}{\myProduct{j}{1}{k}{2}}\right)\left(\myProduct{k}{3}{n}{\myProduct{j}{3}{k}{\left(j-2\right)}}\right)^{-1}\\[1pt]
        & \quad \left(\myProduct{k}{3}{n}{\myProduct{j}{3}{k}{\left(j-1\right)}}\right)^{2}\left(\myProduct{k}{3}{n}{\myProduct{j}{3}{k}{\left(j+1\right)}}\right)^{-1}\left(\myProduct{k}{3}{n}{\myProduct{j}{3}{k}{\left(j+2\right)}}\right).
        \end{split}
    \end{align*}
    In particular, for $i=1,2$, and for all $n\ge2$, 
    $A_{i}(n)=A'_{i}(n)$
    holds. In total we obtain  
    \begin{align*}
    P'(n) &= A'_{1}(n)\,A'_{2}(n) = c\,G(n)\,H(n)
    \end{align*}
    with 
    \begin{align}
        c &=-\frac{245}{288} \label{eqn:constantTerm},\\[5pt]
        G(n) &= \left(\myProduct{k}{1}{n}{-1}\right)\hspace*{-0.25em}\left(\myProduct{k}{1}{n}{\sqrt{3}}\right)^{\hspace*{-0.25em}-1}\hspace*{-0.35em}\left(\myProduct{k}{1}{n}{4}\right)^{\hspace*{-0.25em}-1}\hspace*{-0.35em}\left(\myProduct{k}{1}{n}{24}\right)\hspace*{-0.30em}\left(\myProduct{k}{1}{n}{25}\right)\hspace*{-0.30em}\left(\myProduct{k}{1}{n}{\myProduct{j}{1}{k}{-1}}\right)\hspace*{-0.30em}\left(\myProduct{k}{1}{n}{\myProduct{j}{1}{k}{5}}\right)^{\hspace*{-0.25em}-1}\hspace*{-0.35em}\left(\myProduct{k}{1}{n}{\myProduct{j}{1}{k}{2}}\right), \label{eqn:geometricProductTerm}\\[5pt]
        H(n) &= \left(\myProduct{k}{3}{n}{\left(k+\tfrac{1}{24}\right)}\right)\hspace*{-0.30em}\left(\myProduct{k}{3}{n}{\myProduct{j}{3}{k}{\left(j-2\right)}}\right)^{\hspace*{-0.3em}-1}\hspace*{-0.35em}\left(\myProduct{k}{3}{n}{\myProduct{j}{3}{k}{\left(j-1\right)}}\right)^{\hspace*{-0.25em}2}\hspace*{-0.30em}\left(\myProduct{k}{3}{n}{\myProduct{j}{3}{k}{\left(j+1\right)}}\right)^{\hspace*{-0.25em}-1}\hspace*{-0.4em}\left(\myProduct{k}{3}{n}{\myProduct{j}{3}{k}{\left(j+2\right)}}\right)\label{eqn:hypergeometricProductTerm}
    \end{align}
    such that $P(n) = P'(n)$
    holds for all $n\ge2$. 
\end{example}

\subsection{Shift-coprime representation}\label{subsec:shiftCoPrimeRep}

Finally, we turn our focus to the class of hypergeometric products given in factored form, and whose innermost multiplicands are irreducible monic polynomials. In order to reduce this class of  products further, we will need the following definition.

\begin{definition}\label{defn:shiftCoPrimeAndShiftEquivalent}
    \normalfont Two nonzero polynomials $f(x)$ and $h(x)$ in the polynomial ring $\KK[x]$ are said to be \emph{shift-coprime} if for all $k\in\ZZ$ we have that $\gcd(f(x),h(x+k))=1$. Furthermore, $f(x)$ and $h(x)$ are called \emph{shift-equivalent} if there is a $k\in\ZZ$ such that $\frac{f(x+k)}{h(x)}\in\KK$.
\end{definition}

It is immediate that the shift-equivalence in Definition~\ref{defn:shiftCoPrimeAndShiftEquivalent} induces an equivalence relation on the set of all irreducible polynomials. Let $\cali{D} = \{\Lst{f}{1}{e}\}\subseteq\KK[x]$ where all elements are irreducible and shift equivalent among each other. Then we call $f_{i}\in\cali{D}$ with $i\in\{1,2,\dots,e\}$ the \emph{leftmost polynomial} in $\cal{D}$ if for all $h\in\cal{D}$ there is a $k\in\NN$ with $\frac{f_i(x+k)}{h(x)}\in\KK$. 
It is well known that $\frac{f_i(x+k)}{h(x)}\in\KK$ iff $k\in\ZZ$ is a root of $p(z)=\text{res}_x(f(x),f_i(x+z))\in\KK[z]$; compare~\cite[Sec.~5.3]{petkovvsek1996b}. In particular, if $\KK$ is computable and one can factorize univariate polynomials over $\KK$, one can determine all integer roots of $p(z)$ and thus can decide constructively if there is a $k\in\ZZ$ with $\frac{f_i(x+k)}{h(x)}\in\KK$. 
All the above properties (and slight generalizations) play a crucial role in symbolic summation; compare~\cite{abramov1971summation,paule1995greatest,schneider2005product,abramov2010polynomial,chen2011structure}. In particular, the following simple lemma is heavily used within symbolic summation; see also~\cite[Lemma 4.12]{schneider2005product}.

\begin{lemma}\label{lem:shiftEquivalentRelation}
    Let $\KK(x)$ be a rational function field and let $f(x),\,h(x)\in\KK[x]\sm\KK$ be monic irreducible polynomials that are shift equivalent. Then there is a $g\in\KK(x)^{*}$ with $h(x)=\tfrac{g(x+1)}{g(x)}\,f(x)$  were all the monic irreducible factors in $g$ are shift equivalent to $f(x)$ (resp.\ $h(x)$). If $\KK$ is computable and one can factorize polynomials over $\KK$, then such a $g$ can be computed.
\end{lemma}

\begin{proof}
    Since $f(x)$ and $h(x)$ are shift equivalent and monic, there is a $k\in\ZZ$ with $f(x+k)=h(x)$. If $k\ge0$, set $g:=\prod_{i=0}^{k-1}f(x+i)$. Then 
    \[
        \frac{g(x+1)}{g(x)} = \frac{f(x+k)}{f(x)}=\frac{h(x)}{f(x)}.
    \]
    On the other hand, if $k<0$, set $g:=\prod_{i=1}^{-k}\frac{1}{f(x+i)}$. Then 
    \[
        \frac{g(x+1)}{g(x)} = \frac{1/f(x)}{1/f(x+k)}=\frac{f(x+k)}{f(x)}=\frac{h(x)}{f(x)}.
    \]
    By construction all irreducible monic factors in $g(x)$ are shift equivalent to $f(x)$. Furthermore, $k$ can be computed if $\KK$ is computable and one can factorize polynomials over $\KK$.
\end{proof}

\begin{example}\label{exa:shiftCoPrimeReduced}
    \normalfont Let $\KK(x)$ be a rational function field as defined in Example~\ref{exa:nestedHyperGeoPrdtPreprocessingStep1}. Let $\cali{D}$ be the set defined by the multiplicand representations of the products in factored form given in~\eqref{eqn:hypergeometricProductTerm}. That is,  $\cali{D}=\{f_{1}(x),\,f_{2}(x), \dots, f_{5}(x)\}\subseteq\KK[x]\sm\KK$ where
    \begin{equation}\label{eqn:polyRepOfInnermostMultiplicands}
        f_{1}(x)=x-2,\quad f_{2}(x)=x-1,\quad f_{3}(x)=x+1,\quad f_{4}(x)=x+2 \quad \text{ and } \quad f_{5}(x) = x+\tfrac{1}{24}.
    \end{equation}
    Since $f_{1}(x)$ is shift equivalent with $f_{2}(x),f_{3}(x),f_{4}(x)$, i.e., 
    \[
        f_{1}(x+1) = f_{2}(x), \qquad f_{1}(x+3) = f_{3}(x), \qquad f_{1}(x+4) = f_{4}(x),
    \]
    they fall into the same equivalence class $\cali{E}_{1}=\{f_{1}(x),\,f_{2}(x), f_{3}(x), f_{4}(x)\}$. The other equivalence class is $\cali{E}_{2}=\{f_{5}(x)\}$. For each of these equivalence classes $\cali{E}_{1}$ and $\cali{E}_{2}$, take their leftmost elements: $f_{1}(x)$ and $f_{5}(x)$ respectively. Then by Lemma~\ref{lem:shiftEquivalentRelation}, we can express the elements of each equivalence class in terms of the leftmost polynomial $f_{1}(x)$ or $f_{5}(x)$. More precisely, we have the following relations for the equivalence class $\cali{E}_{1}$: 
    \begin{align}
        f_{2}(x) &= \frac{g_{1}(x+1)}{g_{1}(x)} f_{1}(x), \text{ with } g_{1}(x) = (x-2),\label{eqn:reducingPolyXMinus1} \\[0.25em]
        f_{3}(x) &= \frac{g_{2}(x+1)}{g_{2}(x)} f_{1}(x), \text{ with } g_{2}(x) = (x-2)\,(x-1)\,x, \label{eqn:reducingPolyXPlus1} \\[0.25em]
        f_{4}(x) &= \frac{g_{3}(x+1)}{g_{3}(x)} f_{1}(x), \text{ with } g_{3}(x) = (x-2)\,(x-1)\,x\,(x+1).\label{eqn:reducingPolyXPlus2}
    \end{align}
    Finally, we reduce each component of the hypergeometric product expression $H(n)$ given by~\eqref{eqn:hypergeometricProductTerm}. We will begin with the nesting depth $2$ hypergeometric products in factored form whose innermost multiplicand corresponds to the polynomial $f_{4}(x)$. Using~\eqref{eqn:reducingPolyXPlus2} the product in factored form reduces as follows:
    \small
        \begin{align}
        \myProduct{k}{3}{n}{\myProduct{j}{3}{k}{f_{4}(j)}} &= \myProduct{k}{3}{n}{\myProduct{j}{3}{k}{\left(j+2\right)}} = \left(\myProduct{k}{3}{n}{\myProduct{j}{3}{k}{\frac{g_{3}(j+1)}{g_{3}(j)}}}\right)\myProduct{k}{3}{n}{\myProduct{j}{3}{k}{f_{1}(j)}}=\left(\myProduct{k}{3}{n}{\frac{g_{3}(k+1)}{g_{3}(3)}}\right)\left(\myProduct{k}{3}{n}{\myProduct{j}{3}{k}{f_{1}(j)}}\right) \nn \\
        &=\left(\myProduct{k}{3}{n}{\frac{1}{24}}\right)\left(\myProduct{k}{3}{n}{\left(k-1\right)}\right)\left(\myProduct{k}{3}{n}{k}\right)\left(\myProduct{k}{3}{n}{\left(k+1\right)}\right)\left(\myProduct{k}{3}{n}{\left(k+2\right)}\right)\left(\myProduct{k}{3}{n}{\myProduct{j}{3}{k}{(j-2)}}\right) \nn \\
        &=576\,\left(\myProduct{k}{1}{n}{\frac{1}{24}}\right)\left(\myProduct{k}{3}{n}{\left(k-1\right)}\right)\left(\myProduct{k}{3}{n}{k}\right)\left(\myProduct{k}{3}{n}{\left(k+1\right)}\right)\left(\myProduct{k}{3}{n}{\left(k+2\right)}\right)\left(\myProduct{k}{3}{n}{\myProduct{j}{3}{k}{(j-2)}}\right).\label{eqn:reducedNestingDepth2HyperGeoOverJPlus2}
        \end{align}
    \normalsize
    Using~\eqref{eqn:reducingPolyXPlus1} and~\eqref{eqn:reducingPolyXMinus1}, a similar reduction can be achieved for the nesting depth $2$ hypergeometric products in factored form arising in $H(n)$ whose innermost multiplicands correspond to the polynomials $f_{3}(x)$ and $f_{2}(x)$ respectively. In particular, we have the following: 
    {\fontsize{9.5pt}{0}\selectfont
    \begin{align}
        \myProduct{k}{3}{n}{\myProduct{j}{3}{k}{f_{3}(j)}} = \myProduct{k}{3}{n}{\myProduct{j}{3}{k}{\left(j+1\right)}} &= 36\,\left(\myProduct{k}{1}{n}{\frac{1}{6}}\right)\left(\myProduct{k}{3}{n}{\left(k-1\right)}\right)\left(\myProduct{k}{3}{n}{k}\right)\left(\myProduct{k}{3}{n}{\left(k+1\right)}\right)\left(\myProduct{k}{3}{n}{\myProduct{j}{3}{k}{(j-2)}}\right),\label{eqn:reducedNestingDepth2HyperGeoOverJPlus1}\\[0.5em]
        \myProduct{k}{3}{n}{\myProduct{j}{3}{k}{f_{2}(j)}} = \myProduct{k}{3}{n}{\myProduct{j}{3}{k}{\left(j-1\right)}} &= \left(\myProduct{k}{3}{n}{\left(k-1\right)}\right)\left(\myProduct{k}{3}{n}{\myProduct{j}{3}{k}{(j-2)}}\right).\label{eqn:reducedNestingDepth2HyperGeoOverJMinus1}
    \end{align}}%
\end{example}

\begin{remark}\label{Remark:GeneralReductionForShiftCoPrimeness}
\normalfont
Suppose we are given an expression $A(n)\in\ProdExpr(\KK(x))$ (e.g., given in~\eqref{Equ:ProdEDef}) in terms of $\delta$-refined hypergeometric products of finite nesting depth in factored form where all multiplicand representations are irreducible monic polynomials. Choose 
\begin{equation}\label{Equ:StartingProd}
P'(n)=\myProd{k_{1}}{\delta}{n}\quad\cdots\quad\myProd{k_{m-1}}{\delta}{k_{m-2}}\quad\quad\myProduct{k_{m}}{\delta}{k_{m-1}}{h(k_{m})}
\end{equation}
from $A(n)$ with the multiplicand representation $h(x)\in\KK[x]$. Furthermore, among all shift-equivalent multiplicand representations within the given product expression $A(n)$, let $f(x)$ be the leftmost polynomial which lies in the same equivalence class with $h(x)$. By assumption $f(x)$ and $h(x)$ are monic irreducible with $f(n)\neq0$ and $h(n)\neq0$ for all $n\geq \delta$.
Take $k\in\NN$ with $h(x+k)=f(x)$. Then by Lemma~\ref{lem:shiftEquivalentRelation} we can take $g:=\prod_{i=0}^{k-1}f(x+i)\in\KK[x]$ such that $h(x)=\frac{g(x+1)}{g(x)}f(x)$ holds. Thus
\begin{align*}
    P'(n) &= \myProd{k_{1}}{\delta}{n}\quad\cdots\quad\myProd{k_{m-1}}{\delta}{k_{m-2}}\;\;\quad\myProduct{k_{m}}{\delta}{k_{m-1}}{\frac{g(k_{m}+1)}{g(k_{m})}f(k_{m})}=\myProd{k_{1}}{\delta}{n}\;\;\,\myProd{k_{2}}{\delta}{k_{1}}\quad\cdots\quad\myProd{k_{m-1}}{\delta}{k_{m-2}}\frac{g(k_{m-1}+1)}{g(\delta)}\myProduct{k_{m}}{\delta}{k_{m-1}}{f(k_{m})}\\
    &= \underbrace{\left(\myProd{k_{1}}{\delta}{n}\quad\cdots\quad\myProd{k_{m-1}}{\delta}{k_{m-2}}g(\delta)\right)^{\hspace*{-0.35em}-1}}_{=G'(n)}\;\underbrace{\left(\myProd{k_{1}}{\delta}{n}\quad\cdots\quad\myProd{k_{m-1}}{\delta}{k_{m-2}}g(k_{m-1}+1)\right)}_{=H'(n)}\;
\left(\myProd{k_{1}}{\delta}{n}\quad\cdots\quad\myProduct{k_{m}}{\delta}{k_{m-1}}{f(k_{m})}\right).
\end{align*}
Note that this reduction of a product of nesting depth $m$ leads to a new hypergeometric product $H'(n)$ in factored form of nesting depth less than $m$ where for the multiplicand representation $h'(x):=g(x+1)$ we have that $h'(n)\neq0$ for all $n\geq\max(\delta-1,0)$. In particular, $h'(x)$ consists of monic irreducible factors which are again shift-equivalent to $f(x)$. In addition taking all these new factors together with $f(x)$, it follows that $f(x)$ remains the leftmost polynomial factor. Thus repeating the steps in Subsection~\ref{subsec:transformationToProductFactoredForm} to $H'(n)$ yields again products of the form~\eqref{Equ:StartingProd} with nesting depth $m-1$ with the following property: among all multiplicand representations of this new expression the $f(x)$ is still the leftmost polynomial.

\noindent Note further that also the new geometric product $G'(n)$ occurs with lower bound $\delta$. In order to turn it to a $1$-refined product, we may apply the transformations introduced in Section~\ref{subsec:SyncLowerBounds}.

\noindent Finally, observe that in the special case $m=1$, we get 
$$R(n)=\frac{H'(n)}{G'(n)}=\frac{g(n+1)}{g(\delta)}\in\KK[n].$$ 
Since the product $P'(n)$ itself might arise in the expression under consideration in the form $P'(n)^z$ with $z\in\ZZ$, we might introduce the factor $\frac{1}{R(n)^z}$ in the final expression. However, since $R(n)\neq0$ for all $n\geq\max(\delta-1,0)$, no extra poles will be introduced by this extra factor. Summarizing, also the final expression that has undergone the above transformation can be evaluated for all $n\geq\max(\delta-1,0)$. In particular, the input and output expression will have the same evaluation for each $n\geq\max(\delta-1,0)$.
\end{remark}

\begin{example}[Cont.\ of Ex.~\ref{exa:shiftCoPrimeReduced}]    
\normalfont
    After reducing all nesting depth $2$ hypergeometric products in factored form in the expression $H(n)$, the new polynomial $f_{6}(x)=x$ emerges. It falls into the equivalence class $\cali{E}_{1}$, and the leftmost polynomial of this equivalence class remains unchanged. We get  
    \begin{equation}\label{eqn:reducingPolyX}
    f_{6}(x) = \frac{g_{4}(x+1)}{g_{4}(x)}f_{1}(x) \text{ with } g_{4}(x) = (x-2)\,(x-1).
    \end{equation}    
    by Lemma~\ref{lem:shiftEquivalentRelation}. Using the relations~\eqref{eqn:reducingPolyXPlus2},~\eqref{eqn:reducingPolyXPlus1},~\eqref{eqn:reducingPolyX},~and~\eqref{eqn:reducingPolyXMinus1} we can reduce all nesting depth $1$ hypergeometric products whose multiplicand representations are $f_{4}(x)$, $f_{3}(x)$, $f_{6}(x)$, and $f_{2}(x)$ respectively. More precisely we have the following:
    \begin{align}
        \myProduct{k}{3}{n}{f_{4}(k)} &= \myProduct{k}{3}{n}{\left(k+2\right)} = \frac{(n-1)\,n\,(n+1)\,(n+2)}{24}\myProduct{k}{3}{n}{\left(k-2\right)} \label{eqn:reducedNestingDepth1HyperGeoOverKPlus2} \\
        \myProduct{k}{3}{n}{f_{3}(k)} &= \myProduct{k}{3}{n}{\left(k+1\right)} = \frac{(n-1)\,n\,(n+1)}{6}\myProduct{k}{3}{n}{\left(k-2\right)} \label{eqn:reducedNestingDepth1HyperGeoOverKPlus1} \\
        \myProduct{k}{3}{n}{f_{6}(k)} &= \myProduct{k}{3}{n}{k} \qquad \ \, = \frac{(n-1)\,n}{2}\myProduct{k}{3}{n}{\left(k-2\right)} \label{eqn:reducedNestingDepth1HyperGeoOverK} \\
        \myProduct{k}{3}{n}{f_{2}(k)} &= \myProduct{k}{3}{n}{\left(k-1\right)} = (n-1)\myProduct{k}{3}{n}{\left(k-2\right)}. \label{eqn:reducedNestingDepth1HyperGeoOverKMinus1}
    \end{align}
    Substituting~\eqref{eqn:reducedNestingDepth1HyperGeoOverKPlus2},~\eqref{eqn:reducedNestingDepth1HyperGeoOverKPlus1},~\eqref{eqn:reducedNestingDepth1HyperGeoOverK}, and~\eqref{eqn:reducedNestingDepth1HyperGeoOverKMinus1} into~\eqref{eqn:reducedNestingDepth2HyperGeoOverJPlus2},~\eqref{eqn:reducedNestingDepth2HyperGeoOverJPlus1}, and~\eqref{eqn:reducedNestingDepth2HyperGeoOverJMinus1} and afterwards into the expression~\eqref{eqn:hypergeometricProductTerm} gives 
    {\fontsize{9.5pt}{0}\selectfont
    \begin{equation}\label{eqn:hypergeometricProductTermRefined}
        \hat{H}(n) = \dfrac{2}{3}\,(n-1)^{3}\,n\,(n+1)\,(n+2)\left(\myProduct{k}{1}{n}{24}\right)^{\hspace*{-0.3em}-1}\hspace*{-0.4em}\left(\myProduct{k}{1}{n}{6}\right)\hspace*{-0.3em}\left(\myProduct{k}{3}{n}{\left(k-2\right)}\right)^{\hspace*{-0.3em}3}\hspace*{-0.3em}\left(\myProduct{k}{3}{n}{\left(k+\tfrac{1}{24}\right)}\right)\hspace*{-0.35em}\left(\myProduct{k}{3}{n}{\myProduct{j}{3}{k}{\left(j-2\right)}}\right).
    \end{equation}}%
    Note that $H(n)=\hat{H}(n)$ for all $n\ge2$. Furthermore, the distinct irreducible monic polynomials: $(x-2)$ and $(x+\tfrac{1}{24})$, that corresponds to the distinct innermost multiplicands of the products in factored form in $\hat{H}(n)$ are shift-coprime among each other. Putting~\eqref{eqn:constantTerm} and~\eqref{eqn:geometricProductTerm} in Example~\ref{exa:nestedHyperGeoPrdtPreprocessingStep4} and~\eqref{eqn:hypergeometricProductTermRefined} together, we have that 
    \begin{equation}\label{eqn:shiftCoPrimeRepresentationOfProducts}
    P(n) = \tilde{P}(n) = \tilde{c}\,\tilde{r}(n)\,\tilde{G}(n)\,\tilde{H}(n)
    \end{equation}
    holds for all $n\in\NN$ with $n\ge2$, where the components of $\tilde{P}(n)$ are as follows: 
    \begin{align}
    \tilde{c} &=-\frac{254}{432}, \label{eqn:refinedConstantTerm} \\[0.2em]
    \tilde{r}(n) &= (n-1)^{3}\,n\,(n+1)\,(n+2), \label{eqn:rationalFunctionTerm} \\[0.2em]
    \tilde{G}(n) &= \left(\myProduct{k}{1}{n}{-1}\right)\hspace*{-0.3em}\left(\myProduct{k}{1}{n}{\sqrt{3}}\right)^{\hspace*{-0.35em}-1}\hspace*{-0.3em}\left(\myProduct{k}{1}{n}{2}\right)^{\hspace*{-0.3em}-1}\hspace*{-0.5em}\left(\myProduct{k}{1}{n}{3}\right)\hspace*{-0.3em}\left(\myProduct{k}{1}{n}{25}\right)\hspace*{-0.3em}\left(\myProduct{k}{1}{n}{\myProduct{j}{1}{k}{-1}}\right)\hspace*{-0.3em}\left(\myProduct{k}{1}{n}{\myProduct{j}{1}{k}{5}}\right)^{\hspace*{-0.3em}-1}\hspace*{-0.4em}\left(\myProduct{k}{1}{n}{\myProduct{j}{1}{k}{2}}\right), \label{eqn:geometricProducts}\\[0.2em]
    \tilde{H}(n) &= \left(\myProduct{k}{3}{n}{\left(k-2\right)}\right)^{3}\left(\myProduct{k}{3}{n}{\left(k+\tfrac{1}{24}\right)}\right)\,\left(\myProduct{k}{3}{n}{\myProduct{j}{3}{k}{\left(j-2\right)}}\right).\label{eqn:hypergeometricProductInShifCoPrimeRepForm}
    \end{align}
\end{example}

\noindent For further considerations the following definition will be convenient.

\begin{definition}\label{defn:shiftCoPrimeProductRepresentationForm}
    \normalfont Let $\KK(x)$ be a rational function field and let $H_{1}(n),\dots,H_{e}(n)\in\ProdMon(\KK(x))$ where the arising hypergeometric products are in factored form. We say that $H_{1}(n),\dots,H_{e}(n)$ are in \emph{shift-coprime product representation form} if 
    \begin{enumerate}[\hspace*{1em}(1)]
        \item the multiplicand representation of each product in $H_{i}(n)$ for $1\le i \le e$ is an irreducible monic polynomial in $\KK[x]\sm\KK$;
        \item the distinct multiplicand representations in $H_{1}(n),\dots,H_{e}(n)$ are shift-coprime among each other. 
    \end{enumerate}  
\end{definition}

\noindent Then the above symbolic manipulations can be summarized by the following method.

\begin{remark}\label{remk:algorithmForComputingShiftCoPrimePrdtRepForm}
    \normalfont We are given $H_1(n),\dots,H_e(n)$ in $\ProdMon(\KK(x))$ where the products are in factored form and are all $\delta$-refined for some $\delta\in\NN$. In particular, suppose that all multiplicand representations are monic and that each $H_i(n)$ can be evaluated for $n\geq\nu$ for some $\nu\in\NN$ with $\nu\geq\delta$; note that such a $\nu$ can be derived by applying the available $Z$-function to each rational function factor of $H_i(n)$ and taking its maximum value. Then one can follow the steps below to rewrite $H_1(n),\dots,H_e(n)$ in a shift-coprime product representation form that yield the same evaluations for all $n\geq\nu$. 
    \begin{enumerate}[(1)]
        \item \manuallabel{item1:shiftCoPrimePrdtRepForm}{(1)}Factor the innermost multiplicand representations of all products in $H_1(n),\dots,H_e(n)$ into irreducible monic polynomials in $\KK[x]\sm\KK$ as in~\eqref{eqn:factorAandB}, and expand the product quantifier over the factorization as in~\eqref{eqn:nestedProductsOverNumerator}. Let ${\cal{D}}$ be the set of all the irreducible monic polynomials.
        \item\manuallabel{item2:shiftCoPrimePrdtRepForm}{(2)} Among all the irreducible monic polynomials in ${\cal{D}}\subseteq \KK[x]\sm\KK$, compute the shift equivalence classes say, $\Lst{\cal{E}}{1}{v}$ with respect to the automorphism $\s(x)=x+1$, and let $\cal{R}$ be the set of the leftmost polynomial of each equivalent class. Thus, the elements of the set $\cal{R}$ are shift-coprime among each other and each element represents exactly one equivalence class.
        \item\manuallabel{item3:shiftCoPrimePrdtRepForm}{(3)} Among all the products over the irreducible monic polynomials obtained in step~\ref{item1:shiftCoPrimePrdtRepForm}, take those with the highest nesting depth and reduce them with the elements in $\cal{R}$ following the construction of Remark~\ref{Remark:GeneralReductionForShiftCoPrimeness}. During this rewriting one also obtains extra constants, geometric products and rational expressions from $\KK(n)$ that are collected accordingly.       
        \item\manuallabel{item4:shiftCoPrimePrdtRepForm}{(4)} Go to\footnote{As observed in Remark~\ref{Remark:GeneralReductionForShiftCoPrimeness} the set $\cal R$ of the leftmost polynomials in step~\ref{item2:shiftCoPrimePrdtRepForm} does not change} step~\ref{item1:shiftCoPrimePrdtRepForm} and update the corresponding product expressions until the multiplicand representations of all products are in ${\cal R}$.
    \end{enumerate} 
In particular, if $\KK$ is computable and one can factorize polynomials over $\KK$, all the above steps can be carried out explicitly.
\end{remark}

Summarizing, given $\{P_{1}(n),\dots,P_{e}(n)\}\subseteq\Prod(\KK(x))$, we can bring each $P_i(n)$ to the form~\eqref{Equ:P'(n)} with $P'_i(n):=c_i\,G_i(n)\,H_i(n)$ with $c_i\in\KK^*$, $G_i(n)\in\ProdMon(\KK)$ and $H_i(n)\in\ProdMon(\KK(x))$ such that $P_i(n)=P'_i(n)$ holds for all $n\geq\max(0,\delta-1)$. 
Then applying the sketched algorithm in Remark~\ref{remk:algorithmForComputingShiftCoPrimePrdtRepForm} to $H_{1}(n),\dots,H_{e}(n)$ we get the output $\tilde{H}_1(n),\dots,\tilde{H}_e(n)$ where the $c_i$ and $H_i(n)$ are correspondingly updated to $\tilde{c}_i$ and $\tilde{H}_i(n)$ within Step~3 of Remark~\ref{remk:algorithmForComputingShiftCoPrimePrdtRepForm}. In particular, we obtain for each component an extra factor $\tilde{r}_i(x)\in\KK(x)$ where for $n\in\NN$ with $n\geq\max(0,\delta-1)$, no poles are introduced in the evaluation of $\tilde{r}_i(n)$.
Afterwards, another synchronization will be necessary to bring the new geometric products in $H_i(n)$ to $1$-refined form (by again updating the $\tilde{c}_i$ accordingly). The final result can be summarized in the following proposition. 

\begin{proposition}\label{pro:preprocessingNestedHypergeometricProductsExtended}
    Let $\KK(x)$ be a rational function field 
    and suppose that we are given the hypergeometric products $\{P_{1}(n),\dots,P_{e}(n)\}\subseteq\Prod(\KK(x))$ of nesting depth at most $d\in\NN$. Then there is  a $\delta\in\NN$ and there are 
    \begin{enumerate}[\hspace*{1em}(1)]
        \item $\Lst{\tilde{c}}{1}{e}\in\KK^{*}$;
        \item for all $1\le \ell\le e$ rational functions $r_{\ell}(x)\in\KK(x)^{*}$;
        \item geometric product expressions $\tilde{G}_{1}(n),\dots,\tilde{G}_{e}(n)\in\ProdMon(\KK)$ which are all $1$-refined;
        \item hypergeometric product expressions $\tilde{H}_{1}(n),\dots,\tilde{H}_{e}(n)\in\ProdMon(\KK(x))$ which are $\delta$-refined and are in shift-coprime product representation form
    \end{enumerate}
    such that for $1\le \ell \le e$ and for all $n\ge\max(0,\delta-1)$ we have
    \begin{align}\label{eqn:splittingNestedHyperGeoProductsModified}
    P_{\ell}(n) = \tilde{c}_{\ell}\,\tilde{r}_{\ell}(n)\,\tilde{G}_{\ell}(n)\,\tilde{H}_{\ell}(n)\neq0.
    \end{align}
    If $\KK$ is computable and one can factorize polynomials in $\KK$, then $\delta$ and the above representation can be computed.
\end{proposition}

From now on, we assume that the arising hypergeometric products $P_{1}(n),\dots,P_{e}(n)\in\Prod(\KK(x))$ have undergone the preprocessing steps discussed above  yielding the representation given in~\eqref{eqn:splittingNestedHyperGeoProductsModified}. In general, there are still algebraic relations among the products that occur in the derived expressions~\eqref{eqn:splittingNestedHyperGeoProductsModified} with $1\leq\ell\le e$, i.e., statements~\ref{item2:ProblemRPE} and~\ref{item3:ProblemRPE} of~\ref{prob:ProblemRPE} do not hold yet.  In order to accomplish this task, extra insight from {\dr} theory will be utilized. More precisely, we will show that the hypergeometric products coming from the $\tilde{H}_l$ are already algebraically independent, but the representation of the geometric products have to improved to establish a solution of~\ref{prob:ProblemRPE}.

\section{A naive {\dr} approach: \apE-extensions}\label{sec:apExts}
Inspired by~\cite{karr1981summation,schneider2016difference,schneider2017summation}, this Section focuses on an algebraic setting of {\dr s} (resp.\ fields) in which expressions of $\ProdExpr(\KK(x))$ can be naturally rephrased. 

\subsection{Difference Fields and Difference Rings}\label{subsec:diffFieldsAndDiffRings}

A \emph{\dr} (resp.\ \emph{field}) $\dField{\AA}$ is a ring (resp.\ field) $\AA$ together with a ring (resp.\ field) automorphism $\s:\AA\to\AA$. Subsequently, all rings (resp.\ fields) are commutative with unity; in addition they contain the set of rational numbers $\QQ$, as a subring (resp.\ subfield). The multiplicative group of units of a ring (resp.\ field) $\AA$ is denoted by $\AA^{*}$. A difference ring (resp.\ field) $\dField{\AA}$ is called \emph{computable} if $\AA$ and $\s$ are both computable. 
%
In the following we will introduce \apE-extensions that will be the foundation to represent hypergeometric products of finite nesting depth in difference rings.

\aE-extensions will be used to cover objects like $\zeta^{k}$ where $\zeta$ is a root of unity. In general, let $\dField{\AA}$ be a {\dr} and let $\zeta\in\AA^{*}$ be a $\lambda$-th root of unity with $\lambda>1$ (i.e., $\lambda\in\ZZ_{\geq2}$ with $\zeta^{\lambda}=1$). Take the uniquely determined {\drE} $\dField{\AA[y]}$ of $\dField{\AA}$ where $y$ is transcendental over $\AA$ and $\s(y) = \zeta\,y$. Now consider the ideal $I:=\geno{y^{\lambda}-1}$ and the quotient ring $\EE := \AA[y]/I$. Since $I$ is closed under $\s$ and $\s^{-1}$ i.e., $I$ is a reflexive difference ideal, we can define the map $\s : \EE \to \EE$ with $\s(h+I) = \s(h) + I$ which forms a ring automorphism. Note that by this construction the ring $\AA$ can naturally be embedded into the ring $\EE$ by identifying $a \in \AA$ with $a + I \in \EE$, i.e., $a \mapsto a + I$. Now set $\vartheta := y + I$. Then $\dField{\AA[\vartheta]}$ is a {\drE} of $\dField{\AA}$ subject to the relations $\vartheta^{\lambda} = 1$ and $\s(\vartheta) = \zeta\,\vartheta$. This extension is called an \emph{algebraic extension} (in short \emph{\aE-extension}) of order $\lambda$. The generator $\vartheta$ is called an \emph{\aE-monomial} with its order  $\lambda=\min\{n>0\,|\,\vartheta^{n}=1\}$. Note that the ring $\AA[\vartheta]$ is not an integral domain (i.e., it has zero-divisors) since $(\vartheta - 1)\,(\vartheta^{\lambda-1}+\cdots+\vartheta+1) = 0$ but $(\vartheta - 1) \neq 0 \neq (\vartheta^{\lambda-1}+\cdots+\vartheta+1)$. 
In this setting, the \aE-monomial $\vartheta$  with the relations $\vartheta^{\lambda}=1$ and $\s(\vartheta) = \zeta\,\vartheta$ with $\zeta:=\ee^{\frac{2\,\pi\,\ii}{\lambda}} = (-1)^{\frac{2}{\lambda}}$, models $\zeta^{k}$ subject to the relations $(\zeta^{k})^{\lambda} = 1$ and $\zeta^{k+1}=\zeta\,\zeta^{k}$. 

In addition, we define \pE-extensions in order to treat products of finite nesting depth whose multiplicands are not given by roots of unity. Let $\dField{\AA}$ be a {\dr}, $\alpha \in \AA^{*}$ be a unit, and consider the ring of Laurent polynomials $\AA[t, t^{-1}]$ (i.e., $t$ is transcendental over $\AA$). Then there is a unique {\drE} $\dField{\AA[t, t^{-1}]}$ of $\dField{\AA}$ with $\s(t)=\alpha\,t$ and $\s(t^{-1}) = \alpha^{-1}\,t^{-1}$. The extension here is called a \emph{product-extension} (in short \emph{\pE-extension}) and the generator $t$ is called a \emph{\pE-monomial}. 

We introduce the following notations for convenience. Let $\dField{\EE}$ be a {\drE} of $\dField{\AA}$ with $t\in\EE$.  $\AA\genn{t}$ denotes the ring of Laurent polynomials $\AA[t,\tfrac{1}{t}]$ (i.e., $t$ is transcendental over $\AA$) if $\dField{\AA[t,\tfrac{1}{t}]}$ is a \pE-extension of $\dField{\AA}$. Lastly, $\AA\genn{t}$ denotes the ring $\AA[t]$ with $t\notin\AA$ but subject to the relation $t^{\lambda}=1$ if $\dField{\AA[t]}$ is an \aE-extension of $\dField{\AA}$ of order $\lambda$. 

We say that the {\drE} $\dField{\AA\genn{t}}$ of $\dField{\AA}$ is an \apE-extension (and $t$ is an \apE-monomial) if it is an \aE- or a \pE-extension. Finally, we call $\dField{\AA\genn{t_{1}}\dots\genn{t_{e}}}$ a (nested) \aE-/\pE-/\apE-extension of $\dField{\AA}$ it is built by a tower of such extensions.

In the following we will restrict to the subclass of ordered simple \apE-extension. Here, the following definitions are useful.
\begin{definition}\label{defn:depthFunction}
    \normalfont Let $\dField{\EE}$ be a (nested) \apE-extension of $\dField{\AA}$ with $\EE=\AA\genn{t_{1}}\dots\genn{t_{e}}$ where $\s(t_{i})=\alpha_{i}\,t_{i}$ for $1\le i \le e$. We define the \emph{depth} function  of elements of $\EE$ over $\AA$, $\myd:\EE\to\NN$ as follows:
    \begin{enumerate}[(1)]
        \item For any $h\in\AA$, $\myd_{\AA}(h)=0$.
        \item If $\myd_{\AA}$ is defined for $\dField{\AA\genn{t_{1}}\dots\genn{t_{i-1}}}$ with $i>1$, then we define $\myd_{\AA}(t_{i}):=\myd_{\AA}(\alpha_{i})+1$ and for $f\in\AA\genn{t_{1}}\dots\genn{t_{i}}$, we define 
        $\myd_{\AA}(f) := \max\big(\{ \myd_{\AA}(t_{i}) \,|\, t_{i} \text{ occurs in } f \} \cup \zs\big)$.
    \end{enumerate}
    The \emph{extension depth} of $\dField{\EE}$ over $\AA$ is given by 
    $\myd_{\AA}(\EE) := \big(\myd_{\AA}(t_{1}), \dots, \myd_{\AA}(t_{e})\big)$.
    We call such an extension \emph{ordered}, if $\myd_{\AA}(t_{1})\leq \myd_{\AA}(t_{2})\leq\cdots\leq\myd_{\AA}(t_{e})$. In particular, we say that $\dField{\EE}$ is of \emph{monomial depth} $m$ if $m=\max(0,\myd_{\AA}(t_{1}),\dots,\myd_{\AA}(t_{e}))$. If $\AA$ is clear from the context, we write $\myd_{\AA}$ as $\myd$.
\end{definition}

Now, let $\dField{\EE}$ with $\EE = \AA\genn{t_{1}}\dots\genn{t_{e}}$ be a nested \aE-/\pE-/\apE-extension of a {\dr} $\dField{\AA}$ and let $G$ be a multiplicative subgroup of $\AA^{*}$. 
Following~\cite{schneider2016difference,schneider2017summation}
we call 
\begin{equation}\label{set:PrdtGrpWRTAPMonomials}
G_{\AA}^{\EE} := \{ g\,\ProdLst{t}{1}{e}{v}\,|\, g\in G, \text{ and } v_{i}\in\ZZ\}
\end{equation}
the \emph{product group} over $G$ with respect to \aE-/\pE-/\apE-monomials for the nested \aE-/\pE-/\apE-extension $\dField{\EE}$ of $\dField{\AA}$. In the following we will restrict ourselves to the following subclass of \apE-extensions.
\begin{definition}\label{defn:simpleNestedAExtension}
    \normalfont Let $\dField{\AA}$ be a difference ring and let $\GG$ be a subgroup of $\AA^{*}$. Let $\dField{\EE}$ be an \aE-/\pE-/\apE-extension of $\dField{\AA}$ with $\EE=\AA\langle t_{1}\rangle\dots\langle t_{e}\rangle$.
    Then this extension is called \emph{$\GG$-simple} if for all $1 \le i \le e$,
$$\frac{\s(t_{i})}{t_{i}}\in\GG_{\AA}^{\AA\langle t_{1}\rangle\dots\langle t_{i-1}\rangle}.$$
In addition such a $\GG$-simple extension is called $\GG$-basic\footnote{In other words, products whose multiplicands are roots of unity have nesting depth $1$ (and the roots of unity are from the constant field), whiles the remaining products do not depend on these products over roots of unity.}, if for any \aE-monomial $t_i$ we have $\frac{\s(t_{i})}{t_{i}}\in\const\dField{\AA}^*$ and for any \pE-monomial $t_i$ we have that $\frac{\s(t_{i})}{t_{i}}\in \GG_{\AA}^{\AA\langle t_{1}\rangle\dots\langle t_{i-1}\rangle}$ is free of \aE-monomials.
If $\GG=\AA^*$, such extensions are also called \emph{simple} (resp.\ \emph{basic}) instead of $\AA^*$-simple ($\AA^*$-basic).
\end{definition}

In particular, we will work with the following class of simple \aE-/\pE-/\apE-extensions that are closely related to the products in factored form given in~\eqref{eqn:productFactoredForm}; for concrete constructions see Example~\ref{exa:singleChainAPExtensions} below. 
\begin{definition}\label{defn:singleAndMultipleChainExtensions}
    \normalfont Let $\dField{\AA}$ be a difference ring and $\GG$ be a subgroup of $\AA^*$. We call $\dField{\AA\genn{t_{1}}\dots\genn{t_{e}}}$ a \emph{single chain \aE-/\pE-/\apE-extension of $\dField{\AA}$ over $\GG$} if for all $1\le k \le e$,
    \begin{equation*}
    \s(t_{k}) = c_{k}\,t_{1}\cdots\,t_{k-1}\,t_{k}, \quad \text{ with } \ c_{k}\in\GG.
    \end{equation*}
    We call $c_{1}$ also the \emph{base} of the single chain \aE-/\pE-/\apE-extension. If $\GG=\AA^*$, we also say that $\dField{\AA\genn{t_{1}}\dots\genn{t_{e}}}$ is a \emph{single chain \aE-/\pE-/\apE-extension} of $\dField{\AA}$.     
    Further, we call $\dField{\EE}$ a \emph{multiple chain \aE-/\pE-/\apE-extension} of $\dField{\AA}$ over $\GG$ with base $(\Lst{c}{1}{m})\in\GG^{m}$ if it is a tower of $m$ single chain \aE-/\pE-extensions over $\GG$ with the bases $\Lst{c}{1}{m}$, respectively. If $\GG=\AA^*$, we simply call it a multiple chain \aE-/\pE-/\apE-extension.
\end{definition}

\begin{remark}\label{Remark:SimpleProperties}
\normalfont
    Let $\dField{\AA\genn{t_{1}}\dots\genn{t_{e}}}$ be a single chain \aE-/\pE-/\apE-extension of $\dField{\AA}$ as given in Definition~\ref{defn:singleAndMultipleChainExtensions} and let $\myd:\AA\genn{t_{1}}\dots\genn{t_{e}}\to\NN$ be the depth function over $\AA$. Then we have $\myd(t_k)=k$ for all $1\le k \le e$. In particular, the extension is ordered, its extension-depth is $(1,2,\dots,e)$ and the monomial depth is $e$. Furthermore observe that for $2\leq i\leq e$ we have 
    $$\sigma(t_i)=\sigma(t_{i-1})\,t_i\quad\Leftrightarrow\quad
    \frac{t_i}{\sigma^{-1}(t_i)}=t_{i-1}.$$
\end{remark}

\subsection{Ring of Sequences}\label{subsec:ringOfSequences}
For a field $\KK$ we denote by $\KK^{\NN}$ the set of all sequences 
\begin{align}\label{eqn:seq}
\seqA{a}{}{n} = \geno{a(0),a(1),a(2),\dots\,}
\end{align}
whose terms are in $\KK$. Equipping $\KK^{\NN}$ with component-wise addition and multiplication, we get a commutative ring. In this ring, the field $\KK$ can be naturally embedded into $\KK^{\NN}$ as a subring, by identifying any $c\in\KK$ with the constant sequence $\geno{c,c,c,\dots\,}\in\KK^{\NN}$. Following the construction in \cite[Section 8.2]{petkovvsek1996b}, we turn the shift operator
$S:\KK^{\NN} \rightarrow \ \KK^{\NN}$ with
\begin{equation}\label{Equ:DefS}
           S: \geno{a(0),\,a(1),\,a(2),\,\dots\,}  \mapsto \ \geno{a(1),\,a(2),\,a(3),\,\dots\,}
\end{equation}
into a ring automorphism by introducing an equivalence relation $\sim$ on sequences in $\KK^{\NN}$. Two sequences $A=\seqA{a}{}{n}$ and $B=\seqA{b}{}{n}$ are said to be equivalent (in short $A\sim B$) if and only if there exists a non-negative integer $\delta$ such that 
\[
    \forall \, n \geq \delta : a(n) = b(n).
\]
The set of equivalence classes form a ring again with component-wise addition and multiplication which we will denote by  $\ringOfEquivSeqs := \faktor{\KK^{\NN}}{\sim}$. Now it is obvious that $S:\ringOfEquivSeqs \to \ringOfEquivSeqs$ with~\eqref{Equ:DefS} is bijective and thus a ring automorphism. We call $\dField[S]{\ringOfEquivSeqs}$ also the \emph{difference ring of sequences over $\KK$}. For simplicity, we denote the elements of $\ringOfEquivSeqs$ by the usual sequence notation as in \eqref{eqn:seq} above.

We will follow the convention introduced in~\cite{paule2019towards} to  illustrate how the indefinite products of finite nesting depth covered in this article are modelled by expressions in a difference ring. 

\begin{definition}\label{defn:modelSeqsInDiffRings}
    \normalfont Let $\dField{\AA}$ be a {\dr} with a constant field $\KK=\const\dField{\AA}$. An \emph{evaluation function $\ev:\AA\times\NN\to\KK$ for $\dField{\AA}$} is a function which satisfies the following three properties: 
    \begin{enumerate}[(i)]
        \item for all $c \in \KK$, there is a natural number $\delta \geq 0$ such that
            \begin{align}\label{sta:evalOfConstants}
                \forall\,n\geq\delta\,:\,\ev(c,n)=c;
            \end{align}
        \item for all $f,g\in\AA$ there is a natural number $\delta \geq 0$ such that 
            \begin{equation}\label{sta:evaltnHom}
                \begin{aligned}
                    \forall\,n\geq\delta\,&:\,\ev(f\,g,n)=\ev(f,n)\,\ev(g,n),\\
                    \forall\,n\geq\delta\,&:\,\ev(f+g,n)=\ev(f,n)+\ev(g,n);
                \end{aligned}
            \end{equation}
        \item for all $f\in\AA$ and $i\in\ZZ$, there is a natural number $\delta \geq 0$ such that
            \begin{align}\label{sta:evalOfShiftF}
                \forall\,n\geq\delta\,:\,\ev(\s^{i}(f),n)=\ev(f,n+i).
            \end{align}
    \end{enumerate}
    We say a sequence $\seqA{F}{}{n}\in\ringOfEquivSeqs$ is \emph{modelled}\index{model sequence in {\dr}} by $f\in\AA$ in the {\dr} $\dField{\AA}$, if there is an evaluation function $\ev$ such that
    \[
        F(k)=\ev(f,k)
    \]
    holds for all $k\in\NN$ from a certain point on.
\end{definition}

In this article, our base field is a rational function field $\KK(x)$ which is equipped with the evaluation function $\ev:\KK(x)\times\NN\to\KK$ defined as follows. For $f = \frac{g}{h} \in \KK(x)$ with $h \neq 0$ where $g$ and $h$ are coprime (if $g=0$ we take $h=1$) we have
\begin{align}\label{map:evaForRatFxns}
    \ev(f, k) := \begin{cases}
                    0 & \text{if } h(k) = 0, \\
                    \frac{g(k)}{h(k)} & \text{if } h(k) \neq 0.
                \end{cases}
\end{align}
Here, $g(k)$ and $h(k)$ are the usual polynomial evaluation at some natural number $k$.

Then given a tower of \apE-extension defined over $\dField{\KK(x)}$, one can define an appropriate evaluation function by iterative applications of the following lemma that is implied by ~\cite[Lemma~5.4]{schneider2017summation}.

\begin{lemma}\label{Lemma:ExtendEv}
	Let $\dField{\AA}$ be a {\dr} with constant field $\KK$ and let $\ev:\AA\times\NN\to\KK$ be an evaluation function for $\dField{\AA}$. Let $\dField{\AA\langle t \rangle}$ be an \apE-extension of $\dField{\AA}$ with $\s(t)=\alpha\,t$ ($\alpha\in\AA^*$) and suppose that there is a $\delta\in\NN$ such that $\ev(\alpha,n)\neq0$ for all $n\ge\delta$. Further, take $u\in\KK^{*}$; if $t^{\lambda}=1$ for some $\lambda>1$, we further assume that $u^{\lambda}=1$ holds. Consider the map $\ev^{\prime}:\AA\langle t \rangle\times\NN\to\KK$
	defined by 
	\[
	\ev^{\prime}(\sum_{i}h_{i}\,t^{i}, n)=\sum_{i}\ev(h_{i},n)\,\ev^{\prime}(t,n)^{i}
	\]
	with 
	$\ev^{\prime}(t,n)=u\prod_{k=\delta}^{n}\ev(\alpha,k-1)$.
	Then $\ev'$ is an evaluation function for $\dField{\AA\langle t \rangle}$.
\end{lemma}

We summarize the above constructions with the following example.

\begin{example}\label{exa:singleChainAPExtensions}
    \normalfont Let $\KK=\QQ(\sqrt{3})$ and take the difference field $\dField{\KK(x)}$ where the automorphism is defined by $\s(x)=x+1$ and $\s|_{\KK}=\id$. Furthermore, take the evaluation function $\ev:\KK(x)\times\NN\to\KK$ given by~\eqref{map:evaForRatFxns}. Then 
     we can construct the following single chain extensions of $\dField{\KK(x)}$ in order to 
     model the geometric product $\tilde{G}(n)$ and the hypergeometric product $\tilde{H}(n)$ given in~\eqref{eqn:geometricProducts} and~\eqref{eqn:hypergeometricProductInShifCoPrimeRepForm}.   
    \begin{enumerate}[\hspace*{1em}(1)]
        \item\manuallabel{item:monomialDepth2SingleChainRExtBasedAtMinus1}{(1)} We define the single chain \aE-extension $\dField{\KK(x)\genn{{\vartheta_{1,1}}}\genn{{\vartheta_{1,2}}} }$ of $\dField{\KK(x)}$ over $\KK$ of order $2$, based at $-1$ where the automorphism is given by
                \begin{equation}\label{diffAutoEval:monomialDepth2SingleChainRExtBasedAtMinus1A}
        \begin{aligned}
        \s({\vartheta_{1,1}}) &= -\,{\vartheta_{1,1}}, \qquad & \s({\vartheta_{1,2}}) &= -\,{\vartheta_{1,1}}\,{\vartheta_{1,2}}. \end{aligned}
        \end{equation}
     In addition by applying Lemma~\ref{Lemma:ExtendEv} twice we extend the evaluation function to
        $\ev:\KK(x)\genn{{\vartheta_{1,1}}}\genn{{\vartheta_{1,2}}}\times\NN\to\KK$ with 
        \begin{equation}\label{diffAutoEval:monomialDepth2SingleChainRExtBasedAtMinus1B}
            \begin{aligned}
                \ev({\vartheta_{1,1}},n) &= \myProduct{k}{1}{n}{-1}, \qquad & \ev({\vartheta_{1,2}},n) &= \myProduct{k}{1}{n}{\myProduct{j}{1}{k}{-1}}.
            \end{aligned}
        \end{equation}
        \item\manuallabel{item:monomialDepth1SingleChainPiExtBasedAtSqrt3}{(2)} Similarly, define the single chain \pE-extension $\dField{\KK(x)\genn{y_{1,1}}}$ of $\dField{\KK(x)}$ over $\KK$ based at $\sqrt{3}$ together with the evaluation function $\ev:\KK(x)\genn{{y_{1,1}}}\times\NN\to\KK$ (using Lemma~\ref{Lemma:ExtendEv}) by 
        \begin{align}\label{diffAutoEval:monomialDepth1SingleChainPiExtBasedAtSqrt3}
            \s({y_{1,1}}) = \sqrt{3}\,{y_{1,1}}, && \text{ and } &&  \ev({y_{1,1}},n) = \myProduct{k}{1}{n}{\sqrt{3}}.
        \end{align}       
        \item\manuallabel{item:monomialDepth2SingleChainPiExtBasedAt2}{(3)} Define the single chain \pE-extension $\dField{\KK(x)\genn{{y_{2,1}}}\genn{{y_{2,2}}}}$ of $\dField{\KK(x)}$ over $\KK$ based at $2$ equipped with the evaluation function $\ev:\KK(x)\genn{{y_{2,1}}}\genn{{y_{2,2}}}\times\NN\to\KK$ by
        \begin{equation}\label{diffAutoEval:monomialDepth2SingleChainPiExtBasedAt2}
            \begin{aligned}
                \s({y_{2,1}}) &= 2\,{y_{2,1}}, \qquad & \s({y_{2,2}}) &= 2\,{y_{2,1}}\,{y_{2,2}}, \qquad \text{ and } \\
                \ev(y_{2,1}, n) &= \myProduct{k}{1}{n}{2}, \qquad & \ev(y_{2,2}, n) &= \myProduct{k}{1}{n}{\myProduct{j}{1}{k}{2}}.
            \end{aligned}
        \end{equation}
        \item\manuallabel{item:monomialDepth1SingleChainPiExtBasedAt3}{(4)} Define the single chain \pE-extension $\dField{\KK(x)\genn{y_{3,1}}}$ of $\dField{\KK(x)}$ over $\KK$ based at $3$ together with the evaluation function $\ev:\KK(x)\genn{{y_{3,1}}}\times\NN\to\KK$ by 
        \begin{align}\label{diffAutoEval:monomialDepth1SingleChainPiExtBasedAt3}
            \s({y_{3,1}}) = 3\,{y_{3,1}}, && \text{ and } &&  \ev({y_{3,1}},n) = \myProduct{k}{1}{n}{3}.
        \end{align} 
        \item\manuallabel{item:monomialDepth2SingleChainPiExtBasedAt5}{(5)} Define the single chain \pE-extension $\dField{\KK(x)\genn{{y_{4,1}}}\genn{{y_{4,2}}}}$ of $\dField{\KK(x)}$ over $\KK$ with $5$ as its base accompanied with the evaluation function $\ev:\KK(x)\genn{{y_{4,1}}}\genn{{y_{4,2}}}\times\NN\to\KK$ by
        \begin{equation}\label{diffAutoEval:monomialDepth2SingleChainPiExtBasedAt5}
            \begin{aligned}
                \s({y_{4,1}}) &= 5\,{y_{4,1}},  \qquad &  \s({y_{4,2}}) &= 5\,{y_{4,1}}\,{y_{4,2}}, \qquad \text{ and } \\
                \ev(y_{4,1}, n) &= \myProduct{k}{1}{n}{5}, \qquad & \ev(y_{4,1}, n) &= \myProduct{k}{1}{n}{\myProduct{j}{1}{k}{5}}.
            \end{aligned}
        \end{equation}
        \item\manuallabel{item:monomialDepth1SingleChainPiExtBasedAt25}{(6)} Define the single chain \pE-extension $\dField{\KK(x)\genn{y_{5,1}}}$ of $\dField{\KK(x)}$ over $\KK$ with $25$ as its base together with the evaluation function $\ev:\KK(x)\genn{y_{5,1}}\times\NN\to\KK$ by
        \begin{align}\label{eqn:monomialDepth1SingleChainPiExtBasedAt25}
            \s(y_{5,1}) = 25\,y_{5,1}, && \text{ and } &&  \ev({y_{5,1}},n) = \myProduct{k}{1}{n}{25}.
        \end{align}
        \item\manuallabel{item:monomialDepth2SingleChainPiExtBasedAtXMinus2}{(7)} Define the single chain \pE-extension $\dField{\KK(x)\genn{{z_{1,1}}}\genn{{z_{1,2}}}}$ of $\dField{\KK(x)}$ over $\KK(x)$ based at $(x-2)$ and the evaluation function $\ev:\KK(x)\genn{{z_{1,1}}}\genn{{z_{1,2}}}\times\NN\to\KK$ by
        \begin{equation}\label{diffAutoEval:monomialDepth2SingleChainPiExtBasedAtXMinus2}
            \begin{aligned}
                \s({z_{1,1}}) &= (x-1)\,{z_{1,1}}, \qquad & \s({z_{1,2}}) &= (x-1)\,{z_{1,1}}\,{z_{1,2}}, \qquad \text{ and } \\
                \ev(z_{1,1}, n) &= \myProduct{k}{3}{n}{(k-2)}, \qquad & \ev(z_{1,2}, n) &= \myProduct{k}{3}{n}{\myProduct{j}{3}{k}{(j-2)}}.
            \end{aligned}
        \end{equation}
        \item\manuallabel{item:monomialDepth1SingleChainPiExtBasedAtXPlus1Over24}{(8)} Define the single chain \pE-extension $\dField{\KK(x)\genn{{z_{2,1}}}}$ of $\dField{\KK(x)}$ over $\KK(x)$ with $\left(x+\tfrac{1}{24}\right)$ as its base together with the evaluation function $\ev:\KK(x)\genn{{z_{2,1}}}\times\NN\to\KK$ by 
        \begin{align}\label{diffAutoEval:monomialDepth1SingleChainPiExtBasedAtXPlus1Over24}
            \s({z_{2,1}}) = \left(x+\tfrac{25}{24}\right)\,{z_{2,1}}, && \text{ and } &&  \ev(z_{1,1}, n) = \myProduct{k}{3}{n}{\left(k+\tfrac{1}{24}\right)}.
        \end{align}
    \end{enumerate}
    Putting everything together, we have constructed the multiple chain \apE-extension $\dField{\AA}$ of $\dField{\KK(x)}$ with 
    \begin{equation}\label{dom:multipleChainAPRing}
        \AA = \KK(x)\genn{{\vartheta_{1,1}}}\genn{{\vartheta_{1,2}}}\genn{{y_{1,1}}}\genn{{y_{2,1}}}\genn{{y_{2,2}}}\genn{{y_{3,1}}}\genn{{y_{4,1}}}\genn{{y_{4,2}}}\genn{y_{5,1}}\genn{{z_{1,1}}}\genn{{z_{1,2}}}\genn{{z_{2,1}}}
    \end{equation}
    based at $\left(-1,\,\sqrt{3},\,2,\,2,\,3,\,5,\,5,\,25,\,x-2,\,x-2,\,x+\tfrac{1}{24}\right)$
    where $(1,\,2,\,1,\,1,\,2,\,1,\,1,\,2,\,1,\,1,\,2,\,1)$ is the extension depth. In this ring, the geometric product $\tilde{G}(n)$ and the hypergeometric product $\tilde{H}(n)$ defined in~\eqref{eqn:geometricProducts} and~\eqref{eqn:hypergeometricProductInShifCoPrimeRepForm} are modelled by 
    \begin{align}\label{eqn:hyperGeometricProductRepresentationInRingA}
        g = \frac{\vartheta_{1,1}\,y_{3,1}\,y_{5,1}\,\vartheta_{1,2}\,y_{2,2}}{y_{1,1}\,y_{2,1}\,y_{4,2}}  && \text{ and } && h = {z_{1,1}}^{3}\,{z_{2,1}}\,{z_{1,2}}
    \end{align}
    respectively. That is, $\tilde{G}(n)=\ev(g,\,n)$ holds for all $n\ge1$ and $\tilde{H}(n)=\ev(h,\,n)$ holds for all $n\ge2$. As a consequence, the indefinite hypergeometric product expression $\tilde{P}(n)$ defined in~\eqref{eqn:shiftCoPrimeRepresentationOfProducts} is modelled by the expression
    \begin{equation}\label{eqn:productRepresentationInSimpleAExt}
        \tilde{p} = \frac{254\,(x-1)^{3}\,x\,(x+1)\,(x+2)}{432}\,g\,h\in\AA.
    \end{equation}
    This means that $\tilde{P}(n)=\ev(\tilde{p},\,n)$ holds for all $n\in\NN$ with $n\ge2$.
\end{example}

\begin{remark}\label{remk:arbitrarySimplePExt2MultipleChainPExt}
\normalfont
We can carry out the following construction to model the given hypergeometric products  $P_{1}(n),\dots,P_{e}(n)\in\Prod(\KK(x))$ of finite nesting depth and an expression in terms of these products in a \pE-extension\footnote{Similarly, one can carry out such a construction for \aE-extensions or \apE-extensions by checking in addition if the arising multiplicands are built by roots of unity.}. Here we start with the difference field $\dField{\KK(x)}$ given by $\sigma|_{\KK}=\id$ and $\sigma(x)=x+1$ which is equipped with the evaluation function $\ev$ given in~\eqref{map:evaForRatFxns} and the zero-function~\eqref{eqn:hyperGeoShiftBoundedFxns}. 
    \begin{enumerate}[\hspace*{1em}(1)]
        \item \manuallabel{item1:arbitrarySimplePExt2MultipleChainPExt}{(1)} For $1 \le i \le e$, rewrite $P_{i}(n)$ such that it is composed multiplicatively by products in factored form and such that all products are $\delta$-refined (in its strongest form, one may use the representation given in Proposition~\ref{pro:preprocessingNestedHypergeometricProductsExtended}). Let $\cali{P}$ be the set of all products that occur in the rewritten $P_{i}(n)$.
        \item\manuallabel{item2:arbitrarySimplePExt2MultipleChainPExt}{(2)} Among all nested products of $\cali{P}$  
        with the same multiplicand representation $c(x)\in\KK(x)$, take one of the products, say $F(n)$, with the highest nesting depth $m$. Construct the corresponding single chain \pE-extension $\dField{\KK(x)\langle p_1\rangle\dots\langle p_m\rangle}$ of $\dField{\KK(x)}$ over $\KK(x)$ and extend the evaluation function accordingly such that the outermost \pE-monomial $p_m$ models $F(n)$ with $\ev(p_{\mu},n)=F(n)$ for all $n\geq\delta$. In particular, any arising product with the same multiplicand representation and depth $\mu<m$ is modelled by $p_{\mu}$. Thus we can remove all the products from $\cali{P}$ which have the same multiplicand $c(x)$.
        \item \manuallabel{item3:arbitrarySimplePExt2MultipleChainPExt}{(3)} Repeat step~\ref{item2:arbitrarySimplePExt2MultipleChainPExt} for the remaining elements of $\cali{P}$.
        \item\manuallabel{item4:arbitrarySimplePExt2MultipleChainPExt}{(4)} Combine the constructed single chain \pE-extensions of $\dField{\KK(x)}$ over $\KK(x)$ to obtain a multiple chain \pE-extension $\dField{\AA}$ of $\dField{\KK(x)}$ over $\KK(x)$. In addition, combine the evaluation functions to one extended version.
    \end{enumerate}
In this way, all products arising in the rewritten $P_{1}(n),\dots,P_{e}(n)$ can be modelled by a \pE-monomial within the constructed difference ring $\dField{\AA}$. Furthermore, consider any expression $A(n)$ of the form~\eqref{Equ:ProdEDef} in terms of the original products $P_{1}(n),\dots,P_{e}(n)$ and let $\lambda=\max\{Z(a_{\bs{v}})\mid \bs{v}\in S\}$, i.e., the evaluation $a_{\bs{v}}(n)$ does not introduce poles for any $n\in\NN$ with $n\geq\lambda$. Then replacing the rewritten products in $A(n)$ and afterwards replacing the involved products  by the corresponding \pE-monomials yields a given $a\in\AA$ with $\ev(a,n)=A(n)$ for all $n\geq\max(\delta,\lambda)$. In particular, if $\KK$ is computable and the zero-function $Z$ is computable, this construction can be given explicitly.
\end{remark}

\section{A refined {\dr} approach: \rpiE-extensions}\label{sec:refinedApproachRPiExtns}
In general, the naive construction of an (ordered) multiple chain \pE-extension $\dField{\AA}$ of $\dField{\KK(x)}$ following Remark~\ref{remk:arbitrarySimplePExt2MultipleChainPExt} or a slightly refined construction of an \apE-extension like in Example~\ref{exa:singleChainAPExtensions} introduce algebraic relations among the monomials. 
In order to tackle~\ref{prob:ProblemRPE} above, we will refine \apE-extensions further to the class of \rpiE-extensions. In this regard, the set of constants
\[
\const{\dField{\AA}} = \{ c\in\AA \,|\,\s(c)=c \}
\]
of a difference ring (field) $\dField{\AA}$ plays a decisive role. In general it forms a subring (subfield) of $\AA$ which contains the rational numbers $\QQ$ as subfield. In this article, $\const{\dField{\AA}}$ will always be a field also called the \emph{constant field} of $\dField{\AA}$, which we will also denote by $\KK$. We note further that one can decide if $c\in\AA$ is a constant if $\dField{\AA}$ is computable.

We are now ready to refine \apE-extensions as follows.

\begin{definition}
    \normalfont Let $\dField{\AA\langle t \rangle}$ be an \aE-/\pE-/\apE-extension of $\dField{\AA}$. Then it is called an \emph{\rE-/\piE-/\rpiE-extension} if $\const\dField{\AA\langle t \rangle} = \const\dField{\AA}$ holds. Depending on the type of extension, we call the generator $t$ an \emph{\rE-/\piE-/\rpiE-monomial}, respectively. A (nested) \aE-/\pE-/\apE-extension $\dField{\AA\langle t_{1} \rangle\dots\langle t_{e} \rangle}$ of $\dField{\AA}$ with $\const\dField{\AA\langle t_{1} \rangle\dots\langle t_{e} \rangle}=\const\dField{\AA}$ is also called a \emph{(nested) \rE-/\piE-/\rpiE--extension}.
\end{definition}

Given an \aE-/\pE-/\apE-extension, there exist criteria that enable one able to check if it is an \rE-/\piE-/\rpiE-extension. We refer the reader to see~\cite[Theorem~2.12]{schneider2016difference} for further details and proofs. 
\begin{theorem}\label{thm:rpiSExtCriterion}
    Let $\dField{\AA}$ be a \dr. Then the following statements hold.
    \begin{enumerate} 
        \item\manuallabel{item:piExtCriterion}{(1)} Let $\dField{\AA[t,\tfrac{1}{t}]}$ be a \pE-extension of $\dField{\AA}$ with $\s(t)=\alpha\,t$ where $\alpha\in\AA^{*}$. Then this is a \piE-extension (i.e., $\const{\dField{\AA[t,\tfrac{1}{t}]}} = \const{\dField{\AA}}$) iff there are no $g\in\AA\sm\zs$ and $v\in\ZZ\sm\zs$ with $\s(g)=\alpha^{v}\,g$. 
        \item\manuallabel{item:rExtCriterion}{(2)} Let $\dField{\AA[\vartheta]}$ be an \aE-extension of $\dField{\AA}$ of order $\lambda>1$ with $\s(\vartheta)=\zeta\,\vartheta$ where $\zeta\in\AA^{*}$. Then this is an \rE-extension (i.e., $\const{\dField{\AA[\vartheta]}} = \const{\dField{\AA}}$) iff there are no $g\in\AA\sm\zs$ and $v\in\{1,\dots,\lambda-1\}$ with $\s(g)=\zeta^{v}\,g$. If it is an \rE-extension, $\zeta$ is a primitive $\lambda$-th root of unity.
    \end{enumerate} 
\end{theorem}

We remark that the above definitions and also Theorem~\ref{thm:rpiSExtCriterion} are inspired by Karr's $\Pi\Sigma$-field extensions~\cite{karr1981summation,schneider2001symbolic}. Since we will use this notion later (see Theorem~\ref{thm:shiftCoPrimeAsPiExtension} and~\ref{thm:orderedMultipleChainPExt2OrderedMultipleChainPiExt4HyperAndqHyperProducts} below) we will introduce them already here.

\begin{definition}
	\normalfont Let $\dField{\FF(t)}$ be a difference field extension of a difference field $\dField{\FF}$ with $t$ transcendental over $\FF$ and $\sigma(t)=\alpha\,t+\beta$ with $\alpha\in\FF^*$ and $\beta\in\FF$. This extension is called a \emph{\sigmaE-field extension} if $\alpha=1$ and $\const\dField{\FF(t)}=\const\dField{\FF}$, and it is called a \emph{$\Pi$-field extension} if $\beta=0$ and $\const\dField{\FF(t)}=\const\dField{\FF}$. A difference field $\dField{\KK(t_1)\dots(t_e)}$ is called a \emph{\pisiE-field over $\KK$} if $\dField{\KK(t_1)\dots(t_i)}$ is a \pisiE-extension of $\dField{\KK(t_1)\dots(t_{i-1})}$ for $1\leq i\leq e$ with $\const\dField{\KK}=\KK$.
\end{definition}

Throughout this article, our base case {\df} is $\dField{\KK(x)}$ with the automorphism $\s(x)=x+1$ and $\sigma|_{\KK}=\id$ which in fact is a \sigmaE-extension of $\dField{\KK}$, i.e., $\const\dField{\KK(x)}=\KK$. In particular $\dField{\KK(x)}$ is a \pisiE-field over $\KK$. We conclude this subsection by observing that the check if an \aE-extension is an \rE-extension (see part~2 of Theorem~\ref{thm:rpiSExtCriterion}) is not necessary if the ground field is a \pisiE-field; compare~\cite[Lemma 2.1]{ocansey2018representing}.
\begin{lemma}\label{lem:aExtOverPiSigmaExtIsRExt}
    Let $\dField{\FF}$ be a \pisiSE-field over $\KK$. Then any \aE-extension $\dField{\FF[\vartheta]}$ of $\dField{\FF}$ with order $\lambda>1$ is an \rE-extension.
\end{lemma}

\subsection{Embedding into the ring of sequences}\label{subsec:embedRPiExtsIntoRingOfSeqns}
In this subsection, we will discuss the connection between \rpiE-extensions and the {\dr} of sequences. More precisely, we will show how \rpiE-extensions can be embedded into the difference ring of sequences~\cite{schneider2017summation}; compare also \cite{put1997galois}. This feature will enable us to handle condition \ref{item2:ProblemRPE} of~\ref{prob:ProblemRPE} in the sections below.
 
\begin{definition}\label{defn:DiffRingHomOrMon}
\normalfont
    Let $\dField{\AA}$ and $\dField[\sigma^{\prime}]{\AA^{\prime}}$ be two {\dr s}. The map $\tau: \AA \to \AA^{\prime}$ is called a \emph{{\dr} homomorphism} if $\tau$ is a ring homomorphism, and for all $f \in \AA$, $\tau(\s(f)) = \s^{\prime}(\tau(f))$. If $\tau$ is injective, then it is called a \emph{{\dr} monomorphism} or a \emph{{\dr} embedding}.\footnote{In this case, $\dField{\tau(\AA)}$ is a sub-difference ring of $\dField[\sigma^{\prime}]{\AA^{\prime}}$ where $\dField{\AA}$ and $\dField{\tau(\AA)}$ are the same up to renaming with respect to $\tau$.} If $\tau$ is a bijection, then it is a \emph{{\dr} isomorphism} and we say $\dField{\AA}$ and $\dField[\sigma^{\prime}]{\AA^{\prime}}$ are isomorphic; we write $\dField{\AA}\simeq\dField[\sigma^{\prime}]{\AA^{\prime}}$. Let $\dField{\EE}$ and $\dField[\tilde{\s}]{\tilde{\EE}}$ be {\drE s} of $\dField{\AA}$. Then a {\dr}-homomorphism/isomorphism/monomorphism $\tau:\EE\to\tilde{\EE}$ is called an \emph{$\AA$-homomorphism}/\emph{$\AA$-isomorphism}/\emph{$\AA$-monomorphism}, if $\tau|_{\AA}=\id$. \\
    Let $\dField{\AA}$ be a {\dr} with constant field $\KK$. A {\dr} homomorphism (resp.\ monomorphism) $\tau:\AA \to \ringOfEquivSeqs$ is called $\KK$-\emph{homomorphism} (resp.\ -\emph{monomorphism}) if for all $c \in \KK$ we have that $\tau(c) = {\bs c} := \geno{c,c,c,\dots\,}$.
\end{definition}

The following lemmas provide the key property that will enable us to embed \rpiE-extensions into the ring of sequences. First, we recall that the evaluation function of a difference ring establishes naturally a $\KK$-homomorphism. More precisely, by \cite[Lemma {2.5.1}]{schneider2001symbolic} we get

\begin{lemma}\label{lem:evaluationHomomorphism}
    Let $\dField{\AA}$ be a {\dr} with constant field $\KK$. Then the map  $\tau:\AA\to\ringOfEquivSeqs$ is a $\KK$-homomorphism if and only if there is an evaluation function $\ev:\AA\times\NN\to\KK$ for $\dField{\AA}$ (see Definition~\ref{defn:modelSeqsInDiffRings}) with 
    $\tau(f) = \genn{\ev(f,0),\,\ev(f,1),\dots}$.
\end{lemma}

Starting with our \pisiE-field $\dField{\KK(x)}$ over $\KK$ and the evaluation function~\eqref{map:evaForRatFxns} we can construct for an \apE-extension an appropriate evaluation function by iterative application of Lemma~\ref{Lemma:ExtendEv}. In particular this yields a $\KK$-homomorphism from the given \apE-extension into the ring of sequences by Lemma~\ref{lem:evaluationHomomorphism}. Finally, we utilize the following result from~\cite{schneider2017summation}; compare~\cite[Lemma 2.2]{ocansey2018representing}. 
\begin{theorem}\label{Thm:injectiveHom}
Let $\dField{\AA}$ be a {\df} with constant field $\KK$ and let $\dField{\EE}$ be a basic \rpiE-extension of $\dField{\AA}$, Then any $\KK$-homomorphism $\tau:\EE\to\ringOfEquivSeqs$ is injective.
\end{theorem}

In other words, if we succeed in modeling our nested products  within a basic \rpiE-extension (in particular, a multiple chain \rpiE-extension) over $\dField{\KK(x)}$ with an appropriate evaluation function, then we automatically obtain a $\KK$-embedding.

\subsection{A structural theorem for multiple chain \piE-extensions}\label{subsec:structuralTheoremForMultipleChainPiExtns}

In part~1 of Theorem~\ref{thm:rpiSExtCriterion} a criterion is given that enables one to check with, e.g., the algorithms from~\cite{karr1981summation,schneider2016difference} whether a \pE-extension is a \piE-extension. In~\cite[Lemma 5.1]{ocansey2018representing} (based on~\cite{schneider2010parameterized,schneider2017summation}) this criterion has been generalized for ``single nested'' \pE-extensions as follows.

\begin{lemma}\label{lem:transcendentalCriterionForPrdts}
	Let $\dField{\FF}$ be a {\df} and let $\Lst{f}{1}{s}\in\FF^{*}$. Then the following statements are equivalent. 
	\begin{enumerate}[(1)]
		\item There do not exist $(\Lst{v}{1}{s}) \in \ZZ^{s}\sm\zvs{s}$ and $g\in\FF^{*}$ such that $\frac{\s(g)}{g}=\ProdLst{f}{1}{s}{v}$ holds.
		\item The \pE-extension $\dField{\FF[z_{1},z^{-1}_{1}]\dots[z_{s},z^{-1}_{s}]}$ of $\dField{\FF}$ with $\s(z_{i})=f_{i}\,z_{i}$ for $1\le i\le s$ is a \piE-extension.
		\item The difference field extension $\dField{\FF(z_{1})\dots(z_{s})}$ of $\dField{\FF}$ with $z_i$ transcendental over $\FF(z_1)\dots(z_{i-1})$ and $\s(z_{i})=f_{i}\,z_{i}$  for $1\le i\le s$ is a \piE-field extension.
	\end{enumerate}
\end{lemma}

In Theorem~\ref{thm:multipleChainPExt2MultipleChainPiExt} we will extend this result further to multiple-chain \pE-extensions. Here we utilize that a solution for a certain class of homogeneous first-order difference equations has a particularly simple form; this result is a specialization of~\cite[Corollary 4.6]{schneider2016difference}.

\begin{corollary}\label{cor:semiConstStructure4PiSigmaExtension}
	Let $\dField{\EE}$ be a \piE-extension of a {\df} $\dField{\AA}$ with $\EE=\AA\genn{t_{1}}\dots\genn{t_{e}}$. 
	Then for any $g\in\EE\setminus\{0\}$ with $\s(g)=u\,\ProdLst{t}{1}{e}{z}\,g$ for some $u\in\AA^*$ and $z_i\in\ZZ$ we have
	$$g=h\,\ProdLst{t}{1}{e}{v}$$
	with $h\in\AA^{*}$ and $v_{i}\in\ZZ$.
\end{corollary}

Now we are ready to prove a general criterion that enables one to check if a multiple chain \pE-extension forms a \piE-extension. This result will be heavily used within the next two sections.

\begin{theorem}\label{thm:multipleChainPExt2MultipleChainPiExt}
    Let $\dField{\HH}$ be a {\df} and let $\dField{\HH_{\ell}}$ with $\HH_{\ell}=\HH\genn{t_{\ell,1}}\dots\genn{t_{\ell,s_{\ell}}}$ for $1\le \ell \le m$ be single chain \pE-extensions of $\dField{\HH}$ over $\HH$ with base $c_{\ell}\in\HH^{*}$ where $s_1\geq s_2\geq\dots\geq s_m$. In particular, the automorphisms are given by 
    \begin{equation}\label{diffAuto:singleChainPiMonomialsAuto}
    \s(t_{\ell,k}) = \alpha_{\ell,k}\,t_{\ell,k} \quad \text{where}\quad \alpha_{\ell,k}=c_{\ell}\,t_{\ell,1}\cdots t_{\ell,k-1}\in(\HH^{*})_{\HH}^{\HH\genn{t_{\ell,1}}\dots\genn{t_{\ell,k-1}}}.
    \end{equation}
    Let $\dField{\AA}$ be the ordered multiple chain \pE-ext.\ of $\dField{\HH}$ with $\AA=\HH\genn{t_{1,1}}\dots\genn{t_{w_{1},1}}\dots\genn{t_{1,d}}\dots\genn{t_{w_{d},d}}$ of monomial depth $d:=\max(s_{1},\dots,s_{m})$ with $m=w_{1}\ge w_{2}\ge\dots\ge w_{d}$ composed by the single chain \piE-extensions $\dField{\HH_{\ell}}$ of $\dField{\HH}$ with the automorphism~\eqref{diffAuto:singleChainPiMonomialsAuto}. Then $\dField{\AA}$ is a \piE-extension of $\dField{\HH}$ if and only if \footnote{Note that $(\Lst{c}{1}{m})=(\alpha_{1,1},\dots,\alpha_{w_{1},1})$.} there does not exist a $g\in\HH^*$ and $(\Lst{v}{1}{m})\in\ZZ^{m}\sm\zvs{m}$ such that 
    \begin{equation}\label{eqn:noRelationCriterion}
        \frac{\s(g)}{g}=\ProdLst{c}{1}{m}{v}.
    \end{equation}
\end{theorem}

\begin{proof}     
    \noindent$``\Longrightarrow"$ Suppose that $\dField{\AA}$ is a \piE-extension of $\dField{\HH}$. Then, it is a tower of \piE-extensions $\dField{\AA_{i}}$ of $\dField{\HH}$ where $\AA_{i}=\AA_{i-1}\genn{t_{1,i}}\dots\genn{t_{w_{i},i}}$ for $1\le i\le d$ with $\AA_{0}=\HH$. Since $\dField{\AA_{1}}$ is a \piE-extension of $\dField{\HH}$, it follows by Lemma~\ref{lem:transcendentalCriterionForPrdts} that there does not exist a $g\in\HH$ and $(\Lst{v}{1}{w_{1}})\in\ZZ^{w_{1}}\sm\zvs{w_{1}}$ with $w_{1}=m$ such that~\eqref{eqn:noRelationCriterion} holds.
    
    \noindent $``\Longleftarrow"$ Conversely, suppose that there does not exist  a $g\in\HH$ and $(\Lst{v}{1}{w_{1}})\in\ZZ^{w_{1}}\sm\zvs{w_{1}}$ with $w_{1}=m$ such that~\eqref{eqn:noRelationCriterion} holds. Let $\dField{\AA_{1}}$ with $\AA_{1}=\HH\genn{t_{1,1}}\dots\genn{t_{w_{1},1}}$ be a \pE-extension of $\dField{\HH}$ with $\s(t_{j,1})=\alpha_{j,1}\,t_{j,1}$ for all $1\le j \le w_{1}$. By Lemma~\ref{lem:transcendentalCriterionForPrdts}, $\dField{\AA_{1}}$ is a \piE-extension of $\dField{\HH}$. Let $\dField{\AA_{i}}$ with  $\AA_{i}=\AA_{i-1}\genn{t_{1,i}}\dots\genn{t_{w_{i},i}}$ be the multiple chain \pE-extension of $\dField{\HH}$ 
    with $\myd(t_{1,i})=\dots=\myd(t_{w_{i},i})$ for all $1\le i \le d$ with the automorphism~\eqref{diffAuto:singleChainPiMonomialsAuto}. Assume that $\dField{\AA_{k}}$ is a \piE-extension of $\dField{\HH}$ for all $1\le k\leq\delta$ with $d>\delta\geq1$ and that $\dField{\AA_{\delta+1}}$ is not a \piE-extension of $\dField{\AA_{\delta}}$. Then by Lemma~\ref{lem:transcendentalCriterionForPrdts}, we can take a $g\in\AA_{\delta}\sm\zs$ and $(\upsilon_{1},\Lst{\upsilon}{2}{w_{\delta+1}})\in\ZZ^{w_{\delta+1}}\sm\zvs{w_{\delta+1}}$ such that 
    \begin{equation}\label{eqn:shiftOfG}
    \s(g)=\alpha_{1,\delta+1}^{\upsilon_{1}}\,\alpha_{2,\delta+1}^{\upsilon_{2}}\cdots\,\alpha_{w_{\delta+1},\delta+1}^{\upsilon_{w_{\delta+1}}}\,g
    \end{equation}
    holds. By Corollary~\ref{cor:semiConstStructure4PiSigmaExtension}, it follows that
    $g=h\,t_{1,1}^{v_{1,1}}\,t_{2,1}^{v_{2,1}}\cdots \,t_{w_{1},1}^{v_{w_{1},1}}\cdots\, t_{1,\delta}^{v_{1,\delta}}\,t_{2,\delta}^{v_{2,\delta}}\cdots\, t_{w_{\delta},\delta}^{v_{w_{\delta},\delta}}$
    with $h\in\HH^{*}$ and $v_{i,j}\in\ZZ$. For the left hand side of~\eqref{eqn:shiftOfG} we have that
    \[
    \s(g) = \gamma\, t_{1,\delta}^{v_{1,\delta}}\,t_{2,\delta}^{v_{2,\delta}}\cdots\, t_{w_{\delta},\delta}^{v_{w_{\delta},\delta}}
    \]
    where $\gamma\in\AA_{\delta-1}$ and for the right hand side of~\eqref{eqn:shiftOfG} we have that 
    \[
    \alpha_{1,\delta+1}^{\upsilon_{1}}\,\alpha_{2,\delta+1}^{\upsilon_{2}}\cdots\,\alpha_{w_{\delta+1},\delta+1}^{\upsilon_{w_{\delta+1}}}\,g = \omega\,t_{1,\delta}^{v_{1,\delta}+\upsilon_{1}}\,t_{2,\delta}^{v_{2,\delta}+\upsilon_{2}}\cdots\, t_{w_{\delta+1},\delta}^{v_{w_{\delta+1},\delta}+\upsilon_{w_{\delta+1}}}t_{w_{\delta+1}+1,\delta}^{v_{w_{\delta+1}+1,\delta}}\dots t_{w_{\delta},\delta}^{v_{w_{\delta},\delta}}
    \]
    where $\omega\in\AA_{\delta-1}$. Consequently,
    $v_{k,\delta} = v_{k,\delta} + \upsilon_{k}$
    and thus $\upsilon_{k}=0$ for all $1\le k \le w_{\delta+1}$ which is a 
    contradiction to the assumption that $(\Lst{\upsilon}{1}{w_{\delta+1}})\neq\zv{w_{\delta+1}}$. Thus $\dField{\AA_{d}}$ is a \piE-extension of $\dField{\HH}$.
\end{proof}

\section{The main building blocks to represent nested products in \rpiE-extensions}\label{Sec:ProductsInRPiExt}
Suppose that we are given a finite set of hypergeometric products of finite nesting depth which have been brought into the form as given in Proposition~\ref{pro:preprocessingNestedHypergeometricProductsExtended}. In the following we will show in Sections~\ref{subsec:HypProducts} and~\ref{subsec:simpleRPiExt4GeoPrdts} how these hypergeometric and geometric products can be modelled in \rpiE-extensions. For the treatment of geometric products one has to deal in addition with products defined over roots of unity of finite nesting depth. This extra complication will be treated in Section~\ref{subsec:rootsOfUnity}. Finally, in Section~\ref{Sec:CompleteAlg} below we will combine all these techniques to represent the full class of hypergeometric products of finite nesting depth in \rpiE-extensions.

\subsection{Nested Hypergeometric products with shift-coprime multiplicands}\label{subsec:HypProducts}

In Proposition~\ref{pro:preprocessingNestedHypergeometricProductsExtended} we showed that a finite set of hypergeometric products of finite nesting depth can be brought in a shift-coprime product representation form. In the setting of \pisiE-field extensions the underlying Definition~\ref{defn:shiftCoPrimeAndShiftEquivalent} can be generalized as follows.

\begin{definition}
\normalfont
Let $\dField{\FF(t)}$ be a \pisiE-field extension of $\dField{\FF}$. We call two polynomials $f,g\in\FF[t]$ shift-coprime (or $\sigma$-coprime)  if for all $k\in\ZZ$ we have that $\gcd(f,\sigma^k(h))=1$. 
\end{definition}

Inspired by~\cite{schneider2005product,Schneider:14} we refined Lemma~\ref{lem:transcendentalCriterionForPrdts} in~\cite[Theorem~5.3]{ocansey2018representing} to the following result that was the key tool to represent hypergeometric products of nesting depth 1 in a \piE-extension.

\begin{theorem}\label{thm:shiftCoPrimeAsPiExtension}
    Let $\dField{\FF(t)}$ be a \pisiSE-extension of $\dField{\FF}$. Let $\Lst{f}{1}{s}\in\FF[t]\sm\FF$ be irreducible monic polynomials. Then the following statements are equivalent. 
    \begin{enumerate}[(1)]
        \item For all $i,j$ with $1\le i < j \le s$, $f_{i}$ and $f_{j}$ are shift-coprime.
        \item There does not exist $(\Lst{v}{1}{s})\in\ZZ^{s}\sm\zvs{s}$ and $g\in\FF(t)^{*}$ with 
        $\frac{\s(g)}{g} = \ProdLst{f}{1}{s}{v}$.
        \item The \pE-extension $\dField{\FF(t)[z_{1},z^{-1}_{1}]\dots[z_{s},z^{-1}_{s}]}$ of $\dField{\FF(t)}$ with $\s(z_{i})=f_{i}\,z_{i}$ for $1\le i\le s$ is a \piE-extension.
    \end{enumerate}
\end{theorem}

With this result we are now in the position to refine Theorem~\ref{thm:multipleChainPExt2MultipleChainPiExt} in order to construct a \piE-extension in which we can model hypergeometric products of finite nesting depth that are in shift-coprime representation form.

\begin{theorem}\label{thm:orderedMultipleChainPExt2OrderedMultipleChainPiExt4HyperAndqHyperProducts}
    Let $\dField{\FF(t)}$ be a \pisiSE-extension of $\dField{\FF}$
    Let ${\bs f}=(\Lst{f}{1}{m})\in(\FF[t]\sm\FF)^{m}$ be irreducible monic polynomials. For all $1\le \ell\le m$, let $\dField{\FF_{\ell}}$ with $\FF_{\ell}:=\FF(t)\genn{z_{\ell,1}}\dots\genn{z_{\ell,s_{\ell}}}$ be a single chain \pE-extension of $\dField{\FF(t)}$ with base  $f_{\ell}\in\FF[t]\sm\FF$  with the automorphism
    \begin{equation}\label{diffAuto:singleChainPExtensionOverPiSigmaField}
    \s(z_{\ell,k}) = \alpha_{\ell,k}\,z_{\ell,k} \quad \text{ where } \quad \alpha_{\ell,k} = f_{\ell}\,z_{\ell,1}\cdots\,z_{\ell,k-1}\in(\FF^{*})_{\FF}^{\FF\genn{z_{\ell,1}}\dots\genn{z_{\ell,k-1}}}
    \end{equation}
    and with $s_1\geq s_2\geq\dots\geq s_m$.
    Let $\dField{\HH_{b}}$ with 
    \begin{equation*}\label{dom:orderedMultipleChainPiRing}
    \HH_{b}=\FF(t)\genn{{\bs z}_{\bs 1}}\dots\genn{{\bs z}_{\bs b}}=\FF(t)\genn{z_{1,1}}\dots\genn{z_{w_{1},1}}\dots\genn{z_{1,b}}\dots\genn{z_{w_{b},b}}
    \end{equation*}
    be an ordered multiple chain \pE-extension of $\dField{\FF(t)}$ of monomial depth $b=\max(s_{1},\dots,s_{m})$ with bases $f_{1},\dots,f_{m}$ where $m=w_{1}\ge w_{2}\ge\cdots\ge w_{b}$ which is composed by the single chain \pE-extensions $\dField{\FF_{\ell}}$ of $\dField{\FF(t)}$. Then $\dField{\HH_{b}}$ is a \piE-extension of $\dField{\FF(t)}$ if and only if for all $i,j$ with $1\le i < j \le m$ the $f_{i}$ and $f_{j}$ are shift-coprime.
\end{theorem}

\begin{proof}
    \noindent$``\Longrightarrow"$ If $\dField{\HH_{b}}$ is a \piE-extension of $\dField{\FF(t)}$, then by Theorem~\ref{thm:multipleChainPExt2MultipleChainPiExt} there does not exist a $g\in\FF(t)^{*}$ such that $\tfrac{\s(g)}{g}=\ProdLst{f}{1}{m}{v}$ holds, and by Theorem~\ref{thm:shiftCoPrimeAsPiExtension} for all $i,\,j$ with $1\le i < j \le m$, $f_{i}$ and $f_{j}$ are shift-coprime.
    
    \noindent $``\Longleftarrow"$ Conversely, if for all $i,\,j$ with $1\le i < j \le m$, $f_{i}$ and $f_{j}$ are shift-coprime, then by Theorem~\ref{thm:shiftCoPrimeAsPiExtension} there does not exist a $g\in\FF(t)^{*}$ such that $\tfrac{\s(g)}{g}=\ProdLst{f}{1}{m}{v}$ holds, and by Theorem~\ref{thm:multipleChainPExt2MultipleChainPiExt} $\dField{\HH_{b}}$ is a \piE-extension of $\dField{\FF(t)}$.
\end{proof}

\vskip 0.5em 

Summarizing, we obtain the following crucial result.

\begin{corollary}\label{cor:reducedNormalFormAndPiExts}
    Let $\dField{\KK(x)}$ be a rational {\df} with the automorphism $\s(x)=x+1$ and the evaluation function $\ev:\KK(x)\times\NN\to\KK$ given by~\eqref{map:evaForRatFxns}. Let $\tilde{H}_{1}(n),\dots,\tilde{H}_{e}(n)$ be hypergeometric products in $\Prod(\KK(x))$ of nesting depth at most $b$ which are all in shift-coprime representation form (see Definition~\ref{defn:shiftCoPrimeProductRepresentationForm}) and which are all $\delta$-refined for some $\delta\in\NN$. Then one can construct an ordered multiple chain \piE-extension $\dField{\tilde{\HH}_{b}}$ of $\dField{\KK(x)}$ with  \begin{equation}\label{eqn:orderedMultipleChainPiRingRefinedReducedNormalFormPrdts}
    \tilde{\HH}_{b}=\KK(x)\genn{\tilde{\bs z}_{\bs 1}}\dots\genn{\tilde{\bs z}_{\bs b}}=\KK(x)\genn{\tilde{z}_{1,1}}\dots\genn{\tilde{z}_{p_{1},1}}\dots\genn{\tilde{z}_{1,b}}\dots\genn{\tilde{z}_{p_{b},b}}
    \end{equation}
    which is composed by the single chain \piE-extensions $\dField{\tilde{\FF}_{\ell}}$ of $\dField{\KK(x)}$ with $\tilde{\FF}_{\ell}=\KK(x)\genn{\tilde{z}_{\ell,1}}\dots\genn{\tilde{z}_{\ell,s_{\ell}}}$ with 
    \begin{enumerate}[(1)]
        \item the automorphism $\s:\tilde{\FF}_{\ell}\to\tilde{\FF}_{\ell}$ defined by
        \begin{equation}\label{diffAuto:singleChainPiExtForARefinedReducedNormalFormPrdt}
        \s(\tilde{z}_{\ell,k}) = \tilde{\alpha}_{\ell,k}\,\tilde{z}_{\ell,k} \quad \text{ where } \quad \tilde{\alpha}_{\ell,k} = \tilde{f}_{\ell}\,\tilde{z}_{\ell,1}\cdots\,\tilde{z}_{\ell,k-1}\in(\KK(x)^{*})_{\KK(x)}^{\KK(x)\genn{\tilde{z}_{\ell,1}}\dots\genn{\tilde{z}_{\ell,k-1}}}
        \end{equation}
         for $1 \le \ell \le p_{1}$ and $1 \le k \le s_{\ell}$  where $\tilde{f}_{\ell}\in\KK[x]\sm\KK$ is an irreducible monic polynomial, and 
        \item the evaluation function $\tilde{\ev}:\tilde{\FF}_{\ell}\times\NN\to\KK$ given by $\tilde{\ev}|_{\KK(x)}=\ev$ with~\eqref{map:evaForRatFxns} and 
        \begin{equation}\label{evMap:singleChainPiExtForARefinedReducedNormalFormPrdt}
        \tilde{\ev}(\tilde{z}_{\ell,k}, n) = \myProduct{j}{\delta}{n}{\tilde{\ev}(\tilde{\alpha}_{\ell,k}, j-1)}
        \end{equation}
        for $1 \le \ell \le p_{1}$ and $1 \le k \le s_{\ell}$
    \end{enumerate}
with the following property: for all $1\leq i\leq e$ there are $k,j$ such that
\begin{equation}\label{Equ:MapzToHi}
\ev(\tilde{z}_{k,j},n)=\tilde{H}_i(n),\quad\forall\,n\geq\max(0,\delta-1).
\end{equation}
    Furthermore, for all $g\in\tilde{\HH}_{b}$, the map $\tilde{\tau}:\tilde{\HH}_{b}\to\ringOfEquivSeqs$ defined by
    \begin{align}\label{eqn:diffRingEmbedding}
    \tilde{\tau}(g)=\funcSeqA{\tilde{\ev}(g,n)}{n}
    \end{align}
    is a $\KK$-embedding. If $\KK$ is computable,
    the above construction can be given explicitly. 
\end{corollary}
\begin{proof}
	By the procedure outlined in Remark~\ref{remk:arbitrarySimplePExt2MultipleChainPExt} (skipping step (1) since the input is already in the right form) we obtain
    the ordered multiple chain \pE-extension $\dField{\tilde{\HH}_{b}}$ of $\dField{\KK(x)}$ with~\eqref{eqn:orderedMultipleChainPiRingRefinedReducedNormalFormPrdts} and~\eqref{evMap:singleChainPiExtForARefinedReducedNormalFormPrdt} such that~\eqref{Equ:MapzToHi} holds for all $1\leq i\leq e$ for some $j,k$. Since the bases $\tilde{f}_{1},\dots,\tilde{f}_{p_{1}}$ of the single chain \pE-extensions $\dField{\tilde{\FF}_{1}},\dots,\dField{\tilde{\FF}_{p_{1}}}$ that composes $\dField{\tilde{\HH}_{b}}$ are shift-coprime, it follows  by Theorem~\ref{thm:orderedMultipleChainPExt2OrderedMultipleChainPiExt4HyperAndqHyperProducts} that $\dField{\tilde{\HH}_{b}}$ is a \piE-extension of $\dField{\KK(x)}$. Since $\dField{\tilde{\HH}_{b}}$ is a basic \piE-extension of the rational {\df} $\dField{\KK(x)}$, it follows by Theorem~\ref{Thm:injectiveHom} that the $\KK$-homomorphism $\tilde{\tau}:\tilde{\HH}_{b}\to\ringOfEquivSeqs$ defined by~\eqref{eqn:diffRingEmbedding} is injective. 
    Since $\KK$ is computable, all the above ingredients can be constructed explicitly.
\end{proof}

\begin{example}[Cont. Example~\ref{exa:singleChainAPExtensions}]\label{exa:monomialDepth2OrderedMultipleChainPiExtForShiftCoPrimePolys}
    \normalfont Consider the ordered multiple chain \pE-extension $\dField{\tilde{\HH}}$ of the rational {\df} $\dField{\KK(x)}$ of monomial depth $2$ with $\tilde{\HH}=\KK(x)\genn{z_{1,1}}\genn{z_{2,1}}\genn{z_{1,2}}$ where $\dField{\tilde{\HH}}$ is composed by the single chain \piE-extensions of $\dField{\KK(x)}$ constructed in parts~\ref{item:monomialDepth2SingleChainPiExtBasedAtXMinus2} and~\ref{item:monomialDepth1SingleChainPiExtBasedAtXPlus1Over24} of Example~\ref{exa:singleChainAPExtensions}. Since the bases of $\dField{\tilde{\HH}}$ given by $(x-2)$ and $\left(x+\tfrac{1}{24}\right)$ are shift-coprime with respect to the automorphism $\s(x)=x+1$, it follows that the ordered multiple chain \pE-extension  $\dField{\tilde{\HH}}$ of the rational {\df} $\dField{\KK(x)}$ of monomial depth $2$ is a \piE-extension. Furthermore, it follows by Theorem~\ref{Thm:injectiveHom} that the map $\tilde{\tau}:\tilde{\HH}\to\ringOfEquivSeqs$ defined by $\tilde{\tau}(f)=\genn{\tilde{\ev}(f,n)}_{n\ge0}$ for all $f\in\tilde{\HH}$ is a $\KK$-embedding where $\tilde{\ev}=\ev$ defined in~\eqref{diffAutoEval:monomialDepth2SingleChainPiExtBasedAtXMinus2}, and~\eqref{diffAutoEval:monomialDepth1SingleChainPiExtBasedAtXPlus1Over24}. In particular, for the expression $h$ given by~\eqref{eqn:hyperGeometricProductRepresentationInRingA} we have that 
    $\tilde{H}(n) = \tilde{\ev}(h, n)$
    holds for all $n\ge2$.
\end{example}

\subsection{Geometric products}\label{subsec:simpleRPiExt4GeoPrdts}
In Karr's algorithm~\cite{karr1981summation} and all the improvements~\cite{kauers2006indefinite,schneider2007simplifying,schneider2008refined,abramov2010polynomial,schneider2015fast,schneider2016difference,schneider2017summation} one relies on certain algorithmic properties of the constant field $\KK$. Among those, one needs to solve the following problem. 
\vspace*{-0.75em}
\begin{ProblemSpecBox}[\gop{P}]{ 
        {\gop{P} for $\Lst{\alpha}{1}{w} \in K^{*}$} 
    }\label{prob:ProblemGO}
    {
        Given a field $K$ and $\Lst{\alpha}{1}{w} \in K^{*}$. Compute a basis of the submodule \vspace*{-0.1cm}
        \[
        \VV := \big\{(\Lst{u}{1}{w})\in\ZZ^{w}\,\Big|\,\prodLst{i}{1}{w}{\alpha}{u}=1\big\} \text{ of } \ZZ^{w} \text{ over } \ZZ. 
        \]
        \vspace*{-0.71cm}} 
\end{ProblemSpecBox}

\noindent In our approach \ref{prob:ProblemGO} is crucial to solve~\ref{prob:ProblemRPE}, but one has to solve it not only in a given field $K$ (compare the definition of $\sigma$-computable in~\cite{schneider2005product,kauers2006indefinite}) but one must be able to solve it in any finite algebraic field extension of $K$. This gives rise to the following definition.

\begin{definition}\label{defn:stonglySigmaComputable}
\normalfont
	A field $K$ is \emph{strongly $\s$-computable} if the standard operations in $K$ can be performed, multivariate polynomials can be factored over $K$ and \ref{prob:ProblemGO} can be solved for $K$ and any finite algebraic field extension of $K$.
\end{definition}

Note that Ge's algorithm~\cite{ge1993algorithms} (see also~\cite[Algorithm 7.16, page 84]{kauers2005algorithms}) solves \ref{prob:ProblemGO} over an arbitrary number field $K$. 
Since any finite algebraic extension of an algebraic number field is again an algebraic number field, we obtain the following result; for a weaker result see~\cite[Theorem~3.5]{schneider2005product}.

\begin{lemma}
An algebraic number field $K$ is strongly $\s$-computable.
\end{lemma}
By~\cite[Theorem 5.4]{ocansey2018representing} and the consideration of~\cite[pg. 204]{ocansey2018representing} (see also~\cite[Lemma 5.2.2]{ocansey2019difference}) we 
provided an algorithm that enabled us to handle geometric products of nested depth $1$. More precisely, given a \pE-extension that models such products, Lemma~\ref{lem:depth1PExt2Depth1RPiExt4IrrConstPolysAndAlgNums} states that one can construct an \rpiE-extension in which the products can be rephrased.

\begin{lemma}\label{lem:depth1PExt2Depth1RPiExt4IrrConstPolysAndAlgNums}
	Let $\KK=K(\Lst{\kappa}{1}{u})$ be a rational function field over a field $K$ and $\dField{\KK}$ be a {\df} with $\s(c)=c$ for all $c\in \KK$. Let $\dField{\KK\genn{x_{1}}\dots\genn{x_{e}}}$ be a \pE-extension of $\dField{\KK}$ with $\s(x_{i})=\gamma_{i}\,x_{i}$ for $1\leq i\leq e$ where $\gamma_{i}\in \KK^{*}$. Let $\ev:\KK\genn{x_{1}}\dots\genn{x_{e}}\times\NN\to\KK$ be the evaluation function defined by $\ev(x_{i},n)=\gamma_{i}^{n}$ for $1\le i \le e$. Then: 
	\begin{enumerate}[\hspace*{0.1em}(1)]
		\item One can construct an \rpiE-extension $\dField{\tilde{\KK}\genn{\vartheta}\genn{\tilde{y}_{1}}\dots\genn{\tilde{y}_{s}}}$~\footnote{For concrete instances the \rE-monomial $\vartheta$ might be obsolete. In particular, if $\mu_i=0$ for all $1\leq i\leq e$ in~\eqref{diffRingHom:depth1PAndPiMonomials4IrrConstPolysAndAlgNums} it can be removed.}. of $\dField{\tilde{\KK}}$ with 
		\begin{align}\label{diffAuto:depth1PiMonomials4IrrConstPolysAndAlgNums}
		\s(\vartheta)=\zeta\,\vartheta && \text{ and } &&  \s(\tilde{y}_{k})=\alpha_{k}\,\tilde{y}_{k} 
		\end{align}
		for $1\le k \le s$ where $\tilde{\KK}=\tilde{K}(\Lst{\kappa}{1}{u})$ and $\tilde{K}$ is a finite algebraic field extension of $K$ with $\zeta\in\tilde{K}$ being a primitive $\lambda$-th root of unity and $\alpha_{k}\in\tilde{\KK}^{*}$;
		\item  one can construct the evaluation function $\tilde{\ev}:\tilde{\KK}\genn{\vartheta}\genn{\tilde{y}_{1}}\dots\genn{\tilde{y}_{s}}\times\NN\to\tilde{\KK}$ defined as 
		\begin{align}\label{evMap:depth1PiMonomials4IrrConstPolysAndAlgNums}
		\tilde{\ev}(\vartheta,n) = \zeta^{n} && \text{ and }	&& \tilde{\ev}(\tilde{y}_{k},n) = \alpha^{n}_{k};
		\end{align}
		\item one can construct a difference ring homomorphism $\varphi:\KK\genn{x_{1}}\dots\genn{x_{e}}\to\tilde{\KK}\genn{\vartheta}\genn{\tilde{y}_{1}}\dots\genn{\tilde{y}_{s}}$ with 
		\begin{align}\label{diffRingHom:depth1PAndPiMonomials4IrrConstPolysAndAlgNums}
		\varphi(x_{i}) =\vartheta^{\mu_{i}}\,\tilde{\bs y}^{{\bs v}_{\bs i}} =  \vartheta^{\mu_{i}}\,\tilde{y}^{v_{i,1}}_{1}\cdots\tilde{y}^{v_{i,s}}_{s}
		\end{align}
		for $1\leq i\leq e$ where $0\le\mu_{i}<\lambda$ and $v_{i,k}\in\ZZ$ for $1\leq k\leq s$ 
	\end{enumerate}
     such that for all $f\in\KK\genn{x_{1}}\dots\genn{x_{e}}$ and for all $n\in\NN$, 
    \[
    \ev(f,n)=\tilde{\ev}(\varphi(f),n)
    \]
    holds.
	If $K$ is strongly $\s$-computable, then the above constructions are computable.
\end{lemma}

Using this result we will derive an extended version in Lemma~\ref{lem:orderedMultipleChainPExtOverConstPoly2APiExtOverConstField} that deals with the class of ordered multiple chain \apE-extensions that models geometric products of arbitrary but finite nesting depth. 

In the following let $m\in\ZZ_{\geq1}$, and for $1\le \ell\le m$ let $\dField{\KK_{\ell}}$ with $\KK_{\ell}=\KK\genn{{\bs y}_{\ell}}=\KK\genn{y_{\ell,1}}\dots\genn{y_{\ell,s_{\ell}}}$ be a single chain \pE-extension of $\dField{\KK}$ with base $h_{\ell}\in\KK^{*}$ where  
\begin{equation}\label{diffAuto:singleChainPExtensionOverConstField}
    \s(y_{\ell,i}) = \alpha_{\ell,i}\,y_{\ell,i} \quad \text{ with } \quad \alpha_{\ell,i} = h_{\ell}\,y_{\ell,1}\cdots\,y_{\ell,i-1}\in(\KK^{*})_{\KK}^{\KK\genn{y_{\ell,1}}\dots\genn{y_{\ell,i-1}}}.
\end{equation}
In particular, we assume that $s_1\geq s_2\geq\dots\geq s_m$. Let $\ev:\KK_{\ell}\times\NN\to\KK$ be the evaluation function defined by 
\begin{equation}\label{evMap:singleChainPiMonomialsOverConstField}
    \ev(y_{\ell,i}, n) = \myProduct{j}{1}{n}{\ev(\alpha_{\ell,i}, j-1)} = \myProduct{j}{1}{n}{\alpha_{\ell,i}};
\end{equation}
in particular, for all $c\in\KK$ and $n\geq0$ we set $\ev(c,n)=c$. Let $\dField{\AA}$ be the multiple chain \pE-extension of $\dField{\KK}$ built by the single chain \piE-extensions $\dField{\KK_{\ell}}$ of $\dField{\KK}$ over $\KK$. That is, 
\[
    \AA=\KK\genn{{\bs y}_{\bs 1}}\genn{{\bs y}_{\bs 2}}\dots\genn{{\bs y}_{\bs m}} = \KK\genn{y_{1,1}}\dots\genn{y_{1,s_{1}}}\genn{y_{2,1}}\dots\genn{y_{2,s_{2}}}\dots\genn{y_{m,1}}\dots\genn{y_{m,s_{m}}}.
\]
We emphasize that all the $y_{\ell,i}$ model $1$-refined geometric products in product factored form of an arbitrary but finite nesting depth.
Depending on the context, ${\bs y}_{\ell}$ denotes $(y_{\ell,1},\dots,y_{\ell,s_{\ell}})$ or  $y_{\ell,1},\dots,y_{\ell,s_{\ell}}$ or  $y_{\ell,1}\cdots\,y_{\ell,s_{\ell}}$. Note that the \pE-monomials $y_{\ell,i}$ can be ordered in increasing order of their depths. Namely, take the depth function $\myd:\AA\to\NN$ over $\KK$ of $\dField{\AA}$ and let $d=\max(s_{1},\Lst{s}{2}{m})$ be the maximal depth. Then taking $\AA_{0}=\KK$ we can consider the tower of \pE-extensions $\dField{\AA_{i}}$ of $\dField{\AA_{i-1}}$ with
\[
\AA_{i}=\AA_{i-1}\genn{{\bs y}_{\bs i}} =\AA_{i-1}\genn{y_{1,i}}\genn{y_{2,i}}\dots\genn{y_{w_{i},i}}
\]
for $1\le i \le d$ where $m=w_{1}\ge w_{2}\ge\cdots\ge w_{d}$ and with the automorphism~\eqref{diffAuto:singleChainPExtensionOverConstField}
for $1\le \ell \le w_{i}$. In this way, the \pE-monomials at the $i$-th extension have the depth $\myd(y_{1,i})=\myd(y_{2,i})=\dots\myd(y_{w_i,i})=i$. Further, the ring $\AA_{d}$ is isomorphic to $\AA$ up to reordering of the \pE-monomials. In particular, $\dField{\AA_{d}}$ is an \emph{ordered multiple chain \pE-extension} of $\dField{\KK}$ of monomial depth at most $d$ induced by the single chain \piE-extensions $\dField{\KK_{\ell}}$ of $\dField{\KK}$ for $1\le \ell \le m$ with~\eqref{diffAuto:singleChainPExtensionOverConstField} and~\eqref{evMap:singleChainPiMonomialsOverConstField}. Observe that since $\AA_{d}\simeq\AA$, the evaluation function $\ev:\AA_{i}\times\NN\to\KK$  for all $i$ with $1\le i\le d$ is also defined by~\eqref{evMap:singleChainPiMonomialsOverConstField}.

In order to derive the main result of this subsection in Lemma~\ref{lem:orderedMultipleChainPExtOverConstPoly2APiExtOverConstField}, we need following simple construction.

\begin{lemma}\label{lem:diffRingHom}
    Let $\dField{\AA\genn{t}}$ be a \piE-extension of $\dField{\AA}$ with $\s(t)=\alpha\,t$ and let $\dField{\HH}$ be a {\dr}. Let $\tilde{\rho}:\AA\to\HH$ be a {\dr} homomorphism and let $\rho:\AA\genn{t}\to\HH$ be a ring homomorphism defined by $\rho\,|_{\AA}=\tilde{\rho}$ and $\rho(t)=g$ for some $g\in\HH$. If $\s(g)=\rho(\alpha)\,g$, then $\rho$ is a {\dr} homomorphism.
\end{lemma}

\begin{proof}
    Suppose that $\s(g)=\rho(\alpha)\,g$ holds. Then $\s(\rho(t)) = \s(g) = \rho(\alpha)\,g = \rho(\alpha\,t) = \rho(\s(t))$. Consequently, $\s(\rho(f))=\rho(\s(f))$ for all $f\in\AA\genn{t}$.
\end{proof}

\begin{lemma}\label{lem:orderedMultipleChainPExtOverConstPoly2APiExtOverConstField}
    For $1\le\ell\le m$, let $\dField{\KK_{\ell}}$ with $\KK_{\ell} =\KK\genn{y_{\ell,1}}\dots\genn{y_{\ell,s_{\ell}}}$ be single chain \pE-extensions of $\dField{\KK}$ over a rational function field $\KK=K(\Lst{\kappa}{1}{u})$ with base $h_{\ell}\in\KK^{*}$,  the automorphisms~\eqref{diffAuto:singleChainPExtensionOverConstField} and the evaluation functions~\eqref{evMap:singleChainPiMonomialsOverConstField}. Let $d:=\max(s_{1},\dots,s_{m})$ and $\AA_{0}=\KK$. Consider the tower of {\drE s} $\dField{\AA_{i}}$ of $\dField{\AA_{i-1}}$ where  $\AA_{i}=\AA_{i-1}\genn{y_{1,i}}\genn{y_{2,i}}\dots\genn{y_{w_{i},i}}$ for $1\le i\le d$ with $m=w_{1}\ge\cdots\ge w_{d}$, the automorphism~\eqref{evMap:singleChainPiMonomialsOverConstField} and the evaluation function~\eqref{evMap:singleChainPiMonomialsOverConstField}. In particular, one gets $\dField{\AA_{d}}$ as the ordered multiple chain \pE-extension of $\dField{\KK}$ of monomial depth at most $d$ composed by the single chain \pE-extensions $\dField{\KK_{\ell}}$ of $\dField{\KK}$ for $1\le \ell\le m$ with~\eqref{diffAuto:singleChainPExtensionOverConstField} and~\eqref{evMap:singleChainPiMonomialsOverConstField}. Then one can construct 
    \begin{enumerate}[\hspace*{0.5em}(a)]
        \item an ordered multiple chain \apE-extension $\dField{\GG_{d}}$ of $\dField{\tilde{\KK}}$ of monomial depth at most $d$ with $\tilde{\KK}=\tilde{K}(\Lst{\kappa}{1}{u})$ where $\tilde{K}$ is a finite algebraic field extension of $K$, with 
        \begin{equation}\label{eqn:orderedMultipleChainAPiRing}
        \GG_{d}=\tilde{\KK}\genn{\vartheta_{1,1}}\dots\genn{\vartheta_{\upsilon_{1},1}}\genn{\tilde{y}_{1,1}}\dots\genn{\tilde{y}_{e_{1},1}}\dots\genn{\vartheta_{1,d}}\dots\genn{\vartheta_{\upsilon_{d},d}}\genn{\tilde{y}_{1,d}}\dots\genn{\tilde{y}_{e_{d},d}}
        \end{equation}
        where\footnote{Note that if $\upsilon_{i}=0$ or $e_i=0$, then there is no depth-$i$ \aE-monomial or \pE-monomial of depth $i$, respectively.} $\upsilon_{i}\ge0, e_{i}\ge0$. Here the automorphism is defined for the \aE-monomials by
        \begin{equation}\label{diffAuto:aMonomialsOverConstField}
        \s(\vartheta_{\ell,k}) =\gamma_{\ell,k}\,\vartheta_{\ell,k} \quad \text{ where } \quad \gamma_{\ell,k} = \zeta^{\mu_{\ell}}\,\vartheta_{\ell,1}\cdots\,\vartheta_{\ell,k-1}\in{\UU}_{\tilde{\KK}}^{\tilde{\KK}[\vartheta_{\ell,1}]\dots[\vartheta_{\ell,k-1}]} 
        \end{equation}
        for $1\leq k\leq d$ and $1\le \ell\le\upsilon_{k}$ where $\UU=\genn{\zeta}$ is a multiplicative cyclic subgroup of $\tilde{K}^{*}$ generated by a primitive $\lambda$-th root of unity $\zeta\in\tilde{K}^{*}$, and the automorphism is defined for the \pE-monomials by  
        \begin{equation}\label{diffAuto:multipleChainPiExtensionOverIrrConstPolys}
        \s(\tilde{y}_{\ell,k}) = \tilde{\alpha}_{\ell,k}\,\tilde{y}_{\ell,k} \quad \text{ where } \quad \tilde{\alpha}_{\ell,k} = \tilde{h}_{\ell}\,\tilde{y}_{\ell,1}\cdots\,\tilde{y}_{\ell,k-1}\in(\tilde{\KK}^{*})_{\tilde{\KK}}^{\tilde{\KK}\genn{\tilde{y}_{\ell,1}}\dots\genn{\tilde{y}_{\ell,k-1}}}
        \end{equation}
        for $1\leq k\leq d$ and $1 \le \ell \le e_{k}$;  
        \item an evaluation function $\tilde{\ev}:\GG_{d}\times\NN\to\tilde{\KK}$ defined by~\footnote{For all $c\in\KK$, we set $\tilde{\ev}(c,n)=c$ for all $n\ge0$.} 
        \begin{equation}\label{evMap:aPiMonomialsOverConstField}
        \tilde{\ev}(\vartheta_{\ell,k}, n) = \myProduct{j}{1}{n}{\tilde{\ev}(\gamma_{\ell,k}, j-1)} \quad \text{ and } \quad \tilde{\ev}(\tilde{y}_{\ell,k}, n) = \myProduct{j}{1}{n}{\tilde{\ev}(\tilde{\alpha}_{\ell,k}, j-1)};
        \end{equation}
        \item a difference ring homomorphism $\rho_{d}: \AA_{d}\to\GG_{d}$ defined by $\rho_d|_{\KK}=\id_{\KK}$ and\footnote{Note that any \pE-monomial $y_{\ell,k}$ with depth $k$ is mapped to a power product of \apE-monomials having all depth $k$.}
        \begin{equation}\label{diffHom:orderdMultipleChainPMonomial2OrderedMultipleChainAPiMonomial} 
        \begin{aligned}
        & \ \  \rho_d(y_{\ell,k}) = {\bs\vartheta}_{\bs k}^{{\bs \mu}_{{\bs \ell},{\bs k}}}\,\tilde{\bs y}_{\bs k}^{{\bs v}_{{\bs \ell},{\bs k}}} = \vartheta_{1,k}^{\mu_{\ell,1,k}}\cdots\vartheta_{\upsilon_{k},k}^{\mu_{\ell,\upsilon_{k},k}}\,\tilde{y}_{1,k}^{v_{\ell,1,k}}\cdots\tilde{y}_{e_{k},k}^{v_{\ell,e_{k},k}}
        \end{aligned}
        \end{equation}
for $1\leq\ell\leq m$ and $1\leq k\leq s_{\ell}$ with $\mu_{\ell,i,k}\in\NN$ for $1\le i \le \upsilon_{k}$ and $v_{\ell,i,k}\in\ZZ$ for $1\le i \le e_{k}$
    \end{enumerate} 
    such that the following properties hold:	
    \begin{enumerate}[\hspace*{0.5em}(1)]
        \item \manuallabel{item1:orderedMultipleChainPExtOverConstPoly2APiExtOverConstField}{(1)}
        There does not exist a $(v_1,\dots,v_{e_1})\in\ZZ^{e_1}\setminus\{{\bs 0_{e_1}}\}$ with $\tilde{h}_1^{v_1}\dots\tilde{h}_{e_1}^{v_{e_1}}=1$. 
		\item \manuallabel{item2:orderedMultipleChainPExtOverConstPoly2APiExtOverConstField}{(2)}
        The \pE-extension $\dField{\tilde{\AA}_{d}}$ of $\dField{\tilde{\KK}}$ with 
        \begin{align}\label{dom:depth-d-OrderedMultipleChainPiExtension}
        \tilde{\AA}_{d}=\tilde{\KK}\genn{\tilde{\bs y}_{\bs 1}}\genn{\tilde{\bs y}_{\bs 2}}\dots\genn{\tilde{\bs y}_{\bs d}}=\tilde{\KK}\genn{\tilde{y}_{1,1}}\dots\genn{\tilde{y}_{e_{1},1}}\genn{\tilde{y}_{2,1}}\dots\genn{\tilde{y}_{e_{2},2}}\dots\genn{\tilde{y}_{1,d}}\dots\genn{\tilde{y}_{e_{d},d}}
        \end{align}
        and the  automorphism given in~\eqref{diffAuto:multipleChainPiExtensionOverIrrConstPolys} is a \piE-extension. In particular, it is an ordered multiple chain \piE-extension of monomial depth $d$.  
        \item \manuallabel{item3:orderedMultipleChainPExtOverConstPoly2APiExtOverConstField}{(3)} For all $f\in\AA_{d}$ and for all $n\in\NN$ we have 
        \begin{equation}\label{Equ:evrho}
        \ev(f,n)=\tilde{\ev}(\rho_{d}(f),n).
        \end{equation}
    \end{enumerate}
    If $K$ is strongly $\s$-computable, then the above constructions are computable.
\end{lemma}

\begin{proof}
    Let $\dField{\AA_{d}}$ with $\AA_{d}=\AA_{d-1}\genn{y_{1,d}}\genn{y_{2,d}}\dots\genn{y_{w_{d},d}}$ be the ordered multiple chain \pE-extension of $\dField{\KK}$ of monomial depth $d\in\NN$ as described above with the  automorphism~\eqref{evMap:singleChainPiMonomialsOverConstField} and the evaluation function~\eqref{evMap:singleChainPiMonomialsOverConstField}. We prove the Lemma by induction on the monomial depth $d$.\\     
    If $d=1$, statements~\ref{item2:orderedMultipleChainPExtOverConstPoly2APiExtOverConstField} and~\ref{item3:orderedMultipleChainPExtOverConstPoly2APiExtOverConstField} of the
     Lemma hold by Lemma~\ref{lem:depth1PExt2Depth1RPiExt4IrrConstPolysAndAlgNums}. Hence by Lemma~\ref{lem:transcendentalCriterionForPrdts} there are no $g\in\tilde{\KK}^*$ and $(v_1,\dots,v_{e_1})\in\ZZ^{e_1}\setminus\{{\bs 0_{e_1}}\}$
     with $\tilde{h}_1^{v_1}\dots\tilde{h}_{e_1}^{v_{e_1}}=\frac{\s(g)}{g}=1$ and thus also statement~\ref{item1:orderedMultipleChainPExtOverConstPoly2APiExtOverConstField} of the Lemma holds.\\
     Now let $d\ge2$ and suppose that the Lemma holds for $d-1$.
     That is, we can construct  $\dField{\GG_{d-1}}$ with
    \[
    \GG_{d-1}=\tilde{\KK}\genn{\vartheta_{1,1}}\dots\genn{\vartheta_{\upsilon_{1},1}}\genn{\tilde{y}_{1,1}}\dots\genn{\tilde{y}_{e_{1},1}}\dots\genn{\vartheta_{1,d-1}}\dots\genn{\vartheta_{\upsilon_{d-1},d-1}}\genn{\tilde{y}_{1,d-1}}\dots\genn{\tilde{y}_{e_{d-1},d-1}}
    \]
    which is an ordered multiple chain \apE-extension of $\dField{\tilde{\KK}}$ of monomial depth at most $d-1$ with the automorphism given by~\eqref{diffAuto:aMonomialsOverConstField} for $1\le k \le d-1$ and $1\le \ell\le\upsilon_{k}$ and given by~\eqref{diffAuto:multipleChainPiExtensionOverIrrConstPolys}	for $1\le k \le d-1$ and $1 \le \ell \le e_{k}$. In addition, we get the evaluation function $\tilde{\ev}:\GG_{d-1}\times\NN\to \tilde{\KK}$ defined as~\eqref{evMap:aPiMonomialsOverConstField} 
    and the difference ring homomorphism $\rho_{d-1}: \AA_{d-1}\to\GG_{d-1}$ defined by $\rho_{d-1}|_{\KK}=\id_{\KK}$ and~\eqref{diffHom:orderdMultipleChainPMonomial2OrderedMultipleChainAPiMonomial}    
    such that statements~\ref{item1:orderedMultipleChainPExtOverConstPoly2APiExtOverConstField},~\ref{item2:orderedMultipleChainPExtOverConstPoly2APiExtOverConstField}, and~\ref{item3:orderedMultipleChainPExtOverConstPoly2APiExtOverConstField} of the Lemma hold. We prove the Lemma for the ordered multiple chain \pE-extension $\dField{\AA_{d}}$ of $\dField{\KK}$ with $\AA_{d}=\AA_{d-1}\genn{y_{1,d}}\genn{y_{2,d}}\dots\genn{y_{w_{d},d}}$ where $\myd(y_{1,d})=\dots=\myd(y_{w_d,d})=d$. Since the shift quotient of these \pE-monomials is contained in $\AA_{d-1}$, i.e.,
    \[
    \frac{\s(y_{\ell,d})}{y_{\ell,d}}=\alpha_{\ell,d}\in\AA_{d-1}^{*},
    \]
    we can iteratively apply the difference ring homomorphism $\rho_{d-1}:\AA_{d-1}\to\GG_{d-1}$ to rephrase each $\alpha_{\ell,d}$ in $\GG_{d-1}$. In particular, by Remark~\ref{Remark:SimpleProperties} we have $\s^{-1}(\alpha_{\ell,d})=y_{\ell,d-1}$ and thus by~\eqref{diffHom:orderdMultipleChainPMonomial2OrderedMultipleChainAPiMonomial} we get
    \begin{equation}\label{elem:applyDiffHom2ShiftQuotient}
    \begin{aligned} 
    h_{\ell,d}:=\rho_{d-1}(\s^{-1}(\alpha_{\ell,d})) &= \rho_{d-1}(y_{\ell,d-1})={\bs\vartheta}_{{{\bs d}-{\bs 1}}}^{{\bs \mu}_{{\bs \ell},{{{\bs d}-{\bs 1}}}}}\,\tilde{\bs y}_{{{\bs d}-{\bs 1}}}^{{\bs v}_{{\bs \ell},{{{\bs d}-{\bs 1}}}}}
    \end{aligned}
    \end{equation} 
    where 
    ${\bs \vartheta}_{{\bs d}-{\bs 1}}^{{\bs u}_{{\bs \ell},{\bs d-1}}}=\vartheta_{1,d-1}^{u_{\ell,1,d-1}}\cdots\vartheta_{\upsilon_{d-1},d-1}^{u_{\ell,\upsilon_{d-1},d-1}}$
    and 
   $ \tilde{\bs y}_{\bs d-1}^{{\bs v}_{{\bs \ell},{\bs i}}}=\tilde{y}_{1,d-1}^{v_{\ell,1,d-1}}\cdots\tilde{y}_{e_{d-1},d-1}^{v_{\ell,e_{d-1},d-1}}$
    for $1\le\ell\le w_{d}$ with $\mu_{\ell,k,d-1}\in\NN$ for $1\le k \le \upsilon_{d-1}$ and $v_{\ell,k,d-1}\in\ZZ$ for $1\le k \le e_{d-1}$.\\
    If $h_{\ell,d}=1$, it follows with~\eqref{Equ:evrho} ($d$ replaced by $d-1$) that for all $n\in\NN$ we have $$\ev(y_{\ell,d},n)=\prod_{j=1}^n\ev(\alpha_{\ell,d},j-1)=\prod_{j=1}^n\ev(\sigma^{-1}(\alpha_{\ell,d}),j)=\prod_{j=1}^n\tilde{\ev}(\rho(\sigma^{-1}(\alpha_{\ell,d})),j)=\prod_{j=1}^n\tilde{\ev}(h_{\ell,d},j)=1.$$
    In particular, if $h_{\ell,d}=1$ holds for all $1\leq\ell\leq w_d$, we can set $\GG_d:=\GG_{d-1}$ and extend $\rho_{d-1}$ to $\rho_{d}:\AA_{d}\to\GG_{d-1}$ with $\rho_d(y_{\ell,d})=1$ for $1\leq\ell\leq w_d$. Thus the lemma is proven.\\
    Otherwise, take all \apE-monomials in~\eqref{elem:applyDiffHom2ShiftQuotient} for $1\le\ell\le w_{d}$ with non-zero integer exponents.
    Then they belong to at least one of the single chain \apE-extensions of $\dField{\tilde{\KK}}$ in $\dField{\GG_{d-1}}$. Suppose there are $e_{d}\ge0$ of these single chains \piE-extensions and $\upsilon_{d}\ge0$ of them that are single chain \aE-extensions $\dField{\HH_{b}}$; note that we have $e_d+\upsilon_d\geq1$. By appropriate reordering of $\dField{\GG_{d-1}}$ we may suppose that these $e_d$ single chain \piE-extensions
     $\dField{\FF_{r}}$ of $\dField{\tilde{\KK}}$ with $1\leq r\leq e_d$ are given by 
    $\FF_{r}=\tilde{\KK}\genn{\tilde{y}_{r,1}}\genn{\tilde{y}_{r,2}}\dots\genn{\tilde{y}_{r,d-1}}$
    and the $\upsilon_d$  \aE-extensions $\dField{\HH_{b}}$ of $\dField{\tilde{\KK}}$ with $1\le b \le \upsilon_{d}$ can be given by    
    $\HH_{b}=\tilde{\KK}\genn{\vartheta_{b,1}}\genn{\vartheta_{b,2}}\dots\genn{\vartheta_{b,d-1}}$. Now adjoin the \pE-monomials $\tilde{y}_{r,d}$ to $\FF_r$ with~\eqref{diffAuto:aMonomialsOverConstField} where $k=d$ and $\ell=r$    
     yielding the single chain \pE-extensions $\dField{\FF'_{r}}$ of $\dField{\tilde{\KK}}$ of monomial depth $d$ where 
    \[
    \FF'_{r}=\FF_{r}\genn{\tilde{y}_{r,d}}=\tilde{\KK}\genn{\tilde{y}_{r,1}}\genn{\tilde{y}_{r,2}}\dots\genn{\tilde{y}_{r,d-1}}\genn{\tilde{y}_{r,d}}
    \]
    and adjoin the \aE-monomial $\vartheta_{b,d}$ with~\eqref{diffAuto:multipleChainPiExtensionOverIrrConstPolys} where $k=d$ and $\ell=b$ yielding the single chain \aE-extensions $\dField{\HH'_{b}}$ of $\dField{\tilde{\KK}}$ of monomial depth $d$ where 
    \[
    \HH'_{b}=\HH_{b}\genn{\vartheta_{b,d}}=\tilde{\KK}\genn{\vartheta_{b,1}}\genn{\vartheta_{b,2}}\dots\genn{\vartheta_{b,d-1}}\genn{\vartheta_{b,d}}.
    \]
    Furthermore extend the evaluation functions $\tilde{\ev}:\FF'_{r}\times\NN\to\tilde{\KK}$ and $\tilde{\ev}:\HH'_{b}\times\NN\to\tilde{\KK}$ with~\eqref{evMap:aPiMonomialsOverConstField} where $k=d,\ell=r$ or $k=d,\ell=b$, 
    respectively. Now consider the multiple chain \pE-extension $\dField{\tilde{\AA}_{d}}$ of $\dField{\tilde{\KK}}$ with 
    $$\tilde{\AA}_{d}= \tilde{\KK}\genn{\tilde{y}_{1,1}}\dots\genn{\tilde{y}_{e_{1},1}}\dots\genn{\tilde{y}_{1,d-1}}\dots\genn{\tilde{y}_{e_{d-1},d-1}}\genn{\tilde{y}_{1,d}}\genn{\tilde{y}_{2,d}}\dots\genn{\tilde{y}_{e_{d},d}}$$
    which one gets by taking all \pE-monomials in $\GG_{d-1}$ and the new \pE-monomials in $\FF'_{r}=\FF_{r}\genn{\tilde{y}_{r,d}}$ with $1\leq r\leq e_d$. Here the automorphism is given by~\eqref{diffAuto:aMonomialsOverConstField} and~\eqref{diffAuto:multipleChainPiExtensionOverIrrConstPolys}  and the equipped evaluation function is given by~\eqref{evMap:aPiMonomialsOverConstField}.
    Since there does not exist a $g\in\tilde{\KK}$ and $(\Lst{v}{1}{d})\in\ZZ^{d}\sm\zvs{d}$ with $\tfrac{\s(g)}{g}=\ProdLst{\tilde{h}}{1}{d}{v}$, it follows by Theorem~\ref{thm:multipleChainPExt2MultipleChainPiExt} that $\dField{\tilde{\AA}_{d}}$ is a \piE-extension of $\dField{\tilde{\KK}}$.
    In particular, it is an ordered multiple chain \piE-extension of monomial depth $d$ by construction.     
    Thus statements~\ref{item1:orderedMultipleChainPExtOverConstPoly2APiExtOverConstField} and~\ref{item2:orderedMultipleChainPExtOverConstPoly2APiExtOverConstField} of the Lemma hold. Finally, take the \apE-extension $\dField{\GG_d}$ of $\dField{\GG_{d-1}}$ with $\GG_d=\GG_{d-1}\genn{\vartheta_{1,d}}\genn{\vartheta_{2,d}}\cdots\genn{\vartheta_{\upsilon_{d},d}}\genn{\tilde{y}_{1,d}}\genn{\tilde{y}_{2,d}}\cdots\genn{\tilde{y}_{e_{d},d}}$.\\   
    Let $1\le \ell \le w_{d}$ and consider $h_{\ell,d}$ in~\eqref{elem:applyDiffHom2ShiftQuotient}. If $h_{\ell,d}=1$, we define $g_{\ell}=1$. In this case, $1=\rho(h_{\ell,d})=\rho_{d-1}(\sigma^{-1}(\alpha_{\ell,d}))=\sigma^{-1}(\rho_{d-1}(\alpha_{\ell,d}))$, and thus $\frac{\s(g_{\ell})}{g_{\ell}}=1=\rho_{d-1}(\alpha_{\ell,d})$ holds. 
    Otherwise, if $h_{\ell,d}\neq1$, we set
    \begin{align}\label{elem:depthDProductGroupElem}
    g_{\ell}:=\vartheta_{1,d}^{\mu_{\ell,1,d}}\cdots\vartheta_{\upsilon_{d},d}^{\mu_{\ell,\upsilon_{d},d}}\,\tilde{y}_{1,d}^{v_{\ell,1,d}}\cdots\tilde{y}_{e_{d},d}^{v_{\ell,e_{d},d}}\in(\tilde{\KK}^{*})_{\tilde{\KK}}^{\GG_{d}}.
    \end{align}
    Then based on the Remark~\ref{Remark:SimpleProperties} it follows that $\frac{\tilde{y}_{j,d}}{\s^{-1}(\tilde{y}_{j,d})}=\tilde{y}_{j,d-1}$ and $\frac{\vartheta_{j,d}}{\s^{-1}(\vartheta_{j,d})}=\vartheta_{j,d-1}$ and thus
    $$\frac{g_{\ell}}{\s^{-1}(g_{\ell})}=\vartheta_{1,d-1}^{\mu_{\ell,1,d-1}}\cdots\vartheta_{\upsilon_{d-1},d-1}^{\mu_{\ell,\upsilon_{d-1},d-1}}\,\tilde{y}_{1,d-1}^{v_{\ell,1,d-1}}\cdots\tilde{y}_{e_{d-1},d-1}^{v_{\ell,e_{d-1},d-1}}={\bs\vartheta}_{{{\bs d}-{\bs 1}}}^{{\bs \mu}_{{\bs \ell},{{{\bs d}-{\bs 1}}}}}\,\tilde{\bs y}_{{{\bs d}-{\bs 1}}}^{{\bs v}_{{\bs \ell},{{{\bs d}-{\bs 1}}}}}=\rho_{d-1}(y_{l,d-1}).$$
    Hence also in this case we get
    \[
    \frac{\s(g_{\ell})}{g_{\ell}} = \s(\tfrac{g_{\ell}}{\s^{-1}(g_{\ell})})=\s(\rho_{d-1}(y_{l,d-1}))=\rho_{d-1}(\s(y_{l,d-1}))=\rho_{d-1}(\alpha_{\ell,d}).
    \]
	By iterative application of Lemma~\ref{lem:diffRingHom}, the difference ring homomorphism $\rho_{d-1}:\AA_{d-1}\to\GG_{d-1}$ can be extended to
    \[
    \rho_{d}:\AA_{d-1}\genn{y_{1,d}}\genn{y_{2,d}}\cdots\genn{y_{w_{d},d}}\to\GG_{d-1}\genn{\vartheta_{1,d}}\genn{\vartheta_{2,d}}\cdots\genn{\vartheta_{\upsilon_{d},d}}\genn{\tilde{y}_{1,d}}\genn{\tilde{y}_{2,d}}\cdots\genn{\tilde{y}_{e_{d},d}}
    \]
    with
    $\rho_{d}|_{\AA_{d-1}} = \rho_{d-1}$ and $\rho_{d}(y_{\ell,d})=g_{\ell}$
    for $1\le \ell\le w_{d}$. Finally, we show that for all $f\in\AA_{d}$ and $n\in\NN$ we have $\ev(f,n)=\tilde{\ev}(\rho_{d}(f),n)$. First note that for all $n\ge0$ we have
    \begin{align}\label{eqn:depthDGeoPrdtRecEqn1}
    \ev(y_{\ell,d},n+1)=\ev(\s(y_{\ell,d}),n)=\ev(\alpha_{\ell,d},n)\,\ev(y_{\ell,d},n). 
    \end{align}
    On the other hand, since $\rho_{d}$ is a {\dr} homomorphism, we have that 
\begin{equation}\label{Equ:HomProp}
    \s(\rho_{d}(y_{\ell,d})) = \rho_{d}(\s(y_{\ell,d})) = \rho_{d}(\alpha_{\ell,d})\,\rho_{d}(y_{\ell,d})=\rho_{d-1}(\alpha_{\ell,d})\,\rho_{d}(y_{\ell,d})
  \end{equation}
    for all $n\ge0$. Thus we get 
    \begin{align}\label{eqn:depthDGeoPrdtRecEqn2}
    \tilde{\ev}(\rho_{d}(y_{\ell,d}),n+1)=\tilde{\ev}(\s(\rho_{d}(y_{\ell,d})),n) \stackrel{\eqref{Equ:HomProp}}{=} \tilde{\ev}(\rho_{d-1}(\alpha_{\ell,d}),\,n)\,\tilde{\ev}(\rho_{d}(y_{\ell,d}),n).
    \end{align}
    By the induction hypothesis, $\ev(\alpha_{\ell,d},n)=\tilde{\ev}(\rho_{d-1}(\alpha_{\ell,d}),\,n)$ holds for all $n\in\NN$. Therefore with~\eqref{eqn:depthDGeoPrdtRecEqn1} and~\eqref{eqn:depthDGeoPrdtRecEqn2} it follows that $\ev(y_{\ell,d},n)$ and $\tilde{\ev}(\rho(y_{\ell,d}),n)$ satisfy the same first-order recurrence relation. With $\ev(y_{\ell,d},0)=1$ and   
    \begin{align*}
    \tilde{\ev}(\rho_d(y_{\ell,d}),0)&=\tilde{\ev}(g_l,0)
    =\tilde{\ev}(\vartheta_{1,d},0)^{\mu_{\ell,1,d}}\cdots\tilde{\ev}(\vartheta_{\upsilon_{d},d},0)^{\mu_{\ell,\upsilon_{d},d}}\,\tilde{\ev}(\tilde{y}_{1,d},0)^{v_{\ell,1,d}}\cdots\tilde{\ev}(\tilde{y}_{e_{d},d},0)^{v_{\ell,e_{d},d}}=1
    \end{align*}
	it follows then that $\ev(y_{\ell,d},n)=\tilde{\ev}(\rho_d(y_{\ell,d},n)$ holds for all $n\ge0$. Together with the induction hypothesis $\ev(f,n)=\tilde{\ev}(\rho_{d-1}(f),n)$ for all $f\in\AA_{d-1}$ and $n\in\NN$ we get~\eqref{Equ:evrho}
    for all $f\in\AA_d$ and for all $n\ge0$. Consequently also statement~\ref{item1:orderedMultipleChainPExtOverConstPoly2APiExtOverConstField} of the Lemma holds.\\ 
    Finally, if $K$ is strongly $\s$-computable, the base case $d=1$ can be executed explicitly by activating Lemma~\ref{lem:depth1PExt2Depth1RPiExt4IrrConstPolysAndAlgNums}. In particular the induction step can be performed iteratively and thus the {\dr} $\dField{\GG_{d}}$ with~\eqref{eqn:orderedMultipleChainAPiRing} together with~\eqref{diffAuto:aMonomialsOverConstField},~\eqref{diffAuto:multipleChainPiExtensionOverIrrConstPolys} and~\eqref{evMap:aPiMonomialsOverConstField} can be computed. In addition, the difference ring $\dField{\GG_d}$, the {\dr} homomorphism $\rho_{d}:\AA_{d}\to\GG_{d}$ defined by~\eqref{diffHom:orderdMultipleChainPMonomial2OrderedMultipleChainAPiMonomial}  and the evaluation function $\tilde{\ev}$ can be computed. This completes the proof.
\end{proof}

\begin{remark}\label{remk:rearrangeAPiMonomialsInMultipleChainAPiExtension}
    \normalfont Note that the generators of $\GG_{d}$ with~\eqref{eqn:orderedMultipleChainAPiRing} constructed in Lemma~\ref{lem:orderedMultipleChainPExtOverConstPoly2APiExtOverConstField} can be rearranged to get the \apiE-extension $\dField{\tilde{\KK}[\vartheta_{1,1}]\dots[\vartheta_{\upsilon_{1},1}]\dots[\vartheta_{1,d}]\dots[\vartheta_{\upsilon_{d},d}]\genn{\tilde{y}_{1,1}}\dots\genn{\tilde{y}_{e_{1},1}}\dots\genn{\tilde{y}_{1,d}}\dots\genn{\tilde{y}_{e_{d},d}}}$ of $\dField{\tilde{\KK}}$. Furthermore, a consequence of statement~\ref{item3:orderedMultipleChainPExtOverConstPoly2APiExtOverConstField} of Lemma~\ref{lem:orderedMultipleChainPExtOverConstPoly2APiExtOverConstField} is that the diagram 
    \begin{center}
        \begin{tikzcd}
            \AA \arrow[r, "\psi"] \arrow[d, "\rho", swap]
            & \ringOfEquivSeqs[\KK] \arrow[d, hook, "\rho'"] \\
            \GG_{d} \arrow[r, "\tilde{\tau}"]
            & \ringOfEquivSeqs[\tilde{\KK}]
        \end{tikzcd}
    \end{center}
    commutes where $\AA=\AA_{d}$, $\rho=\rho_{d}$, $\rho'=\id$ and the {\dr} homomorphism $\tilde{\tau}$ and $\psi$ are defined by $\tilde{\tau}(f)=\funcSeqA{\tilde{\ev}(f,n)}{n}$ and $\psi(g)=\funcSeqA{\ev(g,n)}{n}$ respectively.
\end{remark}

\begin{example}[Cont. Example~\ref{exa:singleChainAPExtensions}]\label{exa:monomialDepth2OrderedMultipleChainPiExtForAlgNums}
    \normalfont Take the ordered multiple chain \apE-extension $\dField{\AA'}$ of $\dField{\KK}$ with monomial depth $2$ with $\AA'=\KK\genn{{\vartheta_{1,1}}}\genn{{y_{1,1}}}\genn{{y_{2,1}}}\genn{{y_{3,1}}}\genn{{y_{4,1}}}\genn{y_{5,1}}\genn{{\vartheta_{1,2}}}\genn{{y_{2,2}}}\genn{{y_{4,2}}}$ where $\dField{\AA'}$ is composed by the single chain \apE-extensions of $\dField{\KK}$ constructed in parts~\ref{item:monomialDepth2SingleChainRExtBasedAtMinus1},~\ref{item:monomialDepth1SingleChainPiExtBasedAtSqrt3},~\ref{item:monomialDepth2SingleChainPiExtBasedAt2},~\ref{item:monomialDepth1SingleChainPiExtBasedAt3},~\ref{item:monomialDepth2SingleChainPiExtBasedAt5}, and~\ref{item:monomialDepth1SingleChainPiExtBasedAt25} of Example~\ref{exa:singleChainAPExtensions}. By Lemma~\ref{lem:orderedMultipleChainPExtOverConstPoly2APiExtOverConstField} and Remark~\ref{remk:rearrangeAPiMonomialsInMultipleChainAPiExtension} we can construct the \apE-extension $\dField{\GG}$ of $\dField{\KK}$ where
    \begin{equation}\label{domain:monomialDepth2OrderedMultipleChainPiExt}
        \GG=\KK[\vartheta_{1,1}][\vartheta_{1,2}]\genn{\tilde{y}_{1,1}}\genn{\tilde{y}_{2,1}}\genn{\tilde{y}_{3,1}}\genn{\tilde{y}_{2,2}}\genn{\tilde{y}_{3,2}} 
    \end{equation}
    with the automorphism $\s$ and evaluation function $\tilde{\ev}:\GG\times\NN\to\KK$ given by~\eqref{diffAutoEval:monomialDepth2SingleChainRExtBasedAtMinus1A}, \eqref{diffAutoEval:monomialDepth2SingleChainRExtBasedAtMinus1B} and 
    {\fontsize{7.5pt}{0}\selectfont
    \begin{equation}\label{diffAutoEval:nestingDepth2PiExtOverAlgNums}
        \begin{aligned}
            \s({\tilde{y}_{1,1}}) &= \sqrt{3}\,{\tilde{y}_{1,1}}, & \s({\tilde{y}_{2,1}}) &= 2\,{\tilde{y}_{2,1}}, & \s({\tilde{y}_{3,1}}) &= 5\,{\tilde{y}_{3,1}}, &  \s({\tilde{y}_{2,2}}) &= 2\,{\tilde{y}_{2,1}}\,{\tilde{y}_{2,2}}, & \s({\tilde{y}_{3,2}}) &= 5\,{\tilde{y}_{3,1}}\,{\tilde{y}_{3,2}},\\
            \tilde{\ev}({\tilde{y}_{1,1}},n) &= \myProduct{k}{1}{n}{\sqrt{3}}, & \tilde{\ev}({\tilde{y}_{2,1}},n) &= \myProduct{k}{1}{n}{2}, & \tilde{\ev}({\tilde{y}_{3,1}},n) &= \myProduct{k}{1}{n}{5}, & \ev(\tilde{y}_{2,2}, n) &= \myProduct{k}{1}{n}{\myProduct{j}{1}{k}{2}}, & \ev(\tilde{y}_{3,2}, n) &= 	\myProduct{k}{1}{n}{\myProduct{j}{1}{k}{5}}
        \end{aligned}
    \end{equation}}%
with the following properties.
    By part~\ref{item2:orderedMultipleChainPExtOverConstPoly2APiExtOverConstField} of  Lemma~\ref{lem:orderedMultipleChainPExtOverConstPoly2APiExtOverConstField}, the sub-{\dr} $\dField{\tilde{\DD}}$ of the {\dr} $\dField{\GG}$ with $\tilde{\DD} = \KK\genn{\tilde{y}_{1,1}}\genn{\tilde{y}_{2,1}}\genn{\tilde{y}_{3,1}}\genn{\tilde{y}_{2,2}}\genn{\tilde{y}_{3,2}}$ is an ordered multiple chain \piE-extension of $\dField{\KK}$ with the automorphism $\s$ and the evaluation function $\tilde{\ev}:\tilde{\DD}\times\NN\to\KK$ defined in~\eqref{diffAutoEval:nestingDepth2PiExtOverAlgNums}. In addition by part~\ref{item3:orderedMultipleChainPExtOverConstPoly2APiExtOverConstField} of  Lemma~\ref{lem:orderedMultipleChainPExtOverConstPoly2APiExtOverConstField} we get the {\dr} homomorphism $\rho:\AA'\to\GG$ defined by $\rho|_{\KK[\vartheta_{1,1}][\vartheta_{1,2}]}=\id_{\KK[\vartheta_{1,1}][\vartheta_{1,2}]}$ and
    \begin{equation}\label{Equ:DefineRhoInExa}
        \begin{aligned}
        \rho(y_{1,1}) &= \tilde{y}_{1,1}, \qquad & \rho(y_{2,1}) &= \tilde{y}_{2,1}, \qquad & \rho(y_{3,1}) &= \tilde{y}_{1,1}^{2}, \qquad & \rho(y_{4,1}) &= \tilde{y}_{3,1}, \\
        \rho(y_{5,1}) &= \tilde{y}_{3,1}^{2}, \qquad & \rho(y_{2,2}) &= \tilde{y}_{2,2}, \qquad & \rho(y_{4,2}) &= \tilde{y}_{3,2}
        \end{aligned}
    \end{equation} 
    such that $\tilde{\ev}(\rho(f),n)=\ev(f,n)$ holds for all $n\in\NN$ and $f\in\AA'$, 
\end{example}

\subsection{Nested products with roots of unity}\label{subsec:rootsOfUnity}
Throughout this Subsection, $K$ is a field containing $\QQ$, $\KK_{m}$ is a splitting field for the polynomial $x^{m}-1$ over $K$ (i.e., all roots of the polynomial $x^{m}-1$ are in $\KK_{m}$) for some $m\in\ZZ_{\geq2}$ and $\UU_{m}$ is the set of all $m$-th roots of unity over $K$ in $\KK_{m}$ (which forms a multiplicative subgroup of $K^*$). Then $\Prod(\UU_{m})$ is the set of all geometric products over roots of unity in $\UU_{m}$. For $\GG\subseteq\KK_m(x)$ we define $\ProdExpr(\GG,\UU_{m})$ as the set of all elements 
\begin{equation*}
    \smashoperator{\sum_{\mathclap{\substack{\bs{v}=(\Lst{v}{1}{e})\in S}}}^{}}{a_{\bs{v}}(n) \,P_{1}(n)^{v_{1}}\cdots P_{e}(n)^{v_{e}}}
\end{equation*}
with $e\in\NN$, $S\subseteq\NN^{e}$ finite, $a_{\bs{v}}(x)\in\GG$ for $\bs{v}\in S$ and  $P_{1}(n),\dots,P_{e}(n)\in\Prod(\UU_{m})$. 

The main result of this Subsection in Theorem~\ref{thm:mainResultForNestedProductsOverRootsOfUnity} states that products over roots of unity with finite nesting depth can be represented by the single product $\zeta_{\lambda}^n$ where $\zeta_{\lambda}\in\UU_{\lambda}$ for some $\lambda\geq2$ is a primitive $\lambda$-th root of unity.

\begin{theorem}\label{thm:mainResultForNestedProductsOverRootsOfUnity}
    Suppose we are given the geometric products $A_{1}(n),\dots,A_{e}(n)\in\Prod(\UU_{m})$ in $n$ of nesting depth $r_{i}\in\NN$ with \begin{equation}\label{eqn:finteNestedDepthGeoPrdtsOverRootsOfUnity}
    A_{i}(n) = \myProduct{k_{1}}{\ell_{i,1}}{n}{\zeta_{i,1}\myProduct{k_{2}}{\ell_{i,2}}{k_{1}}{\zeta_{i,2}}\cdots\myProduct{k_{r_{i}}}{\ell_{i,r_{i}}}{k_{r_{i}-1}}{\zeta_{i,r_{i}}}}
    \end{equation}
    for $1\le i\le e$ where $\zeta_{i,j}\in\UU_{m}$, $\ell_{i,j}\in\NN$ for $1\le j \le r_{i}$. Then there exist a $\lambda\in\ZZ_{\geq2}$ with $m\,|\,\lambda$ and a primitive $\lambda$-th root of unity $\zeta_{\lambda}\in\KK_{\lambda}^{*}$ satisfying the following property. For all $1 \le i \le e$ there exist $f_{i,j}\in\KK_{\lambda}$ for $0\le j < \lambda$ such that for
    \begin{equation}\label{Equ:BiDefRootOfUnity}
    B_{i}(n)=\mySum{j}{0}{\lambda-1}{f_{i,j}}\,(\zeta_{\lambda}^{n})^{j}\in\ProdExpr(\KK_{\lambda},\UU_{\lambda})
    \end{equation}
    we have
    \begin{equation}\label{Equ:Ai=BiForRootOfUnity}
    A_{i}(n)=B_{i}(n)
    \end{equation}
    for all $n\ge\max(\ell_{i,1},\dots,\ell_{i,r_i})-1$. In particular, if $\KK$ is computable and one can solve ~\ref{prob:ProblemO} (see below), the above construction can be given explicitly.
\end{theorem}

\subsubsection{The period and algorithmic aspects}\label{subsub:algorithmicMachinery}

For the treatment of Theorem~\ref{thm:mainResultForNestedProductsOverRootsOfUnity} we will introduce the period of a {\dr} element introduced in~\cite{karr1981summation}. In particular, we will use the algorithms from~\cite{schneider2016difference} that enable one to calculate the period within nested \rE-extensions, resp.\ \aE-extensions.

\begin{definition}\label{defn:period}
    \normalfont Let $\dField{\AA}$ be a {\dr}. The \emph{period} of $\alpha\in\AA^{*}$ is defined by 
    \begin{equation*}
        \per(\alpha) = 
            \begin{cases}
                0 & \text{if } \nexists \, n > 0 \text{ s.t. } \s^{n}(\alpha)=\alpha \\
                \min\{ n > 0 \,|\,\s^{n}(\alpha)=\alpha\} & \text{otherwise}.
            \end{cases}
    \end{equation*}
\end{definition}

\noindent As it turns out, this task is connected to compute the order of a ring.

\begin{definition}\label{defn:order}
	\normalfont Let $\AA$ be a ring and let $\alpha\in\AA\setminus\{0\}$. Then the order of $\alpha$ is defined by
	\begin{equation*}
	\ord(\alpha) = 
	\begin{cases}
	0, & \text{if } \nexists \, n > 0 \text{ with } \alpha^{n}=1. \\
	\min\{ n > 0 \,|\,\alpha^{n}=1\}, & \text{otherwise}.
	\end{cases}
	\end{equation*}
\end{definition}

\noindent Namely, if we can solve the following problem (which is ~\ref{prob:ProblemGO} with $w=1$):

\begin{ProblemSpecBox}[\op{P}]{ 
		{\op{P} for $\alpha \in K^{*}$} 
	}\label{prob:ProblemO}
	{
		\emph{Given} a field $K$ and $\alpha\in K^{*}$. \emph{Find} $\ord(\alpha)$.
		
		\vspace*{-0.2cm}
	} 
\end{ProblemSpecBox}

\noindent then we can also compute the period by the following lemma.

\begin{lemma}\label{lem:periodComputability}
    Let $\dField{\EE}$ with $\EE=\KK_{m}[\vartheta_{1}]\dots[\vartheta_{e}]$ be a simple \aE-extension of $\dField{\KK_{m}}$. Then the following statements hold.
    \begin{enumerate}[(1)]
        \item \manuallabel{item1:periodComputability}{(1)}$\per(\vartheta_{i})>0$ for all $1 \le i \le e$.
        \item \manuallabel{item2:periodComputability}{(2)}If $\KK_{m}$ is computable and \ref{prob:ProblemO} is solvable, then $\per(\vartheta_{i})$ is computable for all $1 \le i \le e$.
    \end{enumerate}
\end{lemma}

\begin{proof}
    Statement~\ref{item1:periodComputability} follows by~\cite[Proposition 5.5]{schneider2016difference} (compare~\cite[Proposition 6.2.20]{ocansey2019difference}) and statement~\ref{item2:periodComputability} follows by~\cite[Corollary 5.6]{schneider2016difference} (compare~\cite[Corollary 6.2.21]{ocansey2019difference}).
\end{proof}

\begin{example}[Cont. Example~\ref{exa:monomialDepth2OrderedMultipleChainPiExtForAlgNums}]\label{exa:monomialDepth2SimpleAExt}
    \normalfont Consider the sub-difference ring $\dField{\KK[\vartheta_{1,1}][\vartheta_{1,2}]}$ of $\dField{\GG}$ with \eqref{domain:monomialDepth2OrderedMultipleChainPiExt} with the automorphism and the evaluation function defined in~\eqref{diffAutoEval:monomialDepth2SingleChainRExtBasedAtMinus1A} and~\eqref{diffAutoEval:monomialDepth2SingleChainRExtBasedAtMinus1B}. By construction it is a simple \aE-extension of the {\df} $\dField{\KK}$. Since $\per(c)=1$ for all $c\in\KK$, it follows that $\per(-1)=1$. Applying the algorithms from~\cite{schneider2016difference} (see the comments in the proof of Lemma~\ref{lem:periodComputability}) we compute for the depth-$1$ \aE-monomial $\vartheta_{1,1}$ the period $\per(\vartheta_{1,1})=2$, while for the depth-$2$ \aE-monomial $\vartheta_{1,2}$ we get $\per(\vartheta_{1,2})=4$. 
\end{example}

\subsubsection{Idempotent representation of single \texorpdfstring{\rpE}{RP}-extensions}

In order to prove Theorem~\ref{thm:mainResultForNestedProductsOverRootsOfUnity} we rely on the property that elements in basic \rpE-extensions can be expressed by idempotent elements. We start with the following basic facts inspired by~\cite{put1997galois}.

\begin{lemma}\label{lem:orthogonalIdempotentElements}
    Let $\FF$ be a field and let $\zeta$ be a primitive $\lambda$-th root of unity. Let $\FF[\vartheta]$ be a polynomial ring subject to the relation $\vartheta^{\lambda}=1$. Then the following statements hold.
    \begin{enumerate}[(1)]
        \item\manuallabel{item1:orthogonalIdempotentElements}{(1)} The elements $\Lst{{\bs e}}{0}{\lambda-1}\in\FF[\vartheta]$ with 
        \begin{equation}\label{eqn:idempotentElements}
        {\bs e}_{k} = {\bs e}_{k}(\vartheta) := \smashoperator{\prod_{\mathclap{\substack{{i}={0} \vspace*{1mm}\\  i\neq\lambda-1-k}}}^{{\lambda-1}}} {\frac{\vartheta-\zeta^{i}}{\zeta^{\lambda-1-k}-\zeta^{i}}}
        \end{equation}
        are idempotent and for all $0\le k < \lambda$ we have 
        \begin{equation}\label{eqn:shiftAndEvaluateIdempotentElements}
        {\bs e}_{k}(\zeta^{j}) = 
        \begin{cases}
        1 & \text{if }\, j = \lambda-1-k \\
        0 & \text{if }\, j\neq\lambda-1-k
        \end{cases} \qquad \text{ and } \qquad  
        {\bs e}_{k}(\zeta\,\vartheta)={\bs e}_{k+1\hspace*{-0.5em}\mod\lambda}.
        \end{equation}
        \item\manuallabel{item2:orthogonalIdempotentElements}{(2)} The idempotent elements defined in \eqref{eqn:idempotentElements} are pairwise orthogonal and ${\bs e}_{0}+\cdots+{\bs e}_{\lambda-1}=1$.	
    \end{enumerate}
\end{lemma}

\noindent In Proposition~\ref{pro:directDecompositionOfSingleRPSExt} we state that a simple \rpE-extension $\dField{\EE}$ of a {\df} $\dField{\FF}$ with $\EE=\FF[\vartheta]\langle t_{1} \rangle\dots\langle t_{e} \rangle$ can be decomposed in terms of these idempotent elements. For more details in the general setting of \rpisiE-extensions we refer to~\cite[Theorem 4.3]{schneider2017summation}; compare also~\cite[Corollary 1.16]{put1997galois}, and~\cite[Lemma 6.8]{hardouin2008differential}. 

\begin{proposition}\label{pro:directDecompositionOfSingleRPSExt}
    Let $\dField{\EE}$ with $\EE=\FF[\vartheta]\langle t_{1} \rangle\dots\langle t_{e} \rangle$ be an \rpE-extension of a {\df} $\dField{\FF}$ where $\vartheta$ is an \rE-monomial of order $\lambda$ with $\zeta=\frac{\s(\vartheta)}{\vartheta}\in\const\dField{\FF}^{*}$ and the $t_i$ are \pE-monomials. Let $\Lst{{\bs e}}{0}{\lambda-1}$ be the idempotent, pairwise orthogonal elements in \eqref{eqn:idempotentElements} that sum up to one. Then the following statement holds:
    \begin{enumerate}[(1)]
        \item\manuallabel{item1:directDecompositionOfSingleRPSExt}{(1)} The ring $\EE$ can be written as the direct sum 
        \begin{equation}\label{eqn:directSumOfRPSExt}
        \EE = {\bs e}_{0}\EE \oplus \cdots \oplus {\bs e}_{\lambda-1}\EE
        \end{equation}
        where ${\bs e}_{k}\EE$ forms for all $0\le k < \lambda$ a ring with ${\bs e}_{k}$ being the multiplicative identity element.
        \item\manuallabel{item2:directDecompositionOfSingleRPSExt}{(2)} We have that ${\bs e}_{k}\EE={\bs e}_{k}\tilde{\EE}$ for $0\le k < \lambda$ where $\tilde{\EE}=\FF\langle t_{1} \rangle\dots\langle t_{e} \rangle$.
    \end{enumerate}
\end{proposition}

We are now ready to obtain the following key result; for the corresponding result for nested \aE-extensions of monomial depth $1$ we refer to~\cite[Lemma 2.22]{schneider2017summation}.

\begin{theorem}\label{thm:diffRingHomBtnSimpleAExtAndSingleRExtAndEmbedding2RingOfSeq}
    Let $m\in\ZZ_{\geq2}$ and take a primitive $m$-th root of unity $\zeta_{m}\in\KK_{m}^{*}$.
    Let $\dField{\KK_{m}[\vartheta_{1}]\dots[\vartheta_{e}]}$ be a simple \aE-extension of $\dField{\KK_{m}}$ with  $\s(\vartheta_{i})=\alpha_{i}\,\vartheta_{i}$ for $1\le i\le e$ where $\alpha_i=\zeta_{m}^{u_{i}}\,\vartheta_{1}^{z_{i,1}}\cdots\vartheta_{i-1}^{z_{i,i-1}}$ with $u_i,z_{i,j}\in\NN$.
    Furthermore, let $\ev_{\hspace*{-0.15em}m}:\KK_{m}[\vartheta_{1}]\dots[\vartheta_{e}]\times\NN\to\KK_{m}$ be the evaluation function defined by 
    \begin{equation}\label{eqn:evalOfNestedAMonomials}
    \ev_{\hspace*{-0.15em}m}(\vartheta_{i}, n) = \myProduct{j}{1}{n}{\ev_{\hspace*{-0.15em}m}(\alpha_{i},j-1)},
    \end{equation}
    and let $\tau_{m}:\KK_{m}[\vartheta_{1}]\dots[\vartheta_{e}]\to\ringOfEquivSeqs$ be the $\KK_m$-homomorphism given by $\tau_{m}(f)=\funcSeqA{\ev_{\hspace*{-0.15em}m}(f,\,n)}{n}$ Then the following statements hold.
    \begin{enumerate}[(1)]
        \item \manuallabel{item1:diffRingHomBtnSimpleAExtAndSingleRExtAndEmbedding2RingOfSeq}{(1)} Define $\lambda:=\lcm(m,\per(\vartheta_{1}),\dots,\per(\vartheta_{e}))>1$. Then there is an \rE-extension $\dField{\KK_{\lambda}[\vartheta]}$ of $\dField{\KK_{\lambda}}$ of order\footnote{$\KK_{\lambda}$ is a finite algebraic extension of $\KK_{m}$ and $\zeta\in\KK_m$ is a primitive $\lambda$-th root of unity.} $\lambda$ with $\zeta=\frac{\s(\vartheta)}{\vartheta}\in\KK_{\lambda}^{*}$ such that 
        \begin{equation}\label{diffHom:aExtToSingleRExt}
        \varphi:\KK_{m}[\vartheta_{1}]\dots[\vartheta_{e}]\to\KK_{\lambda}[\vartheta]= {\bs e}_{0}\KK_{\lambda}\oplus\cdots\oplus{\bs e}_{\lambda-1}\KK_{\lambda}
        \end{equation}
        defined with
        \begin{equation}\label{eqn:sinpleAExt2SingleRExtHom}
        \varphi(f)=f_{0}{\bs e}_{0}+\cdots+f_{\lambda-1}{\bs e}_{\lambda-1}
        \end{equation}
        where $f_{i}=\ev_{\hspace*{-0.15em}m}(f,\lambda -1 -i)\in\KK_{m}\subseteq\KK_{\lambda}$ for $0 \le i < \lambda$ is a {\dr} homomorphism; here the ${\bs e}_{k}$ are the idempotent orthogonal elements defined in~\eqref{eqn:idempotentElements}. In particular, $\varphi|_{\KK_{m}}=\id_{\KK_{m}}$.
        \item \manuallabel{item2:diffRingHomBtnSimpleAExtAndSingleRExtAndEmbedding2RingOfSeq}{(2)} Take the evaluation function $\ev_{\hspace*{-0.12em}\lambda}:\KK_{\lambda}[\vartheta]\times\NN\to\KK_{\lambda}$ defined by $\ev_{\hspace*{-0.12em}\lambda}|_{\KK_{\lambda}}=\id$ and $\ev_{\hspace*{-0.12em}\lambda}(\vartheta,\,n)=\zeta^{n}$ and consider the $\KK_{\lambda}$-homomorphism  $\tau_{\lambda}:\KK_{\lambda}[\vartheta]\to\ringOfEquivSeqs[\KK_{\lambda}]$ defined by $\tau_{\lambda}(f)=\funcSeqA{\ev_{\hspace*{-0.12em}\lambda}(f,\,n)}{n}$. Then  for the pairwise orthogonal idempotent elements ${\bs e}_{k}$ defined in~\eqref{eqn:idempotentElements} with $0\le k < \lambda$, we have that 
        \begin{equation}\label{eqn:evalOfIdempotentElements}
        \ev_{\hspace*{-0.12em}\lambda}({\bs e}_{k}, n) = 
        \begin{cases}
        1 & \text{if } \lambda \,|\, n + k + 1, \\
        0 & \text{if } \lambda\nmid n + k + 1.
        \end{cases}
        \end{equation}
        \item \manuallabel{item3:diffRingHomBtnSimpleAExtAndSingleRExtAndEmbedding2RingOfSeq}{(3)} The $\KK_{\lambda}$-homomorphism $\tau_{\lambda}:\KK_{\lambda}[\vartheta]\to\ringOfEquivSeqs[\KK_{\lambda}]$ with the evaluation function defined in part~\ref{item2:diffRingHomBtnSimpleAExtAndSingleRExtAndEmbedding2RingOfSeq} is injective.
        \item \manuallabel{item4:diffRingHomBtnSimpleAExtAndSingleRExtAndEmbedding2RingOfSeq}{(4)}The diagram 
        \begin{equation}\label{commDia:simpleAExtToSingleRExt}
        \begin{gathered}
        \begin{tikzcd}
        \KK_{m}[\vartheta_{1}]\dots[\vartheta_{e}] \arrow[r, "\tau_{m}"] \arrow[d, "\varphi", swap] & \ringOfEquivSeqs[\KK_{m}] \arrow[d, "\varphi'"] \\
        \KK_{\lambda}[\vartheta] \simeq  {\bs e}_{0}\KK_{\lambda}\oplus\cdots\oplus{\bs e}_{\lambda-1}\KK_{\lambda} \arrow[r, "\tau_{\lambda}"]&\ringOfEquivSeqs[\KK_{\lambda}]
        \end{tikzcd}
        \end{gathered}
        \end{equation}
        commutes where $\varphi':\ringOfEquivSeqs[\KK_{m}]\to\ringOfEquivSeqs[\KK_{\lambda}]$ is the injective {\dr} homomorphism defined by $\varphi'(a)=a$.
    \end{enumerate}	
    If $\KK_{m}$ is computable and~\ref{prob:ProblemO} is solvable in $\KK_{m}$, then the above constructions are computable.
\end{theorem}

\begin{proof}
    \begin{enumerate}[(1)]
        \item Since $\zeta_{m}^{u_{i}}\in\KK_{m}^{*}$, $\per(\zeta_{m}^{u_{i}})=1>0$ for all $1\le u_{i}\le e$. In addition, it follows by statement~\ref{item1:periodComputability} of Lemma~\ref{lem:periodComputability} that $\per(\vartheta_{i})>0$ for all $1\le i \le e$. Define $\lambda:=\lcm(m,\per(\vartheta_{1}),\dots,\per(\vartheta_{e}))>1$. Note that $m\mid\lambda$, i.e., $\KK_{\lambda}$ is  an algebraic field extension of $\KK_{m}$.     
        Finally, take  
         $\zeta:=\ee^{\frac{2\,\pi\,\ii}{\lambda}} = (-1)^{\frac{2}{\lambda}}\in\KK_{\lambda}^{*}$ and construct the \aE-extension $\dField{\KK_{\lambda}[\vartheta]}$ of $\dField{\KK_{\lambda}}$ with $\s(\vartheta)=\zeta\,\vartheta$. By Lemma~\ref{lem:aExtOverPiSigmaExtIsRExt} it follows that $\dField{\KK_{\lambda}[\vartheta]}$ is an \rE-extension of $\dField{\KK_{\lambda}}$. By Proposition~\ref{pro:directDecompositionOfSingleRPSExt} we have that 
         $\KK_{\lambda}[\vartheta]={\bs e}_{0}\KK_{\lambda}\oplus\cdots\oplus{\bs e}_{\lambda-1}\KK_{\lambda}$
        where the ${\bs e}_{k}$ for $0\le k < \lambda$ are the orthogonal idempotent elements defined in~\eqref{eqn:idempotentElements}. Now consider the map~\eqref{diffHom:aExtToSingleRExt} defined by~\eqref{eqn:sinpleAExt2SingleRExtHom}. We will now show that $\varphi$ is a ring homomorphism. Observe that for any $c\in\KK_{m}$, $\ev_{\hspace*{-0.15em}m}(c,i)=c$ for all $i\in\NN$ and with statement~\ref{item2:orthogonalIdempotentElements} of Lemma~\ref{lem:orthogonalIdempotentElements} we have that  
        \begin{equation*}
            \varphi(c) = c\,{\bs e}_{0}+\dots+c\,{\bs e}_{\lambda-1} = c\,({\bs e}_{0} +\cdots+{\bs e}_{\lambda-1}) = c.
        \end{equation*}
        Further, let $f,g\in\KK_{m}[\vartheta_{1}]\dots[\vartheta_{e}]$ with $f:=a\,\ProdLst{\vartheta}{1}{e}{v}$ and $g:=b\,\ProdLst{\vartheta}{1}{e}{z}$ where $a,b\in\KK_{m}$ and $v_{i},\,z_{i}\in\NN$ for $1 \le i \le e$. Define $f_{k}:=\ev_{\hspace*{-0.15em}m}(f, \lambda-1-k)$ and $g_{k}:=\ev_{\hspace*{-0.15em}m}(g, \lambda-1-k)$ for $0\le k < \lambda$. Then, 
        \begin{align*}
            \varphi(f+g) &= \ev_{\hspace*{-0.15em}m}(f+g,\lambda-1)\,{\bs e}_{0}+\cdots+\ev_{\hspace*{-0.15em}m}(f+g,0)\,{\bs e}_{\lambda-1}\\
            &=\big(\ev_{\hspace*{-0.15em}m}(f,\lambda-1)+\ev_{\hspace*{-0.15em}m}(g,\lambda-1)\big)\,{\bs e}_{0}+\cdots+\big(\ev_{\hspace*{-0.15em}m}(f,0)+\ev_{\hspace*{-0.15em}m}(g,0)\big)\,{\bs e}_{\lambda-1}\\
            &=\big(\ev_{\hspace*{-0.15em}m}(f,\lambda-1)\,{\bs e}_{0}+\cdots+\ev_{\hspace*{-0.15em}m}(f,0)\,{\bs e}_{\lambda-1}\big) + \big(\ev_{\hspace*{-0.15em}m}(g,\lambda-1)\,{\bs e}_{0}+\cdots+\ev_{\hspace*{-0.15em}m}(g,0)\,{\bs e}_{\lambda-1}\big)\\
            &= \big(f_{0}\,{\bs e}_{0} +\cdots+f_{\lambda-1}\,{\bs e}_{\lambda-1}\big)+\big(g_{0}\,{\bs e}_{0} +\cdots+g_{\lambda-1}\,{\bs e}_{\lambda-1}\big)\\
            &= \varphi(f) + \varphi(g).
        \end{align*}
        Similarly, 
        \begin{align*}
            \varphi(f\,g) &= \ev_{\hspace*{-0.15em}m}(f\,g,\lambda-1)\,{\bs e}_{0}+\cdots+\ev_{\hspace*{-0.15em}m}(f\,g,0)\,{\bs e}_{\lambda-1}\\
            &=\big(\ev_{\hspace*{-0.15em}m}(f,\lambda-1)\,\ev_{\hspace*{-0.15em}m}(g,\lambda-1)\big)\,{\bs e}_{0}+\cdots+\big(\ev_{\hspace*{-0.15em}m}(f,0)\,\ev_{\hspace*{-0.15em}m}(g,0)\big)\,{\bs e}_{\lambda-1}\\
            &=f_{0}\,g_{0}\,{\bs e}_{0}+f_{1}\,g_{1}\,{\bs e}_{1}+\cdots+f_{\lambda-1}\,g_{\lambda-1}\,{\bs e}_{\lambda-1}\\
            &= \big(f_{0}\,{\bs e}_{0} +\cdots+f_{\lambda-1}\,{\bs e}_{\lambda-1}\big) \, \big(g_{0}\,{\bs e}_{0} +\cdots+g_{\lambda-1}\,{\bs e}_{\lambda-1}\big)\\
            &= \varphi(f)\,\varphi(g).
        \end{align*}
        The first equality follows since the ${\bs e}_{i}$ are idempotent. Thus, $\varphi$ is a ring homomorphism. Next we show by induction on the number of \aE-monomials, $e\in\NN$, that $\varphi$ is a {\dr} homomorphism. For the base case, i.e., $e=0$, there are no \aE-monomials. Since
            $\s(\varphi(c)) = \s(c) = c = \varphi(c) = \varphi(\s(c))$
        for all $c\in\KK_{m}$, $\varphi$ is a {\dr}-homomorphism. Now assume that the statement holds for all \aE-monomials $\vartheta_{i}$ with $0\le i < e$, and consider an \aE-monomial $\vartheta_{e}$ with $\s(\vartheta_{e})=\tilde{\alpha}\,\vartheta_{e}$ where $\tilde{\alpha}\in(\KK_{m}^{*})_{\KK_{m}}^{\KK_{m}[\vartheta_{1}]\cdots[\vartheta_{e-1}]}$. Then we will show that  
        \begin{equation}\label{eqn:diffRingHom4AMonomial}
            \s(\varphi(\vartheta_{e})) = \varphi(\s(\vartheta_{e}))
        \end{equation}
        holds. For the left hand side of \eqref{eqn:diffRingHom4AMonomial}, we have that 
        $\varphi(\vartheta_{e}) = \gamma_{0}\,{\bs e}_{0} + \cdots + \gamma_{\lambda-1}\,{\bs e}_{\lambda-1}$
        where $\gamma_{i}=\ev_{\hspace*{-0.15em}m}(\vartheta_{e},\lambda-1-i)\in\KK_{m}$ for $0\le i < \lambda$ are $\lambda$-th roots of unity. Thus, 
        \[
            \s(\varphi(\vartheta_{e})) = \s(\gamma_{0})\,\s({\bs e}_{0})+\cdots+\s(\gamma_{\lambda-1})\,\s({\bs e}_{\lambda-1}).
        \]
        By~\eqref{eqn:shiftAndEvaluateIdempotentElements} we have that $\s({\bs e}_{\lambda-1})={\bs e}_{0}$ and $\s({\bs e}_{i})={\bs e}_{i+1}$ for $0\le i < \lambda-1$. In addition, for $1\leq i<\lambda$ we get $\s(\gamma_i)=\ev_{\hspace*{-0.15em}m}(\vartheta_e,\lambda-i)=\gamma_{i-1}$. For $i=\lambda$ observe that $\per(\vartheta_{e})\mid\lambda$ by definition and thus $\sigma^{\lambda}(\vartheta_e)=\vartheta_e$. 
        Consequently $\sigma(\gamma_0)=\ev_m(\vartheta_e,\lambda)=\ev_m(\s^{\lambda}(\vartheta_e),0)=\ev_m(\vartheta_e,0)=\gamma_{\lambda-1}$.        
        Therefore, 
        \begin{equation}\label{eqn:varphiOnAMonomial}
            \s(\varphi(\vartheta_{e})) = \tilde{\gamma}_{0}\,{\bs e}_{0} + \cdots + \tilde{\gamma}_{\lambda-1}\,{\bs e}_{\lambda-1}
        \end{equation}
        where $\tilde{\gamma}_{0}=\gamma_{\lambda-1}$ and $\tilde{\gamma}_{i}=\gamma_{i-1}$ for $1\le i \le \lambda-1$. 
         
        \noindent For the right hand side of~\eqref{eqn:diffRingHom4AMonomial}, we have         
        \begin{align}
            \varphi(\s(\vartheta_{e})) = \varphi(\tilde{\alpha}\,\vartheta_{e}) & = \varphi(\tilde{\alpha})\,\varphi(\vartheta_{e}) \nn \\
            &=(\alpha_{0}\,{\bs e}_{0}+\cdots+\alpha_{\lambda-1}\,{\bs e}_{\lambda-1})(\gamma_{0}\,{\bs e}_{0}+\cdots+\gamma_{\lambda-1}\,{\bs e}_{\lambda-1}) \nn\\
            &=\alpha_{0}\,\gamma_{0}\,{\bs e}_{0} + \cdots + \alpha_{\lambda-1}\,\gamma_{\lambda-1}\,{\bs e}_{\lambda-1}\label{eqn:diffRingHomOnAMonomial}
        \end{align}
        where $\alpha_{i}=\ev_{\hspace*{-0.15em}m}(\tilde{\alpha},\lambda-1-i)$ and $\gamma_{i}=\ev_{\hspace*{-0.15em}m}(\vartheta_{e},\lambda-1-i)$ for $0 \le i < \lambda$ are $\lambda$-th roots of unity. Again~\eqref{eqn:diffRingHomOnAMonomial} holds since the ${\bs e}_{i}$ are idempotent. 
Finally observe that for $0 \le i < \lambda$ we have
\begin{align*}
\alpha_{i}\,\gamma_{i} &= \ev_{m}(\vartheta_{e},\,\lambda-1-i)\,\ev_{m}(\tilde{\alpha},\,\lambda-1-i) = \ev_{m}(\tilde{\alpha}\,\vartheta_e,\,\lambda-1-i)\\
&=\ev_{m}(\sigma(\vartheta_{e}),\,\lambda-1-i)=\ev_{m}(\vartheta_{e},\,\lambda-i)=\tilde{\gamma}_{i}.
\end{align*}
        With~\eqref{eqn:varphiOnAMonomial} we conclude that~\eqref{eqn:diffRingHom4AMonomial} holds. Thus, $\varphi$ is a {\dr} homomorphism.
        \item By Lemma~\ref{Lemma:ExtendEv} we can define the evaluation function $\ev_{\hspace*{-0.12em}\lambda}:\KK_{\lambda}[\vartheta]\times\NN\to\KK_{\lambda}$ with $\ev_{\hspace*{-0.12em}\lambda}|_{\KK_{\lambda}}=\id$ and $\ev_{\hspace*{-0.12em}\lambda}(\vartheta,\,n)=\zeta^{n}$ and by Lemma~\ref{lem:evaluationHomomorphism} we get the $\KK_{\lambda}$-homomorphism  $\tau_{\lambda}:\KK_{\lambda}[\vartheta]\to\ringOfEquivSeqs[\KK_{\lambda}]$ defined by $\tau_{\lambda}(f)=\funcSeqA{\ev_{\hspace*{-0.12em}\lambda}(f,\,n)}{n}$.    
        Statement~\eqref{eqn:evalOfIdempotentElements} follows by~\eqref{eqn:shiftAndEvaluateIdempotentElements}.
        \item Since $\dField{\KK_{\lambda}[\vartheta]}$ is an \rE-extension of a {\df} $\dField{\KK_{\lambda}}$ it follows by Theorem~\ref{Thm:injectiveHom} that $\tau_{\lambda}$ is injective. 
        \item Let $\alpha\in\KK_{m}[\vartheta_{1}]\dots[\vartheta_{e}]$ and let $\ev_{\hspace*{-0.15em}m},\,\ev_{\hspace*{-0.12em}\lambda}$ be evaluation functions for $\KK_{m}[\vartheta_{1}]\dots[\vartheta_{e}]$ and $\KK_{\lambda}[\vartheta]$ defined by~\eqref{eqn:evalOfNestedAMonomials} and~\eqref{eqn:evalOfIdempotentElements}, respectively. We will show 
        \begin{equation}\label{eqn:commutation}
        \varphi'(\tau_{m}(\alpha)) = \tau_{\lambda}(\varphi(\alpha)).
        \end{equation}
        For the left hand side of~\eqref{eqn:commutation}, we have 
        \[
            \varphi'(\tau_{m}(\alpha)) = \tau_{m}(\alpha) =  \funcSeqA{\ev_{\hspace*{-0.15em}m}(\alpha, n)}{n} \in \ringOfEquivSeqs[\KK_{m}]\subseteq\ringOfEquivSeqs[\KK_{\lambda}].
        \]
        For the right hand side of~\eqref{eqn:commutation} we note by~\eqref{eqn:sinpleAExt2SingleRExtHom} that
        $\varphi(\alpha)=\alpha_{0}\,{\bs e}_{0}+\cdots+\alpha_{\lambda-1}\,{\bs e}_{\lambda-1}$
        holds where $\alpha_{i} = \ev_{\hspace*{-0.15em}m}(\alpha,\lambda-1-i)\in\KK_{m}\subseteq\KK_{\lambda}$ for $0\le i < \lambda$. Thus, 
        \begin{align*}
            \tau_{\lambda}(\varphi(\alpha)) &= \funcSeqA{\ev_{\hspace*{-0.12em}\lambda}(\alpha_{0}\,{\bs e}_{0}+\cdots+\alpha_{\lambda-1}\,{\bs e}_{\lambda-1}, n)}{n} \\
            &=\funcSeqA{\ev_{\hspace*{-0.12em}\lambda}(\alpha_{0}\,{\bs e}_{0}, n)}{n}+\cdots+\funcSeqA{\ev_{\hspace*{-0.12em}\lambda}(\alpha_{\lambda-1}\,{\bs e}_{\lambda-1}, n)}{n}\\
            &=\alpha_{0}\,\funcSeqA{\ev_{\hspace*{-0.12em}\lambda}({\bs e}_{0},n)}{n}+\cdots+\alpha_{\lambda-1}\,\funcSeqA{\ev_{\hspace*{-0.12em}\lambda}({\bs e}_{\lambda-1},n)}{n}\\
            &=\funcSeqA{\ev_{\hspace*{-0.15em}m}(\alpha,n)}{n}.
        \end{align*}
        The last equality follows by~\eqref{eqn:evalOfIdempotentElements}. This implies that the diagram~\eqref{commDia:simpleAExtToSingleRExt} commutes.
    \end{enumerate}
    Finally, if $\KK_{m}$ is computable and~\ref{prob:ProblemO} is solvable in $\KK_{m}$, then by statement~\ref{item2:periodComputability} of Lemma~\ref{lem:periodComputability} $\per(\vartheta_{i})$ is computable for all $1\le i \le e$. Consequently, the \rE-extension $\dField{\KK_{\lambda}[\vartheta]}$ of $\dField{\KK_{\lambda}}$, the evaluation function $\ev_{\hspace*{-0.12em}\lambda}:\KK_{\lambda}[\vartheta]\times\NN\to\KK_{\lambda}$ given in statement~\ref{item2:diffRingHomBtnSimpleAExtAndSingleRExtAndEmbedding2RingOfSeq} and the injective $\KK_{\lambda}$-homomorphism $\tau_{\lambda}:\KK_{\lambda}[\vartheta]\to\ringOfEquivSeqs[\KK_{\lambda}]$ given in statement~\ref{item4:diffRingHomBtnSimpleAExtAndSingleRExtAndEmbedding2RingOfSeq} can be constructed explicitly.
\end{proof}

\begin{remark}\label{remk:evalOfSimpleAMonomials}
    \normalfont By statement~\ref{item4:diffRingHomBtnSimpleAExtAndSingleRExtAndEmbedding2RingOfSeq} of Theorem~\ref{thm:diffRingHomBtnSimpleAExtAndSingleRExtAndEmbedding2RingOfSeq} and~\eqref{eqn:evalOfIdempotentElements} we observe that for a fixed $k\in\NN$ and $\alpha\in\KK_{m}[\vartheta_{1}]\dots[\vartheta_{e}]$ we get
    \begin{equation}\label{eqn:evalOfElementsInSimpleAExts}
    \begin{aligned} 
    \ev_{\hspace*{-0.15em}m}(\alpha,k)&=\ev_{\hspace*{-0.12em}\lambda}(\varphi(\alpha), k)= \alpha_{0}\,\ev_{\hspace*{-0.12em}\lambda}({\bs e}_{0},k) + \cdots + \alpha_{\lambda-1}\,\ev_{\hspace*{-0.12em}\lambda}({\bs e}_{\lambda-1},k)\\
    &=\alpha_{j}\,\ev_{\hspace*{-0.12em}\lambda}({\bs e}_{j},k) = \alpha_{j} = \ev_{\hspace*{-0.15em}m}(\alpha, j)
    \end{aligned}
    \end{equation}
    for some  $j\in\{0,1,\dots,\lambda-1\}$ with $\lambda\,|\,k-j$. In other words, the sequence repeats periodically.
\end{remark}

\begin{example}[Cont. Example~\ref{exa:monomialDepth2SimpleAExt}]\label{exa:singleRExt}
    \normalfont Consider the  $\UU_{2}$-simple \aE-extension $\dField{\KK[\vartheta_{1,1}][\vartheta_{1,2}]}$ of $\dField{\KK}$ with the automorphism and the evaluation function given in~\eqref{diffAutoEval:monomialDepth2SingleChainRExtBasedAtMinus1A} and~\eqref{diffAutoEval:monomialDepth2SingleChainRExtBasedAtMinus1B}, which was constructed in Examples~\ref{exa:monomialDepth2OrderedMultipleChainPiExtForAlgNums} and~\ref{exa:monomialDepth2SimpleAExt} with $\KK=\QQ(\sqrt{3})(=\KK_2)$. From Example~\ref{exa:monomialDepth2SimpleAExt} we already know the period of the \aE-monomials $\vartheta_{1,1}$ and $\vartheta_{1,2}$ in $\KK[\vartheta_{1,1}][\vartheta_{1,2}]$. Set $\lambda=\lcm(m,\per(\vartheta_{1,1}),\,\per(\vartheta_{1,2}))=4$ with $m=2$, take a primitive $4$th root of unity, say $\zeta:=\ee^{\frac{\pi\,\ii}{2}}=(-1)^{\frac{1}{2}} = \ii$ and define  $\tilde{\KK}=\QQ(\ii,\sqrt{3})(=\KK_4)$. Then by statement~\ref{item1:diffRingHomBtnSimpleAExtAndSingleRExtAndEmbedding2RingOfSeq} of  Theorem~\ref{thm:diffRingHomBtnSimpleAExtAndSingleRExtAndEmbedding2RingOfSeq} there is an \rE-extension $\dField{\tilde{\KK}[\vartheta]}$ of $\dField{\tilde{\KK}}$ of order $4$ with the automorphism
    \begin{align}\label{diffAuto:order4SingleRMonomial}
    \s(\vartheta)=\ii\,\vartheta
    \end{align}
    and the evaluation function $\tilde{\ev}:\tilde{\KK}[\vartheta]\times\NN\to\tilde{\KK}$ given by 
    \begin{align}\label{evMap:order4SingleRMonomial}
    \tilde{\ev}(\vartheta, n)=\myProduct{k}{1}{n}{\ii}=\ii^n.
    \end{align}
    We have $\tilde{\KK}[\vartheta] = {\bs e}_{0}\tilde{\KK} \oplus {\bs e}_{1}\tilde{\KK} \oplus {\bs e}_{2}\tilde{\KK} \oplus {\bs e}_{3}\tilde{\KK}$ where the idempotent elements ${\bs e}_{k}$ for $0\le k \le 3$ are defined by 
    \begin{equation}\label{eqn:order4IdempotentElements}
    \begin{aligned}
    {\bs e}_{0} &= \tfrac{\ii}{4}\,\left(\vartheta^{3}+\ii\,\vartheta^{2}-\vartheta-\ii\right), &&
    {\bs e}_{1} = \tfrac{1}{4}\,\big(1-\vartheta+\vartheta^{2}-\vartheta^{3}\big),\\[3pt]
    {\bs e}_{2} &= \tfrac{\ii}{4}\,\left(-\vartheta^{3}+\ii\,\vartheta^{2}+\vartheta-\ii\right), &&
    {\bs e}_{3} = \tfrac{1}{4}\,\big(1+\vartheta+\vartheta^{2}+\vartheta^{3}\big)
    \end{aligned}
    \end{equation}
    with ${\bs e}_{0}+{\bs e}_{1}+{\bs e}_{2}+{\bs e}_{3}=1$. 
    Furthermore, the ring homomorphism $\varphi: \KK[\vartheta_{1,1}][\vartheta_{1,2}]\to\tilde{\KK}[\vartheta]$ defined by $\varphi|_{\KK}=\id_{\KK}$ and 
    $\varphi(\vartheta_{1,i})=\beta_{i,0}{\bs e}_{0}+\beta_{i,1}{\bs e}_{1}+\beta_{i,2}{\bs e}_{2}+\beta_{i,3}{\bs e}_{3}$
    where $\beta_{i,j}=\ev(\vartheta_{i,j},3-j)$ for $i\in\{1,2\}$ and $0\le j \le 3$ is a {\dr} homomorphism. More precisely, for the \aE-monomials we have that  
    \begin{equation}\label{Equ:vartheta12Map}
    \begin{aligned}
    \varphi(\vartheta_{1,1}) &= -{\bs e}_{0}+{\bs e}_{1}-{\bs e}_{2}+{\bs e}_{3} = \vartheta^{2},\\
    \varphi(\vartheta_{1,2}) &= -{\bs e}_{0}-{\bs e}_{1}+{\bs e}_{2}+{\bs e}_{3} = \tfrac{(1-\ii)}{2}\,\vartheta\,(\vartheta^{2}+\ii).
    \end{aligned}
    \end{equation}
    Given $\ev$ and $\tilde{\ev}$ we obtain the difference ring homomorphisms $\tau_2:\KK[\vartheta_{1,1}][\vartheta_{1,2}]\to \ringOfEquivSeqs[\KK]$ defined by $\tau_2(f)=(\ev(f,n))_{n\geq0}$ and  $\tau_4:\tilde{\KK}[\vartheta]\to \ringOfEquivSeqs[\tilde{\KK}]$ defined by $\tau_4(f)=(\tilde{\ev}(f,n))_{n\geq0}$. In particular, by statement~\ref{item3:diffRingHomBtnSimpleAExtAndSingleRExtAndEmbedding2RingOfSeq} of Theorem~\ref{thm:diffRingHomBtnSimpleAExtAndSingleRExtAndEmbedding2RingOfSeq} $\tau_{4}$ is injective. Finally, by defining the embedding $\varphi':\ringOfEquivSeqs[\KK]\to\ringOfEquivSeqs[\tilde{\KK}]$ with $\varphi'(a)=a$ for all $a\in\ringOfEquivSeqs[\KK]$ we conclude by statement~\ref{item4:diffRingHomBtnSimpleAExtAndSingleRExtAndEmbedding2RingOfSeq} of the Theorem~\ref{thm:diffRingHomBtnSimpleAExtAndSingleRExtAndEmbedding2RingOfSeq} that the following diagram commutes
    \begin{equation*}
    \begin{gathered}
    \begin{tikzcd}
    \KK[\vartheta_{1,1}][\vartheta_{1,2}] \arrow[r, "\tau_{2}"] \arrow[d, "\varphi", swap] & \ringOfEquivSeqs[\tilde{\KK}] \arrow[d, "\varphi'"] \\
    \tilde{\KK}[\vartheta] =  {\bs e}_{0}\tilde{\KK}\,\oplus\,{\bs e}_{1}\tilde{\KK}\,\oplus\,{\bs e}_{2}\tilde{\KK}\,\oplus\,{\bs e}_{3}\tilde{\KK} \arrow[r, "\tau_{4}"]&\ringOfEquivSeqs[\tilde{\KK}].
    \end{tikzcd}
    \end{gathered}
    \end{equation*}
\end{example}

We are finally ready to obtain the
\medskip

\noindent\textbf{Proof of Theorem~\ref{thm:mainResultForNestedProductsOverRootsOfUnity}.}
Suppose we are given the geometric products $A_{1}(n),\dots,A_{e}(n)\in\Prod(\UU_{m})$ in $n$ with~\eqref{eqn:finteNestedDepthGeoPrdtsOverRootsOfUnity} where $\zeta_{i,r_i}\neq1$ for $1\leq i\leq e$. As elaborated in Section~\ref{subsec:SyncLowerBounds} we can rewrite each $A_{i}(n)$ as 
    \begin{equation}\label{eqn:synchronisedNestedGeoPrdtOverRootsOfUnity}
    A_{i}(n) = u_{i}\,\tilde{A}_{i}(n)\quad \text{ where } \quad \tilde{A}_{i}(n)=\myProduct{k_{1}}{1}{n}{\zeta_{i,1}\myProduct{k_{2}}{1}{k_{1}}{\zeta_{i,2}}\cdots\myProduct{k_{r_{i}}}{1}{k_{r_{i}-1}}{\zeta_{i,r_{i}}}}
    \end{equation}
    and $u_{i}\in\UU_{m}$ which holds for all $n\geq\max(\ell_{i,1},\dots,\ell_{i,r_i})-1=:\delta_i$.
    Similar to Remark~\ref{remk:arbitrarySimplePExt2MultipleChainPExt} we can rephrase the products in a simple \aE-extension $\dField{\AA}$ of $\dField{\KK_{m}}$ with $\AA=\KK_{m}[\vartheta_{1,1}]\dots[\vartheta_{1,r_1}],\dots,[\vartheta_{e,1}]\dots[\vartheta_{e,r_e}]$ where
    $\alpha_{i,j}:=\frac{\s(\vartheta_{i,j})}{\vartheta_{i,j}}=\zeta_{m}^{u_{i,j}}\,\vartheta_{i,1}\cdots\vartheta_{i,j-1}$
    for $1\le i \le e$  and $1\leq j\leq r_i$ with $u_{i,i}\in\NN$ equipped with the evaluation function $\ev_{\hspace*{-0.15em}m}:\AA\times\NN\to\KK_{m}$ defined by 
    $\ev_{\hspace*{-0.15em}m}(\vartheta_{i,j},n)=\myProduct{k}{1}{n}{\ev_{\hspace*{-0.15em}m}(\alpha_{i,j},k-1)}$
    with the following property. For all $i$ with $1\le i \le e$, there are $\nu_i,\mu_i$ such that the geometric product $\tilde{A}_{i}(n)$ is modelled by $\vartheta_{\nu_i,\mu_i}$, i.e., 
    \begin{equation}\label{eqn:gammaModelsTildA}
    \ev_{\hspace*{-0.15em}m}(\vartheta_{\nu_i,\mu_i},n)=\tilde{A}_{i}(n)
    \end{equation}
    holds for all $n\ge0$. In particular, we get the $\KK$-homomorphism $\tau_m:\AA\to\ringOfEquivSeqs[\KK_{\lambda}]$.  
    By Theorem~\ref{thm:diffRingHomBtnSimpleAExtAndSingleRExtAndEmbedding2RingOfSeq}, there is a single \rE-extension $\dField{\KK_{\lambda}[\vartheta]}$ of $\dField{\KK_{\lambda}}$ subject to the relations $\vartheta^{\lambda}=1$ and $\s(\vartheta)=\zeta_{\lambda}\,\vartheta$ where $\lambda:=\lcm(m,\per(\gamma_{1}),\dots,\per(\gamma_{s}))>0$, $\zeta_{\lambda}:=\ee^{\frac{2\,\pi\,\ii}{\lambda}} = (-1)^{\frac{2}{\lambda}}\in\KK_{\lambda}$ and $\KK_{\lambda}$ is some algebraic extension of $\KK_{m}$ with  $m\,|\,\lambda$. Furthermore, there is a {\dr} homomorphism 
    $\varphi:\AA\to\KK_{\lambda}[\vartheta]	
    $
    and a $\KK_{\lambda}$-embedding 
    $\tau_{\lambda}:\KK_{\lambda}[\vartheta]\to\ringOfEquivSeqs[\KK_{\lambda}]$ with $\tau_{\lambda}(\vartheta)=\funcSeqA{\zeta_{\lambda}^{n}}{n}$
    such that $\tau_m(f)=\tau_{\lambda}(\varphi(f))$ holds for all $f\in\AA$. In particular, we get
    \begin{equation}\label{Equ:ConnectionBetweenEvals}
    \ev_m(\vartheta_{\nu_i,\mu_i},n)=\ev_{\lambda}(\varphi(\vartheta_{\nu_i,\mu_i}),n)
    \end{equation}
    for $1\leq i\leq e$.
    Now define $g_{i,k}\in\KK_{\lambda}$ by
    $\varphi(\vartheta_{\nu_i,\mu_i})= \sum_{k=0}^{\lambda-1}g_{i,k}\,\vartheta^{k} \in \KK_{\lambda}[\vartheta]$.
    Then for  
    $G_{i}(n) := \sum_{k=0}^{\lambda-1}g_{i,k}\,(\zeta_{\lambda}^{n})^{k}$
    with $1\le i \le e$ we get 
    \[
    \ev_m(\vartheta_{\nu_i,\mu_i},n)\stackrel{\eqref{Equ:ConnectionBetweenEvals}}{=}\ev_{\hspace*{-0.12em}\lambda}(\varphi(\vartheta_{\nu_i,\mu_i}),\,n) = G_{i}(n) \quad \forall\,n\ge0.
    \]
    With~\eqref{eqn:synchronisedNestedGeoPrdtOverRootsOfUnity} and~\eqref{eqn:gammaModelsTildA} we conclude that 
    \begin{align*}
    A_{i}(n) &= u_{i}\tilde{A}_{i}(n)= u_{i}\,G_{i}(n)
    \end{align*}
    holds for all $n\ge\delta_i$. In particular, for $B_i(n)$ given in~\eqref{Equ:BiDefRootOfUnity} with $f_{i,k}:=u_i\,g_{i,k}\in\KK_{\lambda}$ we get~\eqref{Equ:Ai=BiForRootOfUnity}.\\   
    If $\KK_m$ is computable and \ref{prob:ProblemO} can be solved, Theorem~\ref{thm:diffRingHomBtnSimpleAExtAndSingleRExtAndEmbedding2RingOfSeq} is constructive and all the above ingredients can be given explicitly.\qed

\begin{example}[Cont. Example~\ref{exa:singleRExt}]
\normalfont
Consider the product expression
$$A(n)=\sqrt{3}\,\myProduct{i}{1}{n}{(-1)} + 2\,\myProduct{k}{1}{n}{\myProduct{i}{1}{k}{(-1)}} + 3\,\left(\myProduct{i}{1}{n}{(-1)}\right)\,\myProduct{k}{1}{n}{\myProduct{i}{1}{k}{(-1)}}\in\ProdExpr\left(\QQ(\sqrt{3}),\UU_{2}\right).
$$
For this instance we follow the construction in Example~\ref{exa:singleRExt} and get $    f=\sqrt{3}\,\vartheta_{1,1} + 2\,\vartheta_{1,2} + 3\,\vartheta_{1,1}\,\vartheta_{1,2}\in\KK[\vartheta_{1,1}][\vartheta_{1,2}]$
with $\ev(f,n)=A(n)$ for all $n\in\NN$. As a consequence we obtain $\tilde{f}:=\varphi(f)=\left(\frac{1}{2}-\frac{\ii}{2}\right)\,\vartheta\,\left((2+3\,\ii) \vartheta^{2}+(1+\ii)\,\sqrt{3}\,\vartheta +(3 + 2\,\ii)\right)\in\tilde{\KK}[\vartheta]$ yielding for $n\in\NN$ the identity
$$A(n)=\ev(f,n)=\tilde{\ev}(\tilde{f},n)=
\left(\tfrac{1}{2}-\tfrac{\ii}{2}\right)\ii^n \left((2+3 \ii) (\ii^n)^2+(1+\ii)
\sqrt{3} \ii^n+(3+2 \ii)\right).$$
In particular, as claimed in Theorem~\ref{thm:mainResultForNestedProductsOverRootsOfUnity}, each of the products in $A(n)$ can be expressed in terms of $\ii^n$. Namely, for all $n\in\NN$ we obtain
\begin{align*}
\prod_{k=1}^n(-1)&=\ev(\vartheta_{1,1},n)=\tilde{\ev}(\varphi(\vartheta),n)=\tilde{\ev}(\vartheta^2,n)=(\ii^n)^2,\\
\prod_{k=1}^n
\prod_{i=1}^k(-1)&=\ev(\vartheta_{1,2},n)=\tilde{\ev}(\varphi(\vartheta),n)=\tilde{\ev}(\tfrac{(1-\ii)}{2}\,\vartheta\,(\vartheta^{2}+\ii),n)=\frac{(1-\ii)}{2}\,\ii^n\,((\ii^n)^{2}+\ii).
\end{align*}
\end{example}

\section{A complete solution of \ref{prob:ProblemRPE}}\label{Sec:CompleteAlg}

We are now ready to combine the building blocks from the previous section to solve \ref{prob:ProblemRPE} in Sections~\ref{Sec:SolForGeometricProds} and~\ref{Sec:SolForHyperProds} below. Afterwards we apply in Section~\ref{subsec:mathematicaDemo} the machinery implemented within the package \texttt{NestedProducts} to concrete examples.

\subsection{The difference ring setting for nested geometric products}\label{Sec:SolForGeometricProds}

First we combine Lemma~\ref{lem:orderedMultipleChainPExtOverConstPoly2APiExtOverConstField} discussed in Subsection~\ref{subsec:simpleRPiExt4GeoPrdts} and Theorem~\ref{thm:diffRingHomBtnSimpleAExtAndSingleRExtAndEmbedding2RingOfSeq} discussed in Subsection~\ref{subsec:rootsOfUnity}. As a consequence, we will obtain the necessary difference ring tools for the full treatment of geometric products of arbitrary but finite nesting depth.

\begin{theorem}\label{thm:orderedMultipleChainPExtOverConstField2SingleRPiExtOverConstField}
    For $1\le\ell\le m$, let $\dField{\KK_{\ell}}$ with $\KK_{\ell} =\KK\genn{y_{\ell,1}}\dots\genn{y_{\ell,s_{\ell}}}$ be  the single chain \pE-extensions of $\dField{\KK}$ over $\KK=K(\Lst{\kappa}{1}{u})$ with  base $h_{\ell}\in\KK^{*}$ for $1\le \ell \le m$, the automorphisms~\eqref{diffAuto:singleChainPExtensionOverConstField} and the evaluation functions~\eqref{evMap:singleChainPiMonomialsOverConstField}.	Let $d:=\max(s_{1},\dots,s_{m})$ and $\AA_{0}=\KK$. Consider the tower of {\drE s} $\dField{\AA_{i}}$ of $\dField{\AA_{i-1}}$ where 
    $\AA_{i}=\AA_{i-1}\genn{y_{1,i}}\genn{y_{2,i}}\dots\genn{y_{w_{i},i}}$
    for $1\le i\le d$ with $m=w_{1}\ge w_{2}\ge\cdots\ge w_{d}$ and the automorphism~\eqref{diffAuto:singleChainPExtensionOverConstField} and the evaluation function~\eqref{evMap:singleChainPiMonomialsOverConstField}. This yields the ordered multiple chain \pE-extension $\dField{\AA_{d}}$ of $\dField{\KK}$ of monomial depth at most $d$ composed by the single chain \piE-extensions $\dField{\KK_{\ell}}$ of $\dField{\KK}$ for $1\le \ell\le m$ with~\eqref{diffAuto:singleChainPExtensionOverConstField} and~\eqref{evMap:singleChainPiMonomialsOverConstField}. Then one can construct 
    \begin{enumerate}[\hspace*{0.01em}(1)]
        \item an \rpiE-extension $\dField{\DD}$ of $\dField{\KK'}$ with 
        \[
        \DD=\KK'[\vartheta]\genn{\tilde{y}_{1,1}}\dots\genn{\tilde{y}_{e_{1},1}}\dots\genn{\tilde{y}_{1,d}}\dots\genn{\tilde{y}_{e_{d},d}},\footnote{
        	For concrete instances the \rE-monomial $\vartheta$ might not be needed. In particular, if $\nu_{\ell,k}=0$ in~\eqref{Equ:GeometricMap}, it can be removed.}
        \]
        where $\KK'=K'(\Lst{\kappa}{1}{u})$ and $K'$ is a finite algebraic field extension of $K$ and 
        with the automorphism 
        \begin{equation*}
        \s(\vartheta)=\zeta'\,\vartheta \qquad \text{ and } \quad \s(\tilde{y}_{\ell,d}) = \tilde{\alpha}_{\ell,k}\,\tilde{y}_{\ell,k} 
        \end{equation*}
        where $\zeta'\in K'$ is a $\lambda'$-th root of unity and 
        \[
        \tilde{\alpha}_{\ell,k} = \tilde{h}_{\ell}\,\tilde{y}_{\ell,1}\cdots\,\tilde{y}_{\ell,k-1}\in(\KK'^{*})_{\KK'}^{\KK'\genn{\tilde{y}_{\ell,1}}\dots\genn{\tilde{y}_{\ell,k-1}}}
        \]
        for $1\leq k\leq d$ and $1 \le \ell \le e_{k}$; 
        \item an evaluation function $\tilde{\ev}:\DD\times\NN\to\KK'$ defined as~\footnote{Note that for all $c\in\KK'$, $\tilde{\ev}(c,n)=c$ for all $n\ge0$.} 
        \begin{equation}\label{evMap:singleRMonomial}
        \tilde{\ev}(\vartheta,n)=\myProduct{j}{1}{n}{\zeta'} \qquad \text { and } \qquad \tilde{\ev}(\tilde{y}_{\ell,k}, n) = \myProduct{j}{1}{n}{\tilde{\ev}(\tilde{\alpha}_{\ell,k}, j-1)};
        \end{equation} 
        \item a difference ring homomorphism $\varphi:\AA_{d}\to\DD$ defined by $\varphi|_{\KK}=\id_{\KK}$ and
        \begin{equation}\label{Equ:GeometricMap} 
\begin{aligned}
& \varphi(y_{\ell,k})=\gamma_{\ell,k}\,\tilde{y}_{1,k}^{v_{\ell,1,k}}\cdots\tilde{y}_{e_{k},k}^{v_{\ell,e_{k},k}}
\end{aligned}
\end{equation}
for $1\leq\ell\leq m$ and $1\leq k\leq s_{\ell}$ with $\gamma_{\ell,k}\in\KK'[\vartheta]$ and $v_{\ell,i,k}\in\ZZ$ for $1\le i \le e_{k}$
    \end{enumerate}
    such that for all $f\in\AA_d$ and $n\in\NN$ we have
\begin{equation}\label{Equ:varPhiProp}
    \ev(f,n)=\tilde{\ev}(\varphi(f),n).
\end{equation}
  If $\KK$ is strongly $\s$-computable, then the constructions above are computable.
\end{theorem}

\begin{proof}
    Let $\dField{\AA_{d}}$ be an ordered multiple chain \pE-extension of $\dField{\KK}$ of monomial depth at most $d$ with the automorphism $\s:\AA_{d}\to\AA_{d}$ defined by~\eqref{diffAuto:singleChainPExtensionOverConstField} and the evaluation function $\ev:\AA_{d}\times\NN\to\KK$ defined by~\eqref{evMap:singleChainPiMonomialsOverConstField}. Then by  Lemma~\ref{lem:orderedMultipleChainPExtOverConstPoly2APiExtOverConstField} we can construct an ordered multiple chain \apE-extension $\dField{\GG_{d}}$ of $\dField{\tilde{\KK}}$ of monomial depth at most $d$ with $\tilde{\KK}=\tilde{K}(\kappa_1,\dots,\kappa_u)$, $\tilde{K}$ being a finite algebraic field extension of $K$ where 
    where $\GG_{d}$ is given by~\eqref{eqn:orderedMultipleChainAPiRing} with the automorphism ~\eqref{diffAuto:aMonomialsOverConstField} and~\eqref{diffAuto:multipleChainPiExtensionOverIrrConstPolys}, the evaluation function $\ev':\GG_{d}\times\NN\to\tilde{\KK}$
    with~\eqref{evMap:aPiMonomialsOverConstField} (where $\tilde{\ev}$ is replaced by $\ev'$), and  the {\dr} homomorphism $\rho_{d}:\AA_{d}\to\GG_{d}$ defined by $\rho_d|_{\KK}=\id_{\KK}$ and~\eqref{diffHom:orderdMultipleChainPMonomial2OrderedMultipleChainAPiMonomial}  
     with the following properties: the sub-{\dr} $\dField{\tilde{\AA}_{d}}$ of $\dField{\GG_{d}}$ where $\tilde{\AA}_{d}$ is given by~\eqref{dom:depth-d-OrderedMultipleChainPiExtension} is a \piE-extension of $\dField{\tilde{\KK}}$. Furthermore, for all $f\in\AA_{d}$ and for all $n\in\NN$ we have 
    \begin{equation}\label{Equ:EvComposition1}
    \ev(f,n)=\ev'(\rho_{d}(f),n).
    \end{equation}    
    If $\upsilon_1=0$ (and thus $\upsilon_2=\dots=\upsilon_d=0$), i.e., no \aE-monomials are involved, we are essentially done. We simply adjoin a redundant \rE-monomial (compare the footnote in Theorem~\ref{thm:orderedMultipleChainPExtOverConstField2SingleRPiExtOverConstField}).
    Otherwise $\upsilon_1\geq1$ and by Remark~\ref{remk:rearrangeAPiMonomialsInMultipleChainAPiExtension} the generators in $\GG_{d}$ can be rearranged to get the \apiE-extension $\dField{\tilde{\GG}}$ of $\dField{\tilde{\KK}}$ where 
    \begin{equation}\label{dom:orderedMultipleChainAPiRing}
    \tilde{\GG}=\tilde{\KK}[\vartheta_{1,1}]\dots[\vartheta_{\upsilon_{1},1}]\dots[\vartheta_{1,d}]\dots[\vartheta_{\upsilon_{d},d}]\genn{\tilde{y}_{1,1}}\dots\genn{\tilde{y}_{e_{1},1}}\dots\genn{\tilde{y}_{1,d}}\dots\genn{\tilde{y}_{e_{d},d}}
    \end{equation}
    with the automorphism given by~\eqref{diffAuto:aMonomialsOverConstField} and~\eqref{diffAuto:multipleChainPiExtensionOverIrrConstPolys} and the evaluation function given by~\eqref{evMap:aPiMonomialsOverConstField} (where $\tilde{\ev}$ is replaced by $\ev'$) satisfying properties~\ref{item1:orderedMultipleChainPExtOverConstPoly2APiExtOverConstField} and~\ref{item2:orderedMultipleChainPExtOverConstPoly2APiExtOverConstField} of Lemma~\ref{lem:orderedMultipleChainPExtOverConstPoly2APiExtOverConstField}. Now consider the sub-{\dr} $\dField{\LL}$ of $\dField{\tilde{\HH}}$ with $\LL=\tilde{\KK}[\vartheta_{1,1}]\dots[\vartheta_{\upsilon_{1},1}]\dots[\vartheta_{1,d}]\dots[\vartheta_{\upsilon_{d},d}]$, which is a {\drE} of $\dField{\tilde{\KK}}$, with the automorphism defined by 
    \begin{equation}\label{diffAuto:multipleChainAExt}
    \s(\vartheta_{\ell,k})=\gamma_{\ell,k}\,\vartheta_{\ell,k} \quad\text{ where }\quad \gamma_{\ell,k}=\zeta^{\mu_{\ell}}\,\vartheta_{\ell,1}\cdots\,\vartheta_{\ell,k-1}\in{\UU}_{\tilde{\KK}}^{\tilde{\KK}[\vartheta_{\ell,1}]\dots[\vartheta_{\ell,k-1}]}
    \end{equation}
    for $1\le k\le d$ and for $1\le\ell\le\upsilon_{k}$ where $\UU=\genn{\zeta}$ is the multiplicative cyclic subgroup of $\tilde{K}$ generated by a primitive $\lambda$-th root of unity, $\zeta\in\tilde{K}^{*}$. Observe that the {\drE} $\dField{\LL}$ of $\dField{\tilde{\KK}}$ with~\eqref{diffAuto:multipleChainAExt} is a simple \aE-extension to which statement~\ref{item1:diffRingHomBtnSimpleAExtAndSingleRExtAndEmbedding2RingOfSeq} of  Theorem~\ref{thm:diffRingHomBtnSimpleAExtAndSingleRExtAndEmbedding2RingOfSeq} can be applied. Thus there is an \rE-extension $\dField{\KK'[\vartheta]}$ of $\dField{\KK'}$ with 
    \begin{equation}\label{diffAuto:singleRExt}
    \s(\vartheta)=\zeta'\,\vartheta
    \end{equation}
    of order $\lambda'$ where $\KK'=K'(\Lst{\kappa}{1}{u})$, $\zeta'$ is a primitive $\lambda'$-th root of unity in $K'$ and $K'$ is a finite algebraic field extension of $\tilde{K}$. Note that the difference ring $\dField{\tilde{\DD}}$ where $\tilde{\DD}$ is given by
    \begin{equation}\label{eqn:orderedMultipleChainPiExt}
    \tilde{\DD}=\KK'\genn{\tilde{y}_{1,1}}\dots\genn{\tilde{y}_{e_{1},1}}\dots\genn{\tilde{y}_{1,d}}\dots\genn{\tilde{y}_{e_{d},d}}
    \end{equation}
     with the automorphism defined by~\eqref{diffAuto:multipleChainPiExtensionOverIrrConstPolys} is a \piE-extension of $\dField{\KK'}$. Thus by Lemma~\ref{lem:aExtOverPiSigmaExtIsRExt} it follows that the \aE-extension $\dField{\tilde{\DD}[\vartheta]}$ of $\dField{\tilde{\DD}}$ with~\eqref{diffAuto:singleRExt} of order $\lambda'$ is an \rE-extension. Note that the generators in the ring $\tilde{\DD}[\vartheta]$ can be rearranged to get $\dField{\DD}$ where  $\DD=\KK'[\vartheta]\genn{\tilde{y}_{1,1}}\dots\genn{\tilde{y}_{e_{1},1}}\dots\genn{\tilde{y}_{1,d}}\dots\genn{\tilde{y}_{e_{d},d}}$ and $\s$ is defined by~\eqref{diffAuto:singleRExt} and~\eqref{diffAuto:multipleChainPiExtensionOverIrrConstPolys}. Since this rearrangement does not change the set of constants, $\dField{\DD}$ is an \rpiE-extension of $\dField{\KK'}$. By statement~\ref{item1:directDecompositionOfSingleRPSExt} of Proposition~\ref{pro:directDecompositionOfSingleRPSExt}
    $\DD = {\bs e}_{0}\DD \oplus \cdots \oplus {\bs e}_{\lambda'-1}\DD$
    and by statement~\ref{item2:directDecompositionOfSingleRPSExt} of the same proposition, ${\bs e}_{k}\DD={\bs e}_{k}\tilde{\DD}$ for $0\le k < \lambda'$. Thus
    $\DD = {\bs e}_{0}\tilde{\DD} \oplus {\bs e}_{1}\tilde{\DD} \oplus \cdots \oplus {\bs e}_{\lambda'-1}\tilde{\DD}$
    holds. Now we show that $\phi:\tilde{\GG}\to\DD$ defined by $\phi|_{\tilde{\KK}}=\id_{\tilde{\KK}}$ with
            \begin{align}
    \phi(\tilde{y}_{\ell,k}) &= \tilde{y}_{\ell,k}, \label{diffHom:orderedMultipleChainPiMonomials}\\
    \phi(\vartheta_{\ell,k}) &= \beta_{\ell,k,0}{\bs e}_{0}+\cdots+\beta_{\ell,k,\lambda'-1}{\bs e}_{\lambda'-1}\label{diffHom:simpleAMonomials}
    \end{align}
    where $\beta_{\ell,k,i}=\ev(\vartheta_{\ell,k},\lambda'-1-i)$ for $0\le i<\lambda'$
    is a {\dr} homomorphism. By statement~\ref{item1:diffRingHomBtnSimpleAExtAndSingleRExtAndEmbedding2RingOfSeq} of Theorem~\ref{thm:diffRingHomBtnSimpleAExtAndSingleRExtAndEmbedding2RingOfSeq}, $\phi|_{\LL}$ which is defined by~\eqref{diffHom:simpleAMonomials} is a {\dr} homomorphism. Since $\phi$ maps $\tilde{y}_{\ell,k}$ to itself, also $\phi$ is a {\dr} homomorphism. Furthermore, for all $f\in\tilde{\GG}$ and for all $n\in\NN$, we have 
    \begin{equation}\label{Equ:EvComposition2}
    \ev'(f,n)=\tilde{\ev}(\phi(f),n).
    \end{equation}
Putting everything together, the map $\varphi:\AA_{d}\to\DD$ with $\varphi=\phi\circ\rho_d$ is a {\dr} homomorphism. 
It is uniquely determined by $\varphi|_{\KK}=\id_{\KK}$ and
    \[
    \varphi(y_{\ell,d})=\phi(\rho_d(y_{\ell,d}))=\gamma_{\ell,d}\,\tilde{y}_{1,d}^{v_{\ell,1,d}}\cdots\tilde{y}_{e_{d},d}^{v_{\ell,e_{d},d}}
    \]
    with $\gamma_{\ell,d}=\beta_{\ell,d,0}{\bs e}_{0}+\cdots+\beta_{\ell,d,\lambda'-1}{\bs e}_{\lambda'-1}\in\KK'[\vartheta]$. 
    Furthermore, by~\eqref{Equ:EvComposition1} and~\eqref{Equ:EvComposition2} it follows that for all $f\in\AA_{d}$ and for all $n\in\NN$ we get
    \[
    \ev(f,n)=\ev'(\rho_d(f),n)= \tilde{\ev}(\phi(\rho_d(f)),n)=\tilde{\ev}(\varphi(f),n). 
    \]
Finally if $K$ is strongly $\s$-computable, then by Lemma~\ref{lem:orderedMultipleChainPExtOverConstPoly2APiExtOverConstField} the {\dr} $\dField{\tilde{\GG}}$ with~\eqref{dom:orderedMultipleChainAPiRing} together with automorphism~\eqref{diffAuto:aMonomialsOverConstField} and~\eqref{diffAuto:multipleChainPiExtensionOverIrrConstPolys}, evaluation function~\eqref{evMap:aPiMonomialsOverConstField} (where $\tilde{\ev}$ is replaced by $\ev'$) and the {\dr} homomorphism $\rho_{d}:\AA_{d}\to\tilde{\GG}$ with~\ref{diffHom:orderdMultipleChainPMonomial2OrderedMultipleChainAPiMonomial} can be computed. Further, by Theorem~\ref{thm:diffRingHomBtnSimpleAExtAndSingleRExtAndEmbedding2RingOfSeq} the {\dr} $\dField{\tilde{\DD}[\vartheta]}$ with the automorphism $\s(\vartheta)=\zeta'\,\vartheta$ and~\eqref{diffAuto:multipleChainPiExtensionOverIrrConstPolys}, the evaluation function~\eqref{evMap:singleRMonomial}
 and the {\dr} homomorphism $\varphi:\tilde{\GG}\to\DD$ given by~\eqref{diffHom:orderedMultipleChainPiMonomials} and~\eqref{diffHom:simpleAMonomials} can be computed. In particular, $\varphi$ and all the components stated in the theorem can be given explicitly.
\end{proof}

\begin{example}[Cont. Example~\ref{exa:monomialDepth2OrderedMultipleChainPiExtForAlgNums},~\ref{exa:singleRExt}]\label{exa:singleRExtAndOrderedMultipleChainPiExtOverConstField}
	\normalfont Take the \apiE-extension $\dField{\GG}$ of $\dField{\KK}$ with~\eqref{domain:monomialDepth2OrderedMultipleChainPiExt} constructed in Example~\ref{exa:monomialDepth2OrderedMultipleChainPiExtForAlgNums} with the automorphism defined in~\eqref{diffAutoEval:monomialDepth2SingleChainRExtBasedAtMinus1A} and \eqref{diffAutoEval:nestingDepth2PiExtOverAlgNums}, and consider the sub-difference ring $\dField{\KK[\vartheta_{1,1}][\vartheta_{1,2}]}$ of $\dField{\GG}$ with the automorphism $\s$ given in~\eqref{diffAutoEval:monomialDepth2SingleChainRExtBasedAtMinus1A}, which is a simple \aE-extension of $\dField{\KK}$ where $\KK=\QQ(\sqrt{3})$. Now we refine the construction from Example~\ref{exa:monomialDepth2OrderedMultipleChainPiExtForAlgNums} by utilizing Example~\ref{exa:singleRExt}. Namely, we take the \rE-extension $\dField{\tilde{\KK}[\vartheta]}$ of $\dField{\tilde{\KK}}$ of order $4$ with the automorphism~\eqref{diffAuto:order4SingleRMonomial}
	and the evaluation function $\tilde{\ev}:\tilde{\KK}[\vartheta]\times\NN\to\tilde{\KK}$ given by~\eqref{evMap:order4SingleRMonomial}
	where $\tilde{\KK}=\QQ(\ii,\sqrt{3})$. Furthermore, take $\dField{\DD}$ where  $\DD=\tilde{\KK}[\vartheta]\genn{\tilde{y}_{1,1}}\genn{\tilde{y}_{2,1}}\genn{\tilde{y}_{3,1}}\genn{\tilde{y}_{2,2}}\genn{\tilde{y}_{3,2}}$ with the automorphism and evaluation function given by~\eqref{diffAuto:order4SingleRMonomial} and~\eqref{evMap:order4SingleRMonomial} for the \rE-monomial $\vartheta$ and~\eqref{diffAutoEval:nestingDepth2PiExtOverAlgNums} for the \piE-monomials $\tilde{y}_{\ell,k}$. By Theorem~\ref{thm:orderedMultipleChainPExtOverConstField2SingleRPiExtOverConstField} $\dField{\DD}$ is an \rpiE-extension of $\dField{\tilde{\KK}}$ where the ring $\DD$ can be written as the direct sum $\DD = {\bs e}_{0}\tilde{\DD} \oplus {\bs e}_{1}\tilde{\DD} \oplus {\bs e}_{2}\tilde{\DD} \oplus {\bs e}_{3}\tilde{\DD}$
	with $\tilde{\DD}=\tilde{\KK}\genn{\tilde{y}_{1,1}}\genn{\tilde{y}_{2,1}}\genn{\tilde{y}_{3,1}}\genn{\tilde{y}_{2,2}}\genn{\tilde{y}_{3,2}}$; here the idempotent elements ${\bs e}_{k}$ for $0\le k \le 3$ are defined by~\eqref{eqn:order4IdempotentElements}. Furthermore, the ring homomorphism $\phi:\GG\to\DD$ defined by $\phi|_{\tilde{\DD}}=\id_{\tilde{\DD}}$ and~\eqref{Equ:vartheta12Map} is a difference ring homomorphism.\\ 
	Finally, consider the \apE-extension $\dField{\AA'}$ of $\dField{\KK}$ as given in Example~\ref{exa:monomialDepth2OrderedMultipleChainPiExtForAlgNums} and consider the {\dr} homomorphism $\rho:\AA'\to\GG$ given in~\eqref{Equ:DefineRhoInExa}. Then with the {\dr} homomorphism $\varphi:\AA'\to\DD$ defined by $\phi(\rho(f))$ for $f\in\AA'$ we get~\eqref{Equ:varPhiProp} for all $n\in\NN$ and $f\in\AA'$. Given this explicit construction we can choose for instance $g\in\AA'$ defined in~\eqref{eqn:hyperGeometricProductRepresentationInRingA} that models $\tilde{G}(n)$ given in~\eqref{eqn:geometricProducts}. This means that $\ev(g,n)=\tilde{G}(n)$ for all $n\geq0$. 
	Thus 
	\begin{equation}\label{Exp:gTilde}
	\tilde{g}:= \varphi(g)=\frac{\varphi(\vartheta_{1,1})\,\varphi(y_{3,1})\,\varphi(y_{5,1})\,\varphi(\vartheta_{1,2})\,\varphi(y_{2,2})}{\varphi(y_{1,1})\,\varphi(y_{2,1})\,\varphi(y_{4,2})}= \frac{(1-\ii)\,\vartheta\,(\ii\,\vartheta^{2}+1)\,{\tilde{y}_{1,1}}\,{\tilde{y}_{3,1}^{2}}\,{\tilde{y}_{2,2}}}{2\,{\tilde{y}_{2,1}}\,{\tilde{y}_{3,2}}}\in\DD
	\end{equation}
	yields for $n\geq0$ the identity
$$\tilde{G}(n)=\ev(g,n)=\tilde{\ev}(\tilde{g},n)=
\frac{1}{2}\frac{(1-\ii)\,(\ii)^{n}\,(\ii\,(\ii^{n})^{2}+1)\,\algSeq{3}{n}\,\intSeqExp{5}{n}{2}\,\intSeq{2}{n+1 \choose 2}}{\intSeq{2}{n}\,\intSeq{5}{n+1 \choose 2}}.$$
\end{example}

\subsection{The solution for nested hypergeometric products}\label{Sec:SolForHyperProds}

So far we have treated hypergeometric products over monic irreducible polynomials of finite nesting depth, say $b$, that are $\delta$-refined for some $\delta\in\NN$; see Definition~\ref{defn:shiftCoPrimeProductRepresentationForm}. Given such hypergeometric products, it follows by Corollary~\ref{cor:reducedNormalFormAndPiExts} that we can construct an ordered multiple chain \piE-extension $\dField{\tilde{\HH}_{b}}$ of $\dField{\KK(x)}$ with $\KK=K(\Lst{\kappa}{1}{u})$ and 
\begin{equation}\label{eqn:orderedMultipleChainPiRing4HyperGeometricPrdts}
\tilde{\HH}_{b}=\KK(x)\genn{\tilde{\bs z}_{\bs 1}}\dots\genn{\tilde{\bs z}_{\bs b}}=\KK(x)\genn{\tilde{z}_{1,1}}\dots\genn{\tilde{z}_{p_{1},1}}\dots\genn{\tilde{z}_{1,b}}\dots\genn{\tilde{z}_{p_{b},b}}.
\end{equation}
In particular, $\dField{\tilde{\HH}_{b}}$ is composed by the single chain \piE-extensions $\dField{\tilde{\FF}_{\ell}}$ of $\dField{\KK(x)}$ for $1 \le \ell \le p_{1}$ with 
\[
\tilde{\FF}_{\ell}=\KK(x)\genn{\tilde{z}_{\ell,1}}\genn{\tilde{z}_{\ell,2}}\dots\genn{\tilde{z}_{\ell,s_{\ell}}},\quad\quad 1 \le k \le s_{\ell}
\]
given by the automorphism $\s:\tilde{\FF}_{\ell}\to\tilde{\FF}_{\ell}$ defined by 
\begin{equation}\label{diffAuto:singleChainPiExtOverRatDiffField}
\s(\tilde{z}_{\ell,k}) = \tilde{\alpha}_{\ell,k}\,\tilde{z}_{\ell,k} \quad \text{ where } \quad \tilde{\alpha}_{\ell,k} = \tilde{f}_{\ell}\,\tilde{z}_{\ell,1}\cdots\,\tilde{z}_{\ell,k-1}\in(\KK(x)^{*})_{\KK(x)}^{\KK(x)\genn{\tilde{z}_{\ell,1}}\dots\genn{\tilde{z}_{\ell,k-1}}}
\end{equation}
and the evaluation function $\tilde{\ev}:\tilde{\FF}_{\ell}\times\NN\to\KK$ defined by
\begin{equation}\label{evMap:singleChainPiMonomialsOverRatDiffField}
\tilde{\ev}(\tilde{z}_{\ell,k}, n) = \myProduct{j}{\delta}{n}{\tilde{\ev}(\tilde{\alpha}_{\ell,k}, j-1)}.
\end{equation}
On the other hand, geometric products over the contents were treated in Subsection~\ref{subsec:simpleRPiExt4GeoPrdts}. 
In Theorem~\ref{thm:orderedMultipleChainPExtOverConstField2SingleRPiExtOverConstField} we constructed a simple \rpiE-extension $\dField{\DD}$ of $\dField{\tilde{\KK}}$ with $\tilde{\KK}=\tilde{K}(\Lst{\kappa}{1}{u})$ where $\tilde{K}$ is a finite algebraic field extension of $K$ in which the geometric products can be modelled. To accomplish this task, we set up a ring of the form
\begin{equation}\label{dom:singleRAndOrderedMultipleChainPiRing}
\DD=\tilde{\KK}[\vartheta]\genn{\tilde{y}_{1,1}}\dots\genn{\tilde{y}_{e_{1},1}}\dots\genn{\tilde{y}_{1,d}}\dots\genn{\tilde{y}_{e_{d},d}}
\end{equation}
with 
\begin{enumerate}[\hspace*{1em}(a)]
    \item the automorphism $\s:\DD\to\DD$ defined by
    \begin{align}
    \s(\vartheta)&=\zeta\,\vartheta, \label{diffAuto:singleRMonomial}\\[5pt]
    \s(\tilde{y}_{\ell,k})&= \tilde{\gamma}_{\ell,k}\,\tilde{y}_{\ell,k} \label{diffAuto:multipleChainPiMonomials}
    \end{align}
    where $\zeta\in\tilde{K}^{*}$ is a $\lambda$-th root of unity and 
    $\tilde{\gamma}_{\ell,k} = \tilde{h}_{\ell}\,\tilde{y}_{\ell,1}\cdots\,\tilde{y}_{\ell,k-1}\in(\tilde{\KK}^{*})_{\tilde{\KK}}^{\tilde{\KK}\genn{\tilde{y}_{\ell,1}}\dots\genn{\tilde{y}_{\ell,k-1}}}$
    for  $1\le k \le d$ and $1 \le \ell \le e_{k}$ and
    \item the evaluation function $\tilde{\ev}:\DD\times\NN\to\tilde{\KK}$ defined by 
    \begin{align}
    \tilde{\ev}(\vartheta, n) &= \myProduct{j}{1}{n}{\zeta},\label{evalMap:singleRMonomial} \\
    \tilde{\ev}(\tilde{y}_{\ell,k}, n) &= \myProduct{j}{1}{n}{\tilde{\ev}(\tilde{\gamma}_{\ell,k}, j-1)}.\label{evalMap:piMonomialsOverConstField}
    \end{align}
\end{enumerate}
In particular, by reordering we obtain the {\drE} $\dField{\tilde{\AA}_{d}}$ of $\dField{\tilde{\KK}}$ with  
\begin{align}\label{dom:depthDOrderedMultipleChainPiExtOverKK}
\tilde{\AA}_{d}=\tilde{\KK}\genn{\tilde{\bs y}_{\bs 1}}\dots\genn{\tilde{\bs y}_{\bs d}}=\tilde{\KK}\genn{\tilde{y}_{1,1}}\dots\genn{\tilde{y}_{e_{1},1}}\dots\genn{\tilde{y}_{1,d}}\dots\genn{\tilde{y}_{e_{d},d}},
\end{align}
the automorphism~\eqref{diffAuto:multipleChainPiMonomials} and the evaluation function~\eqref{evalMap:piMonomialsOverConstField} which is a sub-{\dr} of $\dField{\DD}$. Furthermore, it is an ordered multiple chain \piE-extension of $\dField{\tilde{\KK}}$ and is composed by the single chain \piE-extensions $\dField{\tilde{\KK}_{\ell}}$ of $\dField{\tilde{\KK}}$ where
$\tilde{\KK}_{\ell}=\tilde{\KK}\genn{\tilde{y}_{\ell,1}}\genn{\tilde{y}_{\ell,2}}\dots\genn{\tilde{y}_{\ell,\tilde{e}_{\ell}}}$
for $1 \le \ell \le e_{1}$. 

Putting the two {\dr s} $\dField{\tilde{\HH}_{b}}$ with~\eqref{eqn:orderedMultipleChainPiRing4HyperGeometricPrdts} and $\dField{\DD}$ with~\eqref{dom:singleRAndOrderedMultipleChainPiRing}  together, we will obtain a difference ring in which we can model any finite set of hypergeometric product expressions of finite nesting depth coming from $\ProdExpr(\KK(x))$. Before we can complete this final argument, we have to take care that the two combined extensions yield again an \rpiE-extension. Here we utilize the following result from~\cite[Lemma 5.6]{ocansey2018representing} that holds for single nested \rpiE-extensions.

\begin{lemma}\label{lem:monomialDepthOneRPiExtCombined}
	Let $\dField{\KK(x)}$ be the rational difference field with $\s(x)=x+1$ and let $\dField{\KK(x)\genn{z_{1}}\dots\genn{z_{s}}}$ be a \piE-extension of $\dField{\KK(x)}$ with $\frac{\s(z_{k})}{z_{k}}\in\KK[x]\sm\KK$. Further, let $\KK^{\prime}$ be an algebraic field extension of $\KK$ and let $\dField{\KK^{\prime}\genn{y_{1}}\dots\genn{y_{w}}}$ be a \piE-extension of $\dField{\KK^{\prime}}$ with $\frac{\s(y_i)}{y_i}\in\KK^{\prime}\setminus\{0\}$. Then the {\dr} $\dField{\EE}$ with $\EE=\KK^{\prime}(x)\genn{y_{1}}\dots\genn{y_{w}}\genn{z_{1}}\dots\genn{z_{s}}$ is a \piE-extension of $\dField{\KK^{\prime}(x)}$. Furthermore, the \aE-extension $\dField{\EE[\vartheta]}$ of $\dField{\EE}$ with $\s(\vartheta)=\zeta\,\vartheta$ of order $\lambda$ is an $R$-extension.
\end{lemma}

\noindent Namely, we can enhance the above lemma to nested \rpiE-extensions.

\begin{corollary}\label{cor:simpleRPiExt4HypergeometricPrdts}
    Let $\dField{\KK(x)}$ be a rational {\df} over $\KK$ with $\s(x)=x+1$ and let the difference ring $\dField{\tilde{\HH}_{b}}$ with~\eqref{eqn:orderedMultipleChainPiRing4HyperGeometricPrdts} be an ordered multiple chain \piE-extension of $\dField{\KK(x)}$ with the automorphism~\eqref{diffAuto:singleChainPiExtOverRatDiffField}. Further, let $\tilde{\KK}$ be an algebraic field extension of $\KK$ and let the difference ring $\dField{\tilde{\AA}_{d}}$ with~\eqref{dom:depthDOrderedMultipleChainPiExtOverKK} be the ordered multiple chain \piE-extension of $\dField{\tilde{\KK}}$ with the automorphism~\eqref{diffAuto:multipleChainPiMonomials}. Then the difference ring $\dField{\tilde{\EE}}$ with
    \begin{equation}\label{eqn:orderedMultipleChainPiExt4HyperPrdts}
    \tilde{\EE}=\tilde{\KK}(x)\genn{\tilde{\bs y}_{\bs 1}}\genn{\tilde{\bs z}_{\bs 1}}\dots\genn{\tilde{\bs y}_{\bs d}}\genn{\tilde{\bs z}_{\bs b}}
    \end{equation}
    where $\genn{\tilde{\bs y}_{\bs i}}=\genn{\tilde{y}_{1,i}}\dots\genn{\tilde{y}_{e_{i},i}}$ for $1\le i \le d$ and $\genn{\tilde{\bs z}_{\bs k}}=\genn{\tilde{z}_{1,k}}\dots\genn{\tilde{z}_{p_{k},k}}$ for $1\le k\le b$ is an ordered multiple chain \piE-extension of $\dField{\tilde{\KK}(x)}$. Furthermore, the \aE-extension $\dField{\EE}$ of $\dField{\tilde{\EE}}$ where $\EE=\tilde{\EE}[\vartheta]$ with~\eqref{diffAuto:singleRMonomial} of order $\lambda$ is an \rE-extension. 
\end{corollary}

\begin{proof}
    Take the \piE-extensions $\dField{\tilde{\HH}_{1}}$ of $\dField{\KK(x)}$ with $\tilde{\HH}_{1}=\KK(x)\genn{z_{1,1}}\dots\genn{z_{p_{1},1}}$ and $\dField{\tilde{\AA}_{1}}$ of $\dField{\tilde{\KK}}$ with $\tilde{\AA}_{1}=\tilde{\KK}\genn{\tilde{y}_{1,1}}\dots\genn{\tilde{y}_{e_{1},1}}$ which are both of monomial depth $1$. By Lemma~\ref{lem:monomialDepthOneRPiExtCombined} the {\dr} $\dField{\EE_{1}}$ with $\tilde{\EE}_{1}=\tilde{\KK}(x)\genn{\tilde{y}_{1,1}}\dots\genn{\tilde{y}_{e_{1},1}}\genn{\tilde{z}_{1,1}}\dots\genn{\tilde{z}_{p_{1},1}}$ is a \piE-extension of $\dField{\tilde{\KK}(x)}$ of monomial depth $1$. Consider the ordered multiple chain \pE-extension $\dField{\tilde{\EE}}$ of $\dField{\tilde{\KK}(x)}$ with~\eqref{eqn:orderedMultipleChainPiExt4HyperPrdts} which is composed by the single chain \piE-extensions in the ordered multiple chains $\dField{\tilde{\HH}_{b}}$ and $\dField{\tilde{\AA}_{d}}$. By Theorem~\ref{thm:multipleChainPExt2MultipleChainPiExt} together with Lemma~\ref{lem:transcendentalCriterionForPrdts} it follows that $\dField{\tilde{\EE}}$ is a \piE-extension of $\dField{\tilde{\KK}(x)}$. The quotient field of $\tilde{\EE}$ gives the rational function field $\HH=\KK(x)(\tilde{\bs y}_{\bs 1})(\tilde{\bs z}_{\bs 1})\dots(\tilde{\bs y}_{\bs d})(\tilde{\bs z}_{\bs b})$ and one can extend the automorphism $\sigma$ from $\tilde{\EE}$ to $\HH$ accordingly. Then by Lemma~\ref{lem:transcendentalCriterionForPrdts} ($(2)\Leftrightarrow(3)$) it follows that $\dField{\HH}$ is a \piE-field extension of $\dField{\KK(x)}$. In particular, $\dField{\HH}$ is a \pisiE-field over $\KK$. Thus by Lemma~\ref{lem:aExtOverPiSigmaExtIsRExt} the \aE-extension $\dField{\HH[\vartheta]}$ of $\dField{\HH}$ of order $\lambda$ with the automorphism~\eqref{diffAuto:singleRMonomial} is an \rE-extension.  Therefore $\const(\HH[\vartheta])=\KK$ and with $\KK\subseteq\tilde{\EE}[\vartheta]\subseteq\HH$ it follows that $\const(\tilde{\EE}[\vartheta])=\KK$. But this implies that the \aE-extension $\dField{\tilde{\EE}[\vartheta]}$ of $\dField{\tilde{\EE}}$ of order $\lambda$ with the automorphism~\eqref{diffAuto:singleRMonomial} 
     is an \rE-extension.
\end{proof}

\medskip

Finally, we arrive at the following main result.

\begin{theorem}\label{thm:main}
    Let $\KK=K(\Lst{\kappa}{1}{u})$ be a rational function field, and let 
    $\dField{\KK(x)}$ with $\s(x)=x+1$ be a rational {\df} with the evaluation function $\ev:\KK(x)\times\NN\to\KK$ defined by~\eqref{map:evaForRatFxns}, and the $Z$-function given by~\eqref{eqn:hyperGeoShiftBoundedFxns}. Suppose we are given a finite set of hypergeometric product expressions 
    \begin{equation}\label{elem:hypergeometricPrdts}
    \{A_{1}(n),\dots,A_{m}(n)\}\subseteq\ProdExpr(\KK(x))
    \end{equation}
    of nesting depth at most $d$ for some $d\in\NN$. Then there is a $\delta\in\NN$ and an \rpiE-extension $\dField{\EE}$ of $\dField{\tilde{\KK}(x)}$ of monomial depth at most $d$ where $\tilde{\KK}$ is a finite algebraic field extension of $\KK$ equipped with an evaluation function $\tilde{\ev}:\EE\times\NN\to\tilde{\KK}$ with respect to $\delta$ with the following properties:
    \begin{enumerate}[\hspace*{0.5em}(1)]
        \item\manuallabel{item1:main}{(1)} The map $\tau:\EE\to\ringOfEquivSeqs[\tilde{\KK}]$ with $\tau(f)=\funcSeqA{\tilde{\ev}(f,n)}{n}$ is a $\tilde{\KK}$-embedding.
        \item\manuallabel{item2:main}{(2)} There are elements $\Lst{a}{1}{e}\in\EE^{*}$ such that for $j$ with $1\le j \le e$ and for all $n\ge\delta$ we have
        \[
        A_{j}(n)=\tilde{\ev}(a_{j},n).
        \]
    \end{enumerate}
    If $K$ is a strongly $\s$-computable, such a $\delta$, $\dField{\EE}$ with the evaluation function $\tilde{\ev}$, and the $a_1,\dots,a_m\in\EE$ can be computed.
\end{theorem}

\begin{proof}
    \begin{enumerate}[(a)]
        \item We are given the hypergeometric product expressions in~\eqref{elem:hypergeometricPrdts} where
        \begin{equation*}
        A_j(n)=\smashoperator{\sum_{\mathclap{\substack{{\bs{v}=(\Lst{\nu}{1}{e})\in S_j}{}}}}^{{}}}{a_{j,\bs{v}}(n)\,P_{1}(n)^{\nu_{1}}\cdots P_{e}(n)^{\nu_{e}}}
        \end{equation*}
        with $S_j \subseteq \ZZ^{e}$ finite, $a_{j,\bs{v}}(x)\in\KK(x)$ and $P_{1}(n),\dots,P_{e}(n)\in\Prod(\KK(x))$.    
       Now we follow the construction in Proposition~\ref{pro:preprocessingNestedHypergeometricProductsExtended}. There we can take a $\delta\in\NN$ and construct for all $1\le j \le e$, $c_{j}\in\KK^{*}$, rational functions $r_{j}\in\KK(x)^{*}$, $1$-refined geometric product expressions $\tilde{G}_{j}(n)\in\ProdMon(\KK)$ and $\delta$-refined hypergeometric product expressions in shift-coprime product representation form $\tilde{H}_{j}(n)\in\ProdMon(\KK(x))$ such that 
        \begin{align}\label{eqn:deltaRefinedReduceNormalForm}
        P_{j}(n) &= \tilde{c}_{j}\,\tilde{r}_{j}(n)\,\tilde{G}_{j}(n)\,\tilde{H}_{j}(n)\neq0
        \end{align}
        holds for all $n\ge\max(0,\delta-1)$.
        \item For the hypergeometric product expressions $\tilde{H}_{1}(n),\dots,\tilde{H}_{e}(n)$ in~\eqref{eqn:deltaRefinedReduceNormalForm} we have
        $$\tilde{H}_{i}(n)=\tilde{H}_{i,1}(n)^{n_{i,1}}\cdots \tilde{H}_{i,l_i}(n)^{n_{i,l_i}}$$ 
        for some $l_i\in\NN$ with $n_{i,j}\in\ZZ$ for $1\leq j\leq l_i$ where all the arising hypergeometric products $\tilde{H}_{i,j}(n)$ are $\delta$-refined and in shift-coprime product representation form.
        By Corollary~\ref{cor:reducedNormalFormAndPiExts} we can construct an ordered multiple chain \piE-extension $\dField{\tilde{\HH}_{b}}$ of $\dField{\KK(x)}$ with~\eqref{eqn:orderedMultipleChainPiRing4HyperGeometricPrdts} which is composed by the single chain \piE-extensions $\dField{\tilde{\FF}_{\ell}}$ of $\dField{\KK(x)}$ with $1\le\ell\le p_{1}$ for some $p_{1}\in\NN$ with
        $\tilde{\FF}_{\ell}=\KK(x)\genn{\tilde{z}_{\ell,1}}\genn{\tilde{z}_{\ell,2}}\dots\genn{\tilde{z}_{\ell,s_{\ell}}}$,
        the automorphism $\s:\tilde{\FF}_{\ell}\to\tilde{\FF}_{\ell}$ given in~\eqref{diffAuto:singleChainPiExtOverRatDiffField} and the evaluation function $\tilde{\ev}:\tilde{\FF}_{\ell}\times\NN\to\KK$ defined by $\tilde{\ev}|_{\KK(x)\times\NN}=\ev$ and~\eqref{evMap:singleChainPiMonomialsOverRatDiffField}. In particular, there are $\nu_{i,j},\mu_{i,j}$ such that
        $\tilde{H}_{i,j}(n)=\ev(\tilde{z}_{\nu_{i,j},\mu_{i,j}} ,n)$ holds for all $n\geq\max(0,\delta-1)$. Thus we can take
        $\tilde{h}_{i}=\tilde{z}_{\nu_{i,1},\mu_{i,1}}^{n_{i,1}}\dots\tilde{z}_{\nu_{i,l_i},\mu_{i,l_i}}^{n_{i,l_i}}\in\tilde{\HH_{b}}$ with
        \begin{align}\label{evalMap:deltaRefinedReducedNormalForm}
        \tilde{\ev}(\tilde{h}_{j},n)=\tilde{H}_{j}(n)\quad \forall\,n\ge\delta. 
        \end{align}
        \item Next we treat the geometric product expressions $\tilde{G}_{1}(n),\dots,\tilde{G}_{e}(n)$ in~\eqref{eqn:deltaRefinedReduceNormalForm}. Following Remark~\ref{remk:arbitrarySimplePExt2MultipleChainPExt} we can construct a multiple chain \pE-extension $\dField{\AA}$ of $\dField{\KK}$ where the bases are from $\KK^{*}$ such that there are $\Lst{g}{1}{e}\in\AA$ with $\ev(g_i,n)=\tilde{G}_i(n)$ for all $n\in\NN$.
        Then by Theorem~\ref{thm:orderedMultipleChainPExtOverConstField2SingleRPiExtOverConstField} we can construct an \rpiE-extension $\dField{\DD}$ of $\dField{\tilde{\KK}}$ with~\eqref{dom:singleRAndOrderedMultipleChainPiRing} together with the automorphism $\s:\DD\to\DD$ given in~\eqref{diffAuto:singleRMonomial} and~\eqref{diffAuto:multipleChainPiMonomials}, with       
        the evaluation function $\tilde{\ev}:\DD\times\NN\to\tilde{\KK}$ defined by $\tilde{\ev}(c,n)=c$ for all $c\in\tilde{K}$,  $n\in\NN$, \eqref{evalMap:singleRMonomial} and~\eqref{evalMap:piMonomialsOverConstField}, and 
        with a difference ring homomorphism $\varphi:\AA_{d}\to\DD$ such that for all $f\in\AA$ and $n\in\NN$ we have~\eqref{Equ:varPhiProp}. Thus for $\tilde{g}_j:=\varphi(g_j)$ with $1\leq j\leq e$ we get
        \begin{align}\label{evalMap:geometricPrdts}
        \tilde{\ev}(\tilde{g}_{j},n)=\tilde{ev}(\varphi(g_j),n)\stackrel{\eqref{Equ:varPhiProp}}{=}\ev(g_j,n)=\tilde{G}_{j}(n)\quad \forall\,n\ge0.
        \end{align}
        \item By Corollary~\ref{cor:simpleRPiExt4HypergeometricPrdts} we can merge these two {\dr s} to obtain an \rpiE-extension $\dField{\EE}$ of $\dField{\tilde{\KK}(x)}$ with~\eqref{eqn:orderedMultipleChainPiExt4HyperPrdts} and the automorphism $\s:\EE\to\EE$ and the evaluation function $\tilde{\ev}:\EE\times\NN\to\tilde{\KK}$ defined accordingly. 
        As $\dField{\EE}$ is an \rpiE-extension of the rational {\df} $\dField{\tilde{\KK}(x)}$, it follows by Theorem~\ref{Thm:injectiveHom} that $\tilde{\tau}:\EE\to\ringOfEquivSeqs[\tilde{\KK}]$ defined by~\eqref{eqn:diffRingEmbedding} is a $\tilde{\KK}$-embedding. For $1\le j \le e$, define
        \[
        p_{j}:=\tilde{c}_{j}\,\tilde{r}_{j}\,\tilde{g}_{j}\,\tilde{h}_{j}\in\EE.
        \]
        With $\tilde{\ev}(\tilde{r}_{j},n)=\tilde{r}_{j}(n)$ and the evaluations of $\tilde{h}_{j}$ and $\tilde{g}_{j}$ given in~\eqref{evalMap:deltaRefinedReducedNormalForm} and~\eqref{evalMap:geometricPrdts} together with~\eqref{eqn:deltaRefinedReduceNormalForm} it follows that for all $1\le j \le e$ and for all $n\ge\max(0,\delta-1)$ we have
        \[
        P_{j}(n)=\tilde{c}_{j}\,\tilde{r}_{j}(n)\,\tilde{G}_{j}(n)\,\tilde{H}_{j}(n)=\tilde{\ev}(\tilde{c}_{j},n)\tilde{\ev}(\tilde{r}_{j},n)\tilde{\ev}(\tilde{g}_{j},n)\tilde{\ev}(\tilde{h}_{j},n)=\tilde{\ev}(\tilde{c}_{j}\,\tilde{r}_{j}\,\tilde{g}_{j}\,\tilde{h}_{j},n)=\tilde{\ev}(p_{j},n).
        \]
        \item Finally, we can define $a_j=\smashoperator{\sum_{\mathclap{\substack{{\bs{v}=(\Lst{\nu}{1}{e})\in S_j}{}}}}^{{}}}{a_{j,\bs{v}}\,p_{1}^{\nu_{1}}\cdots p_{e}(n)^{\nu_{e}}}\in\EE$
        for $1\leq j\leq m$ and get $\ev(a_j,n)=A_j(n)$ for all $n\geq\max(\delta-1,0)$.
    \end{enumerate}
Observe that if $\KK$ is strongly $\s$-computable, all the ingredients delivered by Corollary~\ref{cor:reducedNormalFormAndPiExts} and Theorem~\ref{thm:orderedMultipleChainPExtOverConstField2SingleRPiExtOverConstField} can be computed. In particular, $\dField{\EE}$, and $\tau$ with $\tilde{\ev}$ and $\Lst{a}{1}{m}$ can be computed explicitly.
\end{proof}

\medskip

As a consequence we are now in the position to solve \ref{prob:ProblemRPE} as follows.

\begin{corollary}\label{Cor:SolutionToMainProblem}
Let $A(n)\in\ProdExpr(\KK(x))$ with~\eqref{Equ:ProdEDef}. For $A_1(n)=A(n)$ with $m=1$ let $\delta\in\NN$, $\dField{\EE}$ with the evaluation function $\tilde{\ev}$ and the $a:=a_1\in\EE$ be the ingredients as provided in Theorem~\ref{thm:main}. In particular, let $\EE=\tilde{\KK}(x)[\vartheta]\langle p_1\rangle\dots\langle p_s\rangle$ where $\vartheta$ is the \rE-monomial with $\sigma(\vartheta)=\zeta\,\vartheta$ and let $p_1,\dots,p_s$ be the \piE-monomials. Furthermore, let $a=\smashoperator{\sum_{\mathclap{\substack{{\bs{v}=(\Lst{\nu}{0}{s})\in \tilde{S}}{}}}}^{{}}}{b_{\bs{v}}(n)\,\vartheta^{\mu_0}p_{1}^{\mu_{1}}\cdots p_{s}^{\mu_{s}}}$
with $\tilde{S}\subseteq \{0,\dots,\lambda-1\}\times\ZZ^{s}$ finite and $a_{\bs{v}}(x)\in\tilde{\KK}(x)$ for $\bs{v}\in \tilde{S}$. Then the following holds:
\begin{enumerate}
\item[(1)]\manuallabel{item1:SolutionToMainProblem}{(1)} $\ev(\vartheta,n)=\zeta^n$ for all $n\in\NN$; furthermore, for $1\leq i\leq s$ we have $\ev(p_i,n)=Q_i(n)$ for all $n\geq0$ where $Q_i(n)\in\Prod(\tilde{\KK}(x))$. 
\item[(2)] \manuallabel{item2:SolutionToMainProblem}{(2)}For 
$B(n)=\smashoperator{\sum_{\mathclap{\substack{{\bs{v}=(\Lst{\nu}{0}{s})\in \tilde{S}}{}}}}^{{}}}{b_{\bs{v}}(n)\,(\zeta^n)^{\mu_0}Q_{1}(n)^{\mu_{1}}\cdots Q_{s}(n)^{\mu_{s}}}$
we have $A(n)=B(n)$ for all $n\geq\delta$.
\item[(3)]\manuallabel{item3:SolutionToMainProblem}{(3)} The subring 
$$\tau(\tilde{\KK}(x))[\langle\zeta^n\rangle_{n\geq0}][\langle Q_1(n)\rangle_{n\geq0},\langle Q_1(n)^{-1}\rangle_{n\geq0}]\dots[\langle Q_s(n)\rangle_{n\geq0},\langle Q_s(n)^{-1}\rangle_{n\geq0}]$$
of $\ringOfEquivSeqs[\tilde{\KK}]$ forms a Laurent polynomial ring extension of $\tau(\tilde{\KK}(x))[\langle\zeta^n\rangle_{n\geq0}]$. Thus the sequences produced by $Q_1(n),\dots,Q_s(n)$ are algebraically independent among each other over $\tau(\tilde{\KK}(x))[\langle\zeta^n\rangle_{n\geq0}]$. 
\item[(4)]\manuallabel{item4:SolutionToMainProblem}{(4)} We have that $A(n)=0$ for all $n\geq d$ for some $d\in\NN$ if and only if $a=0$ if and only $B(n)$ is the zero-expression. If this holds, $A(n)=0$ for all $n\geq\delta$.
\end{enumerate}
\end{corollary}
\begin{proof}
\begin{enumerate}[\hspace*{1em}(1)]
\item Note that $\dField{\EE}$ is a multiple chain \piE-extension of $\dField{\tilde{\KK}(x)[\vartheta]}$ over $\tilde{\KK}(x)$ equipped by an evaluation function given by the iterative application of Lemma~\ref{Lemma:ExtendEv}. Thus statement~\ref{item1:SolutionToMainProblem} follows.
\item By statement~\ref{item2:main} of Theorem~\ref{thm:main} it follows that $\tilde{\ev}(a,n)=A(n)$ holds for all $n\geq\delta$. By definition of the $\tilde{\ev}$ function we have $\tilde{ev}(a,n)=B(n)$. Thus statement~\ref{item2:SolutionToMainProblem} holds.
\item Since $\tau:\EE\to\ringOfEquivSeqs[\tilde{\KK}]$ with $\tau(f)=\langle\tilde{\ev}(f,n)\rangle_{n\geq0}$ is an injective difference ring homomorphism by statement~\ref{item2:main} of Theorem~\ref{thm:main}, statement~\ref{item3:SolutionToMainProblem} follows.
\item Since $\tau$ is injective, it follows that 
\begin{eqnarray*}
0=A(n)\text{ for all $n\geq d$ for some $d\in\NN$}&\stackrel{\text{item~\ref{item2:SolutionToMainProblem}}}{\Longleftrightarrow}& 0=B(n)=\tilde{\ev}(a,n)\text{ for all $n\geq \max(d,\delta)$}\\
&\Longleftrightarrow& \tau(a)={\bf 0}\\
&\stackrel{\tau\text{ injective}}{\Longleftrightarrow}& a=0\\ 
&\Longleftrightarrow& B(n) \text{ is the zero-expression.}
\end{eqnarray*}
If this is the case, $A(n)=B(n)=0$ for all $n\geq\delta$. This proves the last statement.
\end{enumerate}

\vspace*{-1cm}

\end{proof}

\medskip

\begin{example}[Cont. Examples~\ref{exa:nestedHyperGeoPrdtPreprocessingStep1},~\ref{exa:shiftCoPrimeReduced},~\ref{exa:singleChainAPExtensions},~\ref{exa:monomialDepth2OrderedMultipleChainPiExtForShiftCoPrimePolys},~\ref{exa:monomialDepth2OrderedMultipleChainPiExtForAlgNums}]\label{exa:simpleRPiExt4HypergeometricPrdts4Refined}
    \normalfont Let $\dField{\KK(x)}$ be the rational {\df} with $\KK=\QQ(\sqrt{3})$ equipped with the field automorphism $\s:\KK(x)\to\KK(x)$ and the evaluation function $\ev:\KK(x)\times\NN\to\KK$ defined by $\s(x)=x+1$ and~\eqref{map:evaForRatFxns} respectively. Given the nesting depth $2$ hypergeometric product $P(n)$ with~\eqref{eqn:depth2HyperGeoPrdt} in Example~\ref{exa:nestedHyperGeoPrdtPreprocessingStep1}, we computed the following $3$-refined shift-coprime product representation form:
    \begin{align}\label{eqn:4refinedHyperGeoPrdt}
    P(n) = \tilde{c}\,r(n)\,\tilde{G}(n)\,\tilde{H}(n)
    \end{align}
    where $\tilde{c}$, $r(n)$, $\tilde{G}(n)$ and $\tilde{H}(n)$ are given in \eqref{eqn:refinedConstantTerm},~\eqref{eqn:rationalFunctionTerm},~\eqref{eqn:geometricProducts}, and~\eqref{eqn:hypergeometricProductInShifCoPrimeRepForm} respectively. In particular,~\eqref{eqn:4refinedHyperGeoPrdt} holds for all $n\in\NN$ with $n\ge\delta-1=2$.\\   
	In Example~\ref{exa:monomialDepth2OrderedMultipleChainPiExtForShiftCoPrimePolys} we constructed the ordered multiple chain \piE-extension $\dField{\tilde{\HH}}$ of $\dField{\KK(x)}$ with 
    $\tilde{\HH}=\KK(x)\genn{z_{1,1}}\genn{z_{2,1}}\genn{z_{1,2}}$
    which is composed by the single chain \piE-extensions of $\dField{\KK(x)}$ defined in items~\ref{item:monomialDepth2SingleChainPiExtBasedAtXMinus2} and~\ref{item:monomialDepth1SingleChainPiExtBasedAtXPlus1Over24}. The automorphism and the  evaluation function were defined as given in~\eqref{diffAutoEval:monomialDepth2SingleChainPiExtBasedAtXMinus2} and~\eqref{diffAutoEval:monomialDepth1SingleChainPiExtBasedAtXPlus1Over24}. In particular, the hypergeometric product expression $\tilde{H}(n)$ is modelled by the expression 
    $\tilde{h}={z_{1,1}}^{3}\,{z_{2,1}}\,{z_{1,2}}$
    where $\tilde{h}=h$ is taken from~\eqref{eqn:hyperGeometricProductRepresentationInRingA}. \\    
    Furthermore, we constructed the \rpiE-extension $\dField{\DD}$ of $\dField{\tilde{\KK}}$ with $\tilde{\KK}=\QQ(\sqrt{3},\ii)$ equipped with the evaluation function $\tilde{\ev}$ from Example~\ref{exa:singleRExtAndOrderedMultipleChainPiExtOverConstField}. There $\tilde{G}(n)$ is modelled by~\eqref{Exp:gTilde}.\\
    Merging the two {\dr s} $\dField{\tilde{\HH}}$ and $\dField{\DD}$, we get the \rpiE-extension $\dField{\EE}$ of $\dField{\tilde{\KK}(x)}$ with 
    $\EE=\tilde{\KK}(x)[\vartheta]\genn{\tilde{y}_{1,1}}\genn{\tilde{y}_{2,1}}\genn{\tilde{y}_{3,1}}\genn{z_{1,1}}\genn{z_{2,1}}\genn{\tilde{y}_{2,2}}\genn{\tilde{y}_{3,2}}\genn{z_{1,2}}$ where
    the automorphism $\s:\EE\to\EE$ and the evaluation function $\tilde{\ev}:\EE\times\NN\to\tilde{\KK}$ are defined by~\eqref{diffAuto:order4SingleRMonomial},~\eqref{evMap:order4SingleRMonomial}, \eqref{diffAutoEval:nestingDepth2PiExtOverAlgNums},~\eqref{diffAutoEval:monomialDepth2SingleChainPiExtBasedAtXMinus2}, and~\eqref{diffAutoEval:monomialDepth1SingleChainPiExtBasedAtXPlus1Over24}. Following the proof of Theorem~\ref{thm:main}
    we set
     $p := \tilde{c}\,r\,\tilde{g}\,\tilde{h}$ and get for all $n\geq2$ the simplification 
        \begin{align*}
        P(n) = \tilde{\ev}(p,n) &= -\dfrac{254}{432}\,\tilde{\ev}(r,\,n)\,\tilde{\ev}(\tilde{g},\,n)\,\tilde{\ev}(\tilde{h},\,n) \\
             &=-\dfrac{254}{432}\,(n-1)^{3}\,n\,(n+1)\,(n+2)\ \frac{1}{2}\frac{(1-\ii)\,(\ii)^{n}\,(\ii\,(\ii^{n})^{2}+1)\,\algSeq{3}{n}\,\intSeqExp{5}{n}{2}\,\intSeq{2}{n+1 \choose 2}}{\intSeq{2}{n}\,\intSeq{5}{n+1 \choose 2}} \\[3pt]
             \begin{split}
                & \quad  \left(\myProduct{k}{3}{n}{\left(k-2\right)}\right)^{3}\,\left(\myProduct{k}{3}{n}{\left(k+\tfrac{1}{24}\right)}\right)\,\left(\myProduct{k}{3}{n}{\myProduct{j}{3}{k}{\left(j-2\right)}}\right).
             \end{split}             
        \end{align*}
    Based on this representation in an \rpiE-extension, we can extract the following extra property. Since $\tau:\EE\to\ringOfEquivSeqs$ is a $\tilde{\KK}$-embedding, it follows that the sub-difference ring $(R,S)$ of $(\ringOfEquivSeqs[\tilde{\KK}],S)$ with 
    $$R=\tau(\tilde{\KK}(x))[\langle\ii^n\rangle_{n\geq0}][\langle Q_1(n)\rangle_{n\geq0},\langle Q_1(n)^{-1}\rangle_{n\geq0}]\dots[\langle Q_8(n)\rangle_{n\geq0},\langle Q_8(n)^{-1}\rangle_{n\geq0}]$$
and
\begin{align*}
    Q_{1}(n) &= \algSeq{3}{n}, & Q_{2}(n) &= 2^{n}, & Q_{3}(n) &= 5^{n}, & Q_{4}(n) &= \myProduct{k}{3}{n}{(k-2)}, \\
    Q_{5}(n) &= \myProduct{k}{3}{n}{\left(k+\tfrac{1}{24}\right)}, & Q_{6}(n) &= \intSeq{2}{{n+1 \choose 2}}, & Q_{7}(n) &= \intSeq{5}{{n+1 \choose 2}}, & Q_{8}(n) &= \myProduct{k}{3}{n}{\myProduct{j}{3}{k}{\left(j-2\right)}}
\end{align*}
is isomorphic to $\dField{\EE}$. In particular, we can conclude that $R$ is a Laurent polynomial ring extension of the ring $G=\tau(\tilde{\KK}(x))[\langle\ii^n\rangle_{n\geq0}]$. In a nutshell, 
    the sequences generated by the products $Q_{1}(n),\dots,Q_{8}(n)$ are algebraically independent among each other over the ring $G$. 
\end{example}

\subsection{The Mathematica package --- {\rm \MText{NestedProducts}}}\label{subsec:mathematicaDemo}
In the following we will demonstrate how our tools can be activated with the help of the Mathematica package \MText{NestedProduct}.
We start with the nested hypergeometric product expression
\begin{align}\label{eqn:ratePackageNestedDepth2HypergeoPrdt}
A(n)=\frac{1}{2}\,\myProduct{k}{1}{n-1}{\frac{1}{36}\left(\myProduct{i}{1}{k-1}{\frac{(i+1)\,(i+2)}{4\,(2\,i+3)^{2}}}\right)}\in\ProdExpr(\QQ(x))
\end{align}
from~\cite[Example 3]{kauers2018guess} which was guessed using the Mathematica package \href{https://mat.univie.ac.at/~kratt/rate/rate.html}{RATE} written by Christian Krattenthaler; see~\cite{krattenthaler1997rate}. 
\begin{mdframed}[style=mmawithcounterstyle]
	After loading the package
	\begin{mma}\manuallabel{mma:MathematicaSession1}{Mathematica Session 1}
		\In{
			\MText{<<NestedProducts.m}
		}
		\\
		\Print{\LoadP{\rc{\bf NestedProducts} --- A package by Evans Doe Ocansey --- $\copyright$ RISC}}\\
	\end{mma}
	\noindent into Mathematica, we define the product with the command
	\begin{mma}
		\In{
			{ A =\frac{1}{2}FormalProduct\bigg[\frac{1}{36}\,\Big(FormalProduct\Big[\frac{(i+1)\,(i+2)}{4\,(2\,i+3)^{2}},\,\{i,\, 1,\,k-1\}]\Big),\, \{k,\,1,\,n-1\}\bigg];}
		}\\
	\end{mma}
	\noindent Here \texttt{FormalProduct}[f,\{k,a,n+b\}] (as shortcut one can use \texttt{FProduct}) defines a nested product $\prod_{k=a}^{n+b}f$ where $a,b$ are integers and the multiplicand $f$, free of $n$, must be an expression in terms of nested products whose outermost upper bounds are given by $k$ or an integer shift of $k$. Then applying the command \texttt{ProductReduce} to \texttt{A} we solve \ref{prob:ProblemRPE} and get the result
	\begin{mma}
		\In\MLabel{MMA:OutputofP}{
			B = ProductReduce[A]}\vspace*{0.15cm}\\
		\Out{\manuallabel{mma:session1Output3}{Out[3]}
			\frac{9\,\intSeqExp{2}{n}{5}\left(\myProduct{\MText{k}}{1}{\MText{n}}{\big(k+\tfrac{3}{2}\big)}\right)^{4}\left(\myProduct{\MText{k}}{1}{\MText{n}}{\myProduct{\MText{i}}{1}{\MText{k}}{\left(i+1\right)}}\right)^{2}}{\intSeq{(2\,n+3)}{2}\intSeqExp{3}{n}{2}\,\left(2^{\tbinom{n+1}{2}}\right)^{4}\,\left(\myProduct{\MText{k}}{1}{\MText{n}}{\left(k+1\right)}\right)^{3}\left(\myProduct{\MText{k}}{1}{\MText{n}}{\myProduct{\MText{i}}{1}{\MText{k}}{\big(i+\tfrac{3}{2}\big)}}\right)^{2}}
		}
		\\
	\end{mma}
\end{mdframed}
Internally, the package synchronizes in~\eqref{eqn:ratePackageNestedDepth2HypergeoPrdt} the upper bounds to $n$ and $k$, respectively. This yields 
\begin{align}\label{eqn:nestedDepth2HypergeoPrdtsWithSymbVarBounds}
\frac{9\,(n+1)\,(n+2)}{2\,(2\,n+3)}\left(\myProduct{j}{1}{n}{\frac{4\,(2\,j+3)^{2}}{(j+1)\,(j+2)}}\right)\,\left(\myProduct{k}{1}{n}{\frac{(2\,k+3)^{2}}{9\, (k+1)\,(k+2)}\,\left(\myProduct{j}{1}{k}{\frac{(j+1)\,(j+2)}{4\,(2\,j+3)^{2}}} \right)} \right).
\end{align}
Then~\eqref{eqn:nestedDepth2HypergeoPrdtsWithSymbVarBounds} is reduced further to~\ref{mma:session1Output3} in terms of the algebraically independent products
\begin{align}\label{Equ:ProdExpKauers}
\intSeq{2}{n},&&\intSeq{3}{n},&&\myProduct{k}{1}{n}{\big(k+\tfrac{3}{2}\big)},&&\myProduct{k}{1}{n}{(k+1)},&&\intSeq{2}{{n+1\choose 2}}=\myProduct{k}{1}{n}{\myProduct{i}{1}{k}{2}},&&\myProduct{k}{1}{n}{\myProduct{i}{1}{k}{\big(i+\tfrac{3}{2}\big)}},&&\myProduct{k}{1}{n}{\myProduct{i}{1}{k}{\left(i+1\right)}}.
\end{align}
\noindent    Note that one could have represented the product in~\eqref{eqn:ratePackageNestedDepth2HypergeoPrdt} directly within a \piE-extension by simply taking the inner product as the first \piE-monomial and the outermost product as the second \piE-monomial. However, for more complicated expression such a representation can be rather challenging. 

\begin{mdframed}[style=mmawithcounterstyle,nobreak=true]
	The full capability of our machinery can be illustrated by combining, e.g., the expression~\eqref{eqn:ratePackageNestedDepth2HypergeoPrdt} with the following related product (where one of the inner products is slightly modified):
	\begin{mma}
		\In{
			{ A_2 =A+
				FProduct\Big[\frac{4 (3 + 2 k)^4}{(k + 1)^2 (2 k + 1)^4 (k + 2)^2}
				FProduct[\frac{-(i + 1) (i + 2)}{4 (2 i - 1)^2}, \{i, 1, k\}], \{k, 1, n\}\Big];}
		}\\
	\end{mma}
	\noindent Then it is not immediate how this expression can be represented in an \rpiE-extension. But applying our toolbox this task can be automatically accomplished:
	\begin{mma}
		\In B_2=ProductReduce[A_2]\\
		\Out \dfrac{\left(81\,\left(n^{2}+3\,n+2\right)+\left(1+\ii\right)\,\left(2\,n+3\right)^{4}\left(\ii\right)^{n}+\left(1-\ii\right)\,\left(2\,n+3\right)^{4}\,\left((\ii)^{n}\right)^{3}\right)\,2^{n}\left(\myProduct{\MText{k}}{1}{\MText{n}}{\myProduct{\MText{i}}{1}{\MText{k}}{(i+1)}}\right)^{2}}{81\,(n+2)\,\left(2^{\tbinom{n+1}{2}}\right)^{4} \left(\myProduct{\MText{k}}{1}{\MText{n}}{\left(k+1\right)}\right)^{3}\left(\myProduct{\MText{k}}{1}{\MText{n}}{\myProduct{\MText{i}}{1}{\MText{k}}{\left(i-\tfrac{1}{2}\right)}}\right)^{2}}\\
	\end{mma}
\end{mdframed}

\noindent By solving \ref{prob:ProblemRPE} the expression can be rephrased in an \rpiE-extension with the \rE-monomial $\ii^n$ and the \piE-monomials given in~\eqref{Equ:ProdExpKauers}. In short, together with $\ii^n$ the expression can be reduced again in terms of the algebraic independent products~\eqref{Equ:ProdExpKauers}.

\medskip

Similar expressions as given in~\eqref{eqn:ratePackageNestedDepth2HypergeoPrdt} arise during challenging evaluations of determinants; see, e.g., \cite{mills1983alternating,zeilberger1996proof,krattenthaler2001advanced}. We expect that the new tools elaborated in this article will prove beneficial in related but more complicated product expressions. 

\medskip

\normalfont  We conclude this section by combining our tools from above with the summation package \MText{Sigma}. 

\vspace*{-1em}

\begin{mdframed}[style=mmawithcounterstyle]
	After loading in
	\begin{mma}
		\In \MText{<<Sigma.m} \\
		\Print \LoadP{Sigma - A summation package by Carsten Schneider
			\copyright\ RISC-Linz}\\
		\notag
	\end{mma}
	\noindent we insert the sum   
	\begin{mma}
		\In\MLabel{MMA:Sum} mySum=SigmaSum[
		\Big(
		-1
		+(1+k) (2+k)^2 \prod_{j=1}^k (1+j)^2
		\Big) \prod_{j=1}^k j \prod_{i=1}^j (1+i)^2 \vspace*{-0.1cm}
		\newline
		-\frac{4}{3} \Big(
		1
		+2 (1+k)^2 (3+k) \prod_{j=1}^k -j (2+j)
		\Big) \prod_{j=1}^k 2 j \prod_{i=1}^j -i (2+i)
		,\{k,1,n\}];\\ 
	\end{mma}
	\noindent Afterwards we can activate the available summation algorithms of \texttt{Sigma} with the function call \texttt{SigmaReduce} and succeed in eliminating the summation sign:
	\begin{mma}
		\In SigmaReduce[mySum]\\
		\Out 4-\frac{1}{3}\,(1+n)^{5}\,(2+n)^{2}\,\Big(-3+(1+i)\,\big(-\ii+(\ii^{n})^{2}\big)\,(3+n)\,\ii^{n}\Big)\,(n!)^{5}\,\left(\prod_{i=1}^{k}\prod_{j=1}^{i} j\right)^{2}\\
	\end{mma}
\end{mdframed}

\noindent In other words, we have derived the simplification
\begin{multline*}
\sum_{k=1}^n\Bigg(\Big(-1+(1+k)\,(2+k)^{2}\,\myProduct{j}{1}{k}{(1+j)^{2}}\Big)\,\myProduct{j}{1}{k}{j}\,\myProduct{i}{1}{j}{(1+i)^{2}}\\[-0.1cm]
-\frac{4}{3} \Big(1+2\,(1+k)^{2}\,(3+k)\,\myProduct{j}{1}{k}{-j\,(2+j)}\Big)\,\myProduct{j}{1}{k}{2\,j \myProduct{i}{1}{j}{-i\,(2+i)}}\Bigg)
\\[-0.3cm]
=4-\frac{1}{3}\,(1+n)^{5}\,(2+n)^{2}\,\Big(-3+(1+i)\,\big(-\ii+(\ii^{n})^{2}\big)\,(3+n)\,\ii^{n}\Big)\,(n!)^{5}\,\left(\myProduct{i}{1}{k}{\myProduct{j}{1}{i} j}\right)^{2}.
\end{multline*}
We emphasize that the summand given in~\myIn{\ref{MMA:Sum}} has been transformed internally with the package \texttt{NestedProducts} to the form
\begin{multline*}
\frac{1}{3}\,(1+k)^{2}\,\left(\myProduct{i}{1}{k}{i}\right)^{3} \left(\myProduct{i}{1}{k}{\myProduct{j}{1}{i}{j}}\right)^{2}\Bigg(-3+(-1-\ii)\,(2+k)\,\ii^{k}+(-1+\ii)\,(2+k)\,(\ii^{k})^{3}\\[-0.3cm]
+(1+k)^{3}\,(2+k)^{2}\,\Big(3+(-1+\ii)\,(3+k)\,\ii^{k}+(-1-\ii)\,(3+k)\,(\ii^{k})^{3}\Big)\left(\myProduct{i}{1}{k}{i}\right)^{2}\Bigg).
\end{multline*} 
Then \texttt{Sigma} reads off the derived products and rephrases them directly to a tower of \rpiE-extensions (without the exploitation of the available tools in \texttt{Sigma} that can check whether the constant field remains unchanged). Afterwards the underlying summation algorithms of \texttt{Sigma} are applied to derive the final result.

\section{Conclusion}\label{Sec:Conclusion}

We enhanced non-trivially the ideas from~\cite{schneider2005product,Schneider:14,ocansey2018representing} (related also to~\cite{abramov2010polynomial,chen2011structure}) in order to solve \ref{prob:ProblemRPE} in Theorem~\ref{thm:main} and Corollary~\ref{Cor:SolutionToMainProblem} above. There we cannot only reduce or simplify expressions in terms of hypergeometric products of nesting depth $1$ but in terms of hypergeometric products of nesting depth $\geq1$. More precisely, the expression can be reduced to an expression in terms of one root of unity product of the form $\zeta^n$ and hypergeometric products $Q_1(n),\dots,Q_s(n)\in\Prod(\KK(x))$ of arbitrary but finite nesting depth which are algebraically independent among each other. This latter property has been extracted from results elaborated in~\cite{schneider2017summation} (which are inspired by~\cite{put1997galois}). Combined with the existing difference ring algorithms for symbolic summation~\cite{karr1981summation,schneider2001symbolic,schneider2016difference} this yields a complete summation machinery to reduce and simplify nested sums over hypergeometric products of arbitrary but finite nesting depth. 

A natural future task is to enhance this combined toolbox of the packages \MText{NestedPro\-ducts} and \MText{Sigma} further and to tackle, e.g., definite summation problems. In particular, the interaction with the available creative telescoping algorithms~\cite{schneider2007simplifying,schneider2008refined,StructuralTheorems:10,schneider2015fast} and recurrence solving algorithms~\cite{ABPS:20} should be explored further.

Following the ideas from~\cite{ocansey2018representing} one might extend the above machinery to the class of nested $q$-hypergeometric products covering also the multibasic and mixed case~\cite{bauer1999multibasic}. Another open task is to combine the above ideas with contributions from~\cite{schneider2020minimal} (based on Smith normal form calculations) to find optimal representations of such nested products. This means that in the output expression the order $\lambda$ of the primitive root of unity $\zeta$ in $\zeta^n$ and the number $s$ of algebraically independent products $Q_1(n),\dots,Q_s(n)$ should be minimized.

Finally, it would be interesting to see if the class of hypergeometric products of finite nesting depth can be generalized further to products of the form~\eqref{eqn:nestedHypergeometricProduct} where in the multiplicands the products do not appear only in form of Laurent polynomial expressions.



\providecommand{\bysame}{\leavevmode\hbox to3em{\hrulefill}\thinspace}
\providecommand{\MR}{\relax\ifhmode\unskip\space\fi MR }
\providecommand{\MRhref}[2]{%
	\href{http://www.ams.org/mathscinet-getitem?mr=#1}{#2}
}
\providecommand{\href}[2]{#2}

\end{document}